\documentclass[aps,prx,showpacs,twocolumn,reprint, nofootinbib,superscriptaddress]{revtex4-2}

\usepackage[dvips]{graphicx}
\usepackage{amsmath,amssymb,amsthm,mathrsfs,amsfonts,dsfont}
\usepackage{hyperref}
\usepackage{bm}
\usepackage{enumerate}
\usepackage{color}
\usepackage{graphicx}
\usepackage[capitalize]{cleveref}

\usepackage[utf8]{inputenc}
\usepackage{mathtools}
\usepackage{graphicx}
\usepackage{tikz}
\usetikzlibrary{arrows}
\usepackage{qcircuit}

\usepackage{algorithm}
\usepackage{algpseudocode}

\newtheorem{remark}{Remark}
\newtheorem{definition}{Definition}
\newtheorem{statement}{Statement}
\newtheorem{theorem}{Theorem}
\newtheorem{corollary}{Corollary}
\newtheorem{lemma}{Lemma}

\newtheorem{property}{Property}

\usepackage{xcolor}

\newcommand{\fillr}{\mathcal{F}}

\newcommand{\lnorm}[1]{\left\lVert #1 \right\rVert_{2}}

\newcommand{\step}{\Delta_N}

\renewcommand{\H}{\mathcal{H}}
\newcommand{\ket}[1]{|#1\rangle}
\newcommand{\bra}[1]{\langle #1 |}
\newcommand{\op}[2]{|#1\rangle \langle #2|}
\newcommand{\ip}[2]{\langle #1 | #2 \rangle}

\DeclarePairedDelimiter\ceil{\lceil}{\rceil}

\newcommand{\vertiii}[1]{{\left\vert\kern-0.25ex\left\vert\kern-0.25ex\left\vert #1 \right\vert\kern-0.25ex\right\vert\kern-0.25ex\right\vert}}

\begin{document}

\title{Preparing Arbitrary Continuous Functions in Quantum Registers With Logarithmic Complexity}

\newcommand{\equaltext}{Both authors contributed equally}
\author{Arthur G. Rattew}
\email{arthur.rattew@materials.ox.ac.uk}
\affiliation{Department of Materials, University of Oxford, Parks Road, Oxford OX1 3PH, United Kingdom}
\affiliation{JPMorgan Chase Bank N.A., Future Lab for Applied Research and Engineering}

\author{B\'alint Koczor}
\email{balint.koczor@materials.ox.ac.uk}
\thanks{\\\equaltext}
\affiliation{Department of Materials, University of Oxford, Parks Road, Oxford OX1 3PH, United Kingdom}

\begin{abstract} 
Quantum computers will be able solve important problems with significant polynomial and exponential speedups over their classical counterparts, for instance in option pricing in finance, and in real-space molecular chemistry simulations. However, key applications can only achieve their potential speedup if their inputs are prepared efficiently. We effectively solve the important problem of efficiently preparing quantum states following arbitrary continuous (as well as more general) functions with complexity logarithmic in the desired resolution, and with rigorous error bounds. This is enabled by the development of a fundamental subroutine based off of the simulation of rank-1 projectors. Combined with diverse techniques from quantum information processing, this subroutine enables us to present a broad set of tools for solving practical tasks, such as state preparation, numerical integration of Lipschitz continuous functions, and superior sampling from probability density functions. As a result, our work has significant implications in a wide range of applications, for instance in financial forecasting, and in quantum simulation.
\end{abstract}

\maketitle

\section{Introduction}\label{sec:background}
With rapidly accelerating progress in hardware development, quantum
computers are quickly becoming a realistic technology. New records are constantly being
set, demonstrating that these machines are capable of significantly outperforming the best classical computers in certain settings~\cite{aruteQuantumSupremacyUsing2019,zhongPhaseProgrammableGaussianBoson2021,wuStrongQuantumComputational2021,ebadiQuantumPhasesMatter2021,gongQuantumWalksProgrammable2021a}.
As the technology is scaled up, it will ultimately enable powerful quantum algorithms 
capable of both polynomial and exponential speedups over their classical counterparts.
However, many of these algorithms require efficient state preparation procedures to achieve their speedup, as is the case in: option pricing~\cite{chakrabarti2021threshold, stamatopoulos2020option, herman2022survey}, machine learning~\cite{lloyd2014quantum, mitarai2018quantum, pistoia2021quantum}, matrix inversion~\cite{harrow2009quantum}, quantum chemistry~\cite{chan2022grid, ward2009preparation}, and quantum Monte Carlo based algorithms~\cite{woerner2019quantum, montanaro2015quantum}. As such, it remains a crucial open question: can we \emph{efficiently} prepare quantum states given a set of prespecified amplitudes?

Unfortunately, a quantum circuit capable of producing an arbitrary quantum state necessitates exponential complexity in general~\cite{plesch2011quantum},  although recent works have shown that the exponential circuit depth cost can actually be made linear in exchange for utilizing an exponential number of ancillary qubits~\cite{sun2021asymptotically,rosenthal2021query,zhang2022quantum}.
As a result, in practice, it is essential to exploit specific properties of the state being produced.

In this work, we effectively solve the important problem of efficiently preparing continuous (as well as some more general) functions in quantum registers, and we do so with rigorous theoretical guarantees.
While we illustrate the core subroutine underpinning our novel technique in \cref{fig:fig1}, our main result is stated as \cref{theorem:total_query_complexity_informal}. 
Succinctly and informally stated, we prove that our approach has a constant query complexity (asymptotically in the number of qubits), and obtains a bounded error which can be made arbitrarily small. A remarkable feature of our approach is that the performance of our algorithm becomes asymptotically independent of the qubit count as we increase the number of qubits, and that we can reach this asymptotic independence exponentially quickly in the number of qubits.

The constant cost of our algorithm depends inverse-polynomially on the \textit{filling ratio} of the function -- a quantity that we define in a subsequent section and illustrate in \cref{fig:fig1a} as the ratio of blue area to the area of the total bounding box. While it is possible to construct pathological functions for which the filling ratio is arbitrarily small, for realistic functions encountered in practice this quantity is typically very reasonable (e.g. usually around $10^{-1}$ to $10^{-3})$. We additionally suggest ways that this constant cost can be made tractable even in pathological cases. Moreover, if future works improve the error-dependence of our simulation subroutine, the impact of the filling ratio will be even further negated.

The implications of our work go significantly beyond just enabling the speedup of existing algorithms via efficient state preparation. Functions have been used for centuries to model physical, chemical, biological, and financial phenomena, and have played crucial roles in the industrial revolution and in modern engineering. As the use of analytical techniques for functions modelling realistic systems rapidly becomes intractable, numerical techniques often become the only available option~\cite{hildebrand1987introduction}. As a result, computational tools accelerating such numerical analyses are of broad and vital importance. To this end, the techniques we develop in this paper not only allow for the preparation of broad classes of functions in quantum registers, but also enable a variety of tools such as for numerical integration of Lipschitz continuous functions, for the optimization of continuous functions, and even for improved sampling from probability density functions.

We begin by reviewing related work.
In the next section we briefly summarise our problem statement and our main results. There, we informally state our rigorous error bounds which confirm that our approach is efficient, noting that we defer all technical details to the appendix.
We then summarize the implementation of a key subroutine in our approach: the efficient time-evolution of rank-one dense matrices corresponding to projectors onto discretized continuous functions. We also discuss the limitations of our approach in a subsequent section.
We then demonstrate the main algorithm in relevant applications, and support the practicality of the approach with large-scale numerical simulations confirming theoretical expectations. We finally discuss further applications and then conclude.

\begin{figure}[tb]
	\begin{centering}
		\includegraphics[width=0.48\textwidth]{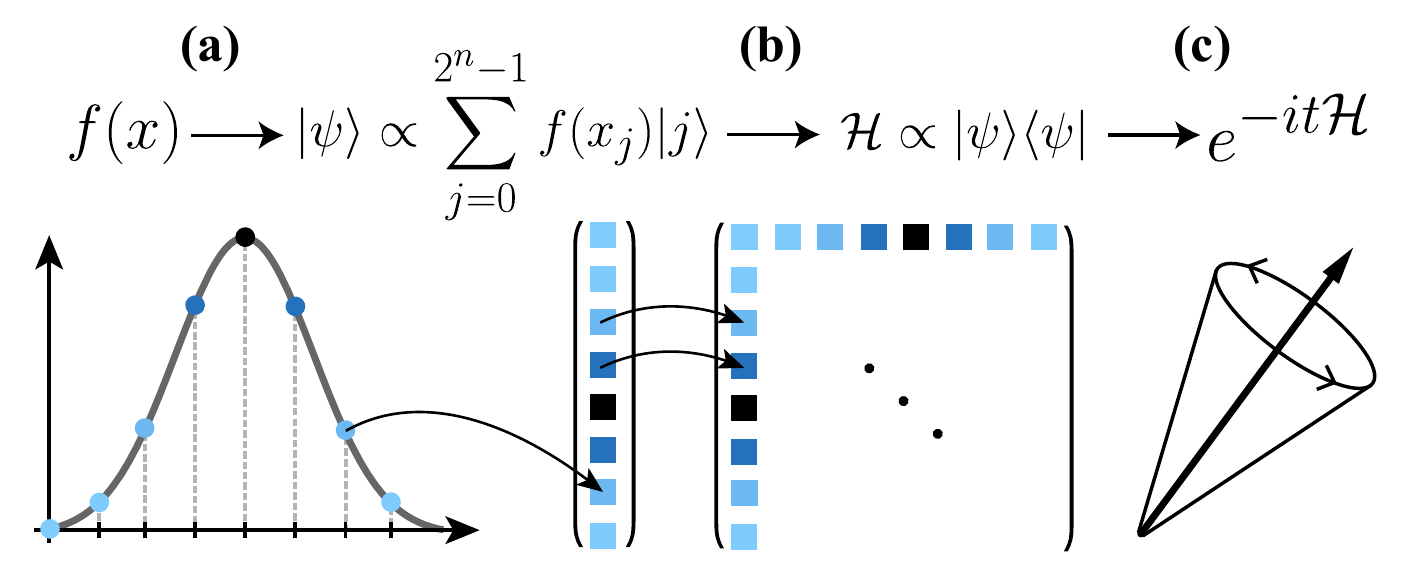}
		\caption{
		Key rank-1 Hamiltonian simulation subroutine.
			\textbf{(a)}
			 Samples (coloured circles) of the continuous function $f(x_k)$ (grey smooth curve) are proportional to the amplitudes (coloured squares)
			 of the quantum state $\ket{\psi}$ as per \cref{eq:psidef}.
			\textbf{(b)} We embed this vector into
			a rank-1 matrix $\mathcal{H} \propto \ket{\psi} \bra{\psi}$
			whose matrix entries are given as
			$f(x_k) f^*(x_l)$ (colours are only illustrative, not exact) and are efficiently computed by a quantum oracle, e.g., via arithmetic operations. 
			\textbf{(c)} We then efficiently simulate the
			time evolution	$e^{- i t \mathcal{H} }$
			via low-rank~\cite{rebentrost2018quantum} and	1-sparse~\cite{childs2004quantum,aharonov2003adiabatic,childs2003exponential,berry2007efficient}
			 simulation techniques.
			 This powerful time evolution resource allows us to prepare the ground state of $\mathcal{H}$ as our desired
			$\ket{\psi}$ using, e.g., adiabatic evolution, phase estimation
			and Hadamard tests. Note that each stage occurs coherently in superposition.
			\label{fig:fig1}
		}
	\end{centering}
\end{figure}

\subsection{Prior Work}

As early as in 2002 Grover and Rudolph presented a procedure for preparing an efficiently integrable probability density (i.e., non-negative, $L^1$-normalised) function~\cite{grover2002creating}, which we briefly summarise in \cref{si:grover_rudolph}.
Under certain assumptions on the efficient integrability of the function this approach theoretically has a complexity
polylogarithmic in the desired resolution, but exponentially scaling variants have also been suggested~\cite{marin2021quantum, holmes2020efficient, vazquez2021efficient}.
Nevertheless, in certain applications any variant of the procedure negates any potential quantum speedup due to its use of quantum integration~\cite{herbert2021no, chakrabarti2021threshold}. For example, quantum Monte Carlo based algorithms lose their speedup as pointed out by Herbert~\cite{herbert2021no} and similarly in derivative pricing as observed in the work of Chakrabarti \textit{et al.}~\cite{chakrabarti2021threshold}.
In contrast, the set of techniques presented in this work strictly generalize the Grover-Rudolph-type approaches, as they are not limited to efficiently integrable non-negative probability density functions -- they can be applied to any complex valued continuous function and beyond. Moreover, our techniques are not reliant on quantum integration, and thus do not suffer from the limitations of the Grover-Rudolph based approaches just mentioned.

We also note that Holmes \textit{et al.} recently presented an approximate procedure for preparing smooth analytic functions~\cite{holmes2020efficient} via a piecewise polynomial approximation. 
However, it is not straightforward to assess the approximation error of the approach, its overall complexity has not been established, and it is asymptotically limited in its use of the piecewise polynomial approximation.

\section{Algorithm Overview}
While our approach is applicable to the general class of efficiently computable mappings
$F^{(n)}: \{0,1\}^n \mapsto \mathbb{C}$ via \cref{sec:generalisation}, in the following we restrict ourselves
to the practically most important case of discretised continuous functions.
In particular, given an arbitrary continuous function $f : [a,b] \mapsto \mathbb{C}$ on the closed interval $[a, b]$, we wish to prepare the $n$-qubit normalized quantum state
\begin{align}\label{eq:psidef}
    \ket{\psi} =
    \frac{1}{\mathcal{N}}\sum_{j = 0}^{N - 1}f(x_j)\ket{j},
\end{align}
where $\mathcal{N}$ is the normalization factor, $N=2^n$, $\ket{j}$ is a standard basis vector with $j \in \{0, 1, ..., N - 1\}$, and $x_j$ is the $j^{th}$ grid point in a uniform discretization of the interval $[a,b]$---while we later remark that non-uniform grid spacings may sometimes be advantageous. 
In general, preparing an arbitrary state in a Hilbert space of $n$ qubits requires exponential complexity as the dimension of the space grows exponentially in $n$~\cite{plesch2011quantum, zhang2022quantum}. However, the algorithm we are about to present demonstrates that the very general set of states of the form shown in \cref{eq:psidef} (i.e. those corresponding to discretized continuous functions) can be prepared efficiently.

In the Appendix, we present our results for the analogous, but special, case of efficiently integrable probability density functions as in the case of the Grover-Rudolph~\cite{grover2002creating} approach. However, for brevity, here we focus only on the integration-free case where the state is prepared via samples of $f$ rather than through its integration. Moreover, we note that the point-wise state preparation approach is strictly more general than integration-based approaches.

Below we present a rank-1 simulation procedure which is both a fundamental subroutine and a primary contribution of this paper. In summary, we establish how the time evolution $e^{- i t \mathcal{H} }$ under the dense rank-one matrix $\mathcal{H} \propto \ket{\psi} \bra{\psi}$ can be efficiently implemented. Simulating this time evolution is a powerful resource and can be combined with diverse tools from quantum information processing to prepare the quantum states $\ket{\psi}$
efficiently. For instance, we consider phase estimation and the Hadamard test in the appendix, but this subroutine also enables other very useful algorithms, e.g. for efficiently integrating Lipschitz continuous functions.
Nevertheless, in the main text we focus on the conceptually most straightforward variant for preparing states of the form in \cref{eq:psidef}, based on adiabatic evolution.

We first formally define the filling ratio of the function being prepared as it is intrinsic to the asymptotic complexity of our state preparation procedures. We then informally state our main result for the adiabatic procedure, and then we describe the algorithm.

As stated in \cref{def:peakedness}, the filling ratio $\mathcal{F}$ of function $f$ is defined as $\fillr :=  \lVert f \rVert_1 / \lVert f \rVert_{max}$ and has a value $0 < \fillr \le 1$. Intuitively, it can be understood as the integral of the absolute function (i.e. $|f(x)|$) divided by the area of the box bounding the absolute function (i.e. the box with width $b - a$ and height $\max_{x\in[a, b]}|f(x)|$). The filling ratio is clearly pictured for some common functions in Figure~\ref{fig:fig1a}, as the ratio under the curves, to the total area of the bounding box. Finally, we note that we formally define the notion of efficient computability in \cref{subsec:efficient_time_evo} as well as in \cref{def:eff_comp}.

\begin{theorem}[informal version of \cref{theorem:total_query_complexity}]\label{theorem:total_query_complexity_informal}
	For an efficiently computable continuous function $f$, our approach prepares the $n$-qubit state $\ket{\psi} = \frac{1}{\mathcal{N}}\sum_{j=0}^{N-1}f(x_j)\ket{j}$ 
	with $N=2^n$ up to error $\epsilon$ with query complexity $O\left(
		\mathcal{F}^p / \epsilon^2
		\right)$,
	where $\mathcal{F}$ is the filling ratio and $p=-4$. The error $\epsilon$ is the deviation from an ideal unitary in terms of a spectral distance.
	The algorithm uses $O(n + d)$ ancillary qubits (where $d$ is the number of digits used in the discretization of $f$), and has a probability of failure bounded by $O\left(\epsilon^2\right)$.
\end{theorem}
It is worth additionally noting that the success probability only enters our bound to simplify our proofs -- in practice our approach effectively succeeds with probability $1$. This is discussed comprehensively in the Appendix. 

To enable the adiabatic state preparation procedure, we define a parameterized family of continuous functions $f_s := (1-s)f_0 + sf_1$, where $f_0\propto 1$ (i.e. corresponds to an easy-to-prepare state) and $f_1$ corresponds to the state we wish to prepare (i.e. $f_1 := f$) with $0\le s \le 1$. We then define the parameterized quantum state, $\ket{\psi_s} \propto \sum_{j=0}^{N-1}f_s(x_j)\ket{j}$. The procedure begins by creating the state $\ket{\psi_0}$, which by construction of $f_0$ is simply $\ket{+}^{\otimes n}$. We then slowly adiabatically morph this starting state to the final state $\ket{\psi_1}$ by evolving according to the rank-1 projector $\H(s) \propto \op{\psi_s}{\psi_s}$. Here, one must pick the total evolution time $T$ so as to satisfy the adiabatic theorem and ensure that the state remains in the ground state of the morphing Hamiltonian~\cite{van2001powerful}. As we prove in the \cref{sec:direct_state_prep}, $T$ is constant bounded (and can therefore be selected utilizing that bound).
Of course, we can't implement a continuous time evolution in practice, and so we discretize the temporal grid uniformly into $r$ time-steps. In the $j^{th}$ time-step, we evolve according to the time-independent Hamiltonian $\H(\frac{j}{r})$, yielding the overall evolution:
\begin{equation}\label{eqn:adiabatic_discretization_informal}
	e^{-i \tfrac{T}{r}\H(\frac{r}{r})}e^{-i \tfrac{T}{r}\H(\frac{r-1}{r})}\cdots e^{-i \tfrac{T}{r}\H(\frac{1}{r})}.
\end{equation} 
As shown in \cref{sec:direct_state_prep} this discretization has constant bounded error, which can be made arbitrarily small. Naturally, it is important to explain how each evolution $e^{-i \frac{T}{r}\H(\frac{j}{r})}$ is implemented, and we do so in \cref{subsec:efficient_time_evo}.

However, we first explain the conceptual difference between ``conventional'' adiabatic quantum computation, and the approach we have just described. In the standard formulation of adiabatic quantum computation, one defines the interpolated Hamiltonian $\H(s) = (1-s) \mathcal{H}_0 + s \mathcal{H}_1$, where $\H_0$ has an easy to prepare ground-state, and $\H_1$ encodes the problem of interest~\cite{farhi2000quantum}. However, in this definition the energy gap often vanishes (or becomes exponentially small), resulting in the total evolution time $T$ becoming intractable. By defining our interpolated Hamiltonian as the rank-1 projector $\mathcal{H}(s) \propto \ket{\psi_s} \bra{\psi_s}$, we circumvent this limitation. Indeed, in the appendix we prove that the spectral gap $g(s)$ is generally lower bounded by a constant, e.g., $g(s) \geq 1/2$. This ensures that our total evolution time required is independent of the problem size $n$.

\subsection{Efficient Implementation of the Time Evolution}\label{subsec:efficient_time_evo}

Before describing our time evolution procedure, we outline our oracle query model. Our procedure allows the preparation of a quantum state following an arbitrary continuous distribution, so long as the function is \textit{efficiently computable}. Precisely, by efficiently computable, we mean that given an input $x$ represented with $n$ bits of precision, there exists a \textit{classical} (and thus a quantum algorithm) returning $f(x)$ with $d$ bits of precision with complexity scaling as $\text{poly}(n, d)$. Equivalently, an efficiently computable function is any function with the oracle implementation $O_f$ efficiently performing the mapping $O_f\ket{x}\ket{0} = \ket{x}\ket{f(x)}$. See \cref{def:eff_comp}. Without a loss of generality, both the input and output of $f$ are given in a binary encoding. Intuitively, if there is an efficient classical program (e.g. in Python) evaluating $f(x)$,
then in principle, the function is efficiently computable. Indeed optimised, efficient
quantum-arithmetic procedures are already available~\cite{haner2018optimizing}. Classically, to evaluate the function $f$ with $N$ distinct inputs would require $O(N)$ queries to the function -- in the quantum setting, only $O(1)$ queries are required. This exponential advantage is intuitively suggestive of the potential power of the applications utilizing this procedure.

Furthermore, the time complexity of our approach is ultimately determined by the circuit depth. Given in most practical cases the function values can be computed via arithmetic operations, we expect that the oracle implementation requires less than a quadratic circuit depth (in terms of $n$ and the number of digits $d$ used in the discretization of $f$)~\cite{haner2018optimizing}.

\begin{figure}[tb]
	\begin{centering}
		\includegraphics[width=0.4\textwidth]{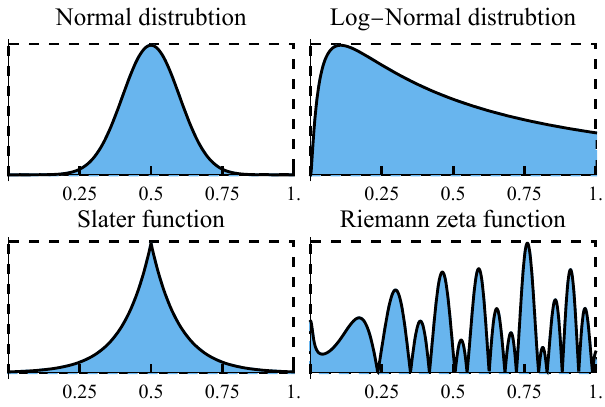}
		\caption{
			Applications of efficiently loading functions into quantum 
			registers include
			probability distributions (normal and log-normal distributions shown)
			Slater-type functions in  quantum chemistry \cite{mcardle2020quantum}
			as well as the Riemann zeta
			function~\cite{borwein2000efficient}.
			Our state-preparation error bounds depend on the filling ratio $\fillr$ as the absolute area under a function (light blue area)
			relative to its absolute maximum value (area in the dashed rectangle).		
			\label{fig:fig1a}
		}
	\end{centering}
\end{figure}

Given access to this oracle, we utilise the approach of Rebentrost \textit{et al.}~\cite{rebentrost2018quantum} to efficiently simulate quantum dynamics under the rank-$1$ matrix $A(s)$ whose entries are given by $[A(s)]_{kl}:= f_s(x_k) f^*_s(x_l)$. Note that the function values $f_s := (1-s)f_0 +s f_1$ can be computed by an oracle, but our definition also permits storing $f_0(x_j)$ and $f_1(x_j)$ in QRAM and computing linear combinations thereof using only addition and multiplication. Depending on the QRAM model, this could potentially make the cost of the oracle query $O(1)$, and thus the algorithm's circuit depth would be the same as its query complexity.

We proceed by encoding the matrix $A(s) \in \mathbb{C}^{N\times N}$ into a quadratically
larger one-sparse matrix $S_A \in \mathbb{C}^{N^2\times N^2}$ --  whereby a one-sparse matrix contains only a single non-zero entry in each row/column. This property allows us
to utilise powerful one-sparse simulation techniques~\cite{childs2004quantum,aharonov2003adiabatic,childs2003exponential,berry2007efficient} thereby \emph{exactly}
simulating the time evolution $e^{- i t S_A}$ which acts on the $n$-qubit main
register as well as on an $n$-qubit ancillary register -- and the simulation requires a further $d$-qubit, and two $1$-qubit ancilla registers. After measuring and discarding the ancilla register we implement the mapping
\begin{equation*}
	U''(\Delta t,s) |\psi_0\rangle   = e^{- i \Delta t \mathcal{H}(s) }  |\psi_0\rangle +  |\mathcal{E}\rangle
\end{equation*}
under the effective rank-1 Hamiltonian $\mathcal{H}(s) := A(s)/N$ up to a bounded error term $\lVert \mathcal{E} \rVert \in O(\Delta t^2)$. By setting the time step sufficiently small $\Delta t \ll 1$ we can simulate the
evolution $U'(\Delta t,s)  = e^{- i \Delta t \mathcal{H}(s) }$ up to an
arbitrarily small error. This allows us to implement the piecewise constant 
adiabatic evolution for overall time $T$ in $r$ iterations (timesteps) as per \cref{eqn:adiabatic_discretization_informal}.
This has the added benefit of also decreasing the adiabatic discretization error as we increase $r$.
In the Appendix, we prove that applying this unitary evolution to the initial state $|+\rangle^{\otimes n}$ indeed maps onto our desired state $\ket{\psi}$ with a bounded error that only depends on properties of the encoded function.

\subsection{Limitations}\label{subsec:error_analysis_and_limitations}
Our approach incurs two types of algorithmic error. First, we approximate an idealised continuous adiabatic evolution through a finite series of $r$ time-independent evolutions.
Second, our Hamiltonian simulation also incurs algorithmic error, as previously mentioned. However, both kind of errors can be suppressed arbitrarily by increasing the number of iterations $r$ and thus decreasing the time each time-independent unitary is evolved for (i.e. by increasing the resolution in the temporal discretization). 

While our approach is very general and is also applicable to non-continuous general functions (mappings),
its main practical limitation is that its complexity depends on
the filling ratio $\fillr$ of the function. In practice, we expect this filling ratio to be a \textit{modest constant}. See \cref{table:fillratios} for filling ratios of common probability distributions and other functions. Nevertheless, we can easily construct artificial, pathological instances where $\fillr$ can be arbitrarily small. 
As a result, while not essential, improving the algorithm's asymptotic dependence on $\mathcal{F}$ is a promising direction for future research.

\newcommand{\scform}[2]{#1{\times}10^{#2}}
\begin{table}[tb]
	\begin{tabular}{ @{\hspace{2mm}} c @{\hspace{3mm}}  c c @{\hspace{8mm}} c @{\hspace{3mm}} c c @{\hspace{2mm}}}
		\hline \\[-4mm]		
		\multicolumn{3}{l}{Normal distribution}& \multicolumn{3}{l}{Log-normal distribution}\\[-1mm]
		$\mu$ & $\sigma$	& filling ratio $\fillr$ 
		&	 $\mu$ & $\sigma$ & filling ratio $\fillr$\\[1mm] 
		\hline \\[-4mm]
		$1/2$ & $0.1$& $\scform{2.5}{-1}$
		& $0$ & $0.5$& $\scform{5.5}{-1}$ \\
		$1/2$ & $0.05$& $\scform{1.3}{-1}$
		& $0$ & $1.0$ & $\scform{7.6}{-1}$ \\
		$1/2$ & $0.01$& $\scform{2.5}{-2}$
		& $0$ & $1.5$ & $\scform{6.1}{-3}$
		\\[3mm] 
		\hline \\[-5mm]
		\multicolumn{3}{l}{Slater function}& \multicolumn{3}{l}{Riemann zeta function}\\[-1mm]
		-- & $\alpha$	& filling ratio $\fillr$ 
		&	 -- & $\alpha$ & filling ratio $\fillr$\\[1mm] 
		\hline \\[-4mm]
		& $5$ &  $\scform{3.7}{-1}$
		& & $10$ & $\scform{6.1}{-1}$
		\\
		& $10$ & $\scform{2.0}{-1}$
		& & $20$ & $\scform{5.0}{-1}$
		\\
		& $20$ & $\scform{1.0}{-1}$
		& & $100$ & $\scform{3.1}{-1}$\\[1mm]
		\hline
	\end{tabular}
	\caption{\label{table:fillratios}
		Filling ratios (as illustrated in \cref{fig:fig1a}) for common probability distributions
		as a function of their parameters.
		All data was computed for a domain $0 \leq x \leq 1$
		and the Slater function is $e^{-\alpha |x - 0.5|}$.
		The filling ratio of the Riemann zeta function along the critical line $\zeta(1/2+ i \alpha x)$
		is nearly independent from the scaling parameter $\alpha$ given its heavily oscillatory nature. 
	} 
\end{table}

\begin{figure*}[tb]
	\begin{centering}
		\includegraphics[width=\textwidth]{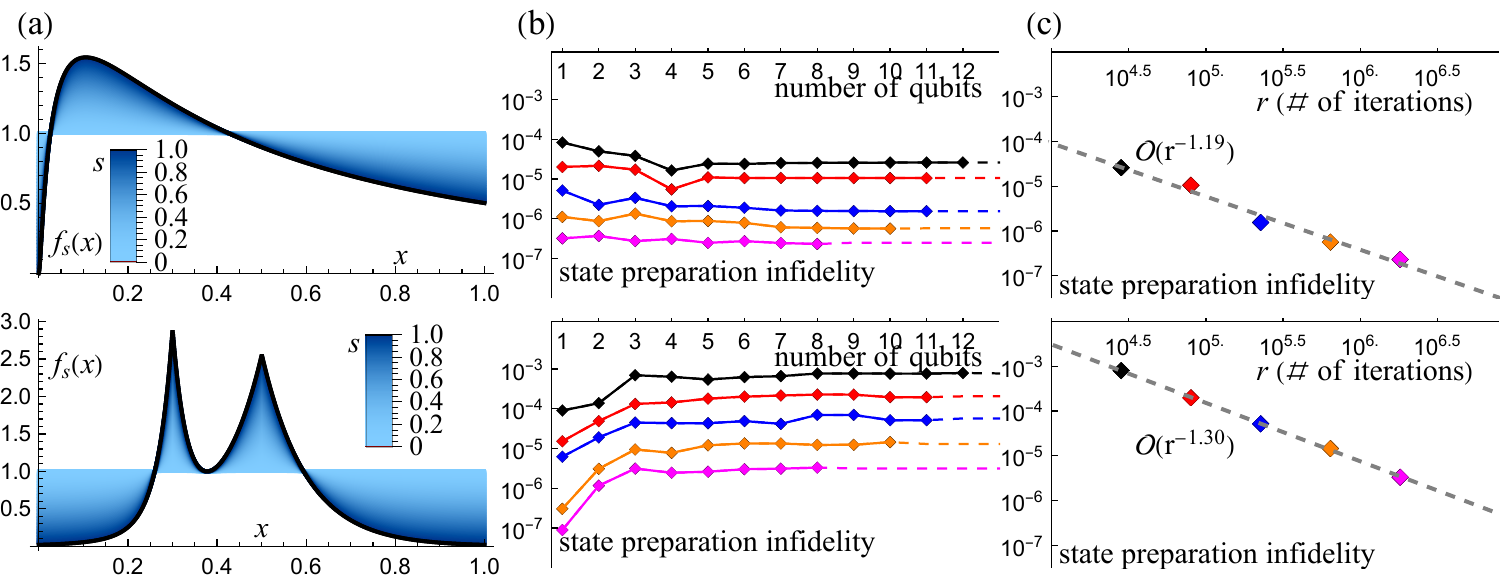}
		\caption{
			(a) Preparing a log-normal distribution (above) as relevant in financial systems~\cite{crow1987lognormal, stamatopoulos2020option, chakrabarti2021threshold} and (below) a Slater-type function  used in quantum chemistry.
			We efficiently simulate
			$e^{-i \Delta t \mathcal{H}(s)}$ where $\mathcal{H}(s)$ encodes the continuous function $f_s$
			as its ground state and we adiabatically evolve $f_s$ ($s$ is colour coded)
			from the constant function $f_0=1$ 
			to the desired function $f_1$ (black solid line).
			(b) State preparation error $1-\Phi$
			as the infidelity $\Phi = |\langle \psi | \tilde{\psi} \rangle|^2$
			at increasing numbers $n$ of qubits explicitly taking into
			account all sources of error by simulating dynamics of the $2n$-qubit
			\emph{exact} $1$-sparse simulation.
			Increasing $n$ (horizontal direction) the infidelity
			approaches a constant in exponential order -- solid diamonds are simulated while
			dashed lines are extrapolated given the high classical computational cost of simulation.  function.
			Increasing the number of iterations $r$ (varying colours in vertical direction) allows us to better approximate the ideal
			adiabatic evolution and decreases the error.
			(c) The theoretically
			leading error  $\mathcal{O}(1 /\sqrt{r})$ of 
			adiabatic evolution is dominated in the practically relevant region by the error of the low-rank simulation technique $\mathcal{O}(1/r)$ 
			due to a larger absolute factor as confirmed by the empirical fits (dashed grey line).
			\label{fig:fig2}
		}
	\end{centering}
\end{figure*}

For an example of a pathological instance that helps build intuition as to the limitation of $\fillr$, consider the unstructured search problem where we are given a function $F^{(n)}: \{0,1\}^n \mapsto \{0,1\}$ (where only one input $x=m$ corresponds to a non-zero output).
Preparing the function $F^{(n)}$ in an $n$-qubit quantum register would reveal the non-zero value at $m$ upon measuring the state
and thus would solve the unstructured search problem with high probability~\cite{grover1996fast, nielsen2002quantum}. While this is not a discretized (continuous) function, it is directly analogous to a box function of width $\approx 2^{-n}$ depending on the qubit count $n$.
This results in a generalised filling ratio $\tilde{\mathcal{F}} \approx 2^{-n/2}$ that decreases exponentially as we detail in \cref{sec:generalisation} and thus our approach would require $\mathcal{O}(2^{2n})$ calls to the oracle $f(x)$ and is thus less efficient than a classical direct search. Note that such a pathological problem instance was only constructed because the generalised filling ratio $\tilde{\mathcal{F}}$ is not an asymptotic constant (while the special-case filling ratio $\fillr$ for continuous functions is an asymptotic limit in $n$ with guaranteed convergence to a constant).

Finally, the impact of the filling ratio may in practice be circumvented. In one scenario, suppose the function being prepared has the majority of its density in an interval $[a', b'] \subset [a, b]$ that is relatively small compared to the domain.
Then, we can simply prepare the function in $[a', b']$ (at some smaller qubit count) with a proportionally larger filling ratio. 
We can then add further qubits in the $\ket{0}$ state 
and shift the resulting function via quantum addition.
If the function being prepared is non-zero on an exponentially small portion of the overall state space, and the location of the non-zero states are known, then this technique can yield an exponential improvement in the filling ratio. 
As a result, functions that are very narrow (i.e. very peaked), but have all of their density within a single small known 
region, including periodic narrow peaks (e.g. functions approaching a Dirac comb),
can have their effective filling ratios increased significantly.
We also note that $d$-sparse simulation techniques may also be useful for such very peaked functions if we can specify an oracle that efficiently computes the locations of non-zero function values~\cite{berry2015hamiltonian}.

Another technique one can employ to mitigate the constant cost of the filling ratio is to utilize a non-uniform grid. In such a grid, more samples of the function are taken in regions where the function is particularly narrow, thus giving a larger effective filling ratio. We leave such non-uniform grids as a topic for future exploration.

\section{Numerical Demonstration}

We now demonstrate our approach in the preparation of two single-variate continuous functions which are practically relevant. First, we consider the probability density function of a log-normal distribution in \cref{fig:fig2} (a, above) which is relevant in modelling a wide range of natural phenomena as well as financial systems~\cite{crow1987lognormal, stamatopoulos2020option, chakrabarti2021threshold}.
Second, we consider a Slater-type wave-function as a linear combination of primitive functions $\beta e^{-\alpha |x - x_0|}$ with constants $\alpha,\beta, x_0\in \mathbb{R}$. Gaussian and Fourier approximations to atomic orbitals are primarily used in quantum chemistry applications but despite their paramount practical utility they cannot accurately capture the crucial cusps at positions of atomic nuclei \cite{mcardle2020quantum}. Indeed our approach allows the efficient preparation of \emph{exact} Slater-type orbitals thereby increasing the accuracy of grid-based simulations~\cite{mcardle2020quantum, chan2022grid}.

In \cref{fig:fig2} (a) we initially prepare a uniform superposition $|+\rangle^{\otimes n}$ in a register of $n$ qubits and then adiabatically morph this into our desired function. In the figure, the color gradient represents the smooth interpolation of $f_s$ throughout the adiabatic trajectory, with the lightest shade representing $f_0$ and the black line representing $f_1$. Taking a slice in the figure where only a certain color is kept (corresponding to a time $s$), would yield a plot of the function $f_s$.

In \cref{fig:fig2} (b) we explicitly simulate the dynamics under $e^{- i T/r S_A}$ that acts on overall $2n$ qubits thereby taking into account all sources of error (including error in binary encoding of the function values $f(x)$). Here, each line corresponds to a series of simulations evaluated with a fixed time $T$ and a fixed number of iterations $r$ as we increase the number of qubits $n$. Our simulations nicely confirm our bounds in the following two distinct asymptotic limits.

First, as we increase the number of qubits $n$ the state-preparation error (solid lines) rapidly, in exponential order, approaches a constant value as expected from our asymptotic bounds -- the asymptotic constant error (dashed lines are extrapolations) is then determined by the filling ratio $\fillr$ and by $T$ and $r$.
Second, the series of lines from black to magenta in \cref{fig:fig2} (b) represent simulations with an increasing $r$ and confirm that as we increase the number of iterations $r$ the error decreases polynomially (as both the precision of the simulation step increases, and the discretization error of the adiabatic evolution decreases).

It is important to note that our error bounds in \cref{theorem:total_query_complexity_informal} are very pessimistic as we have used an unrealistic, worst-case scenario argument in our proofs. In particular, while our bounds scale as $O(\sqrt{1/r})$, in practically relevant cases the error is expected to be dominated by the algorithmic error of the low-rank simulation approach -- as it depends polynomially on the filling ratio. A simple argument shows that neglecting the adiabatic evolution error allows us to obtain the tighter error bound $\mathcal{O}(1/r)$. This is nicely confirmed in \cref{fig:fig2} (c) whereby we plot the asymptotic constant errors from \cref{fig:fig2} (b) as a function of the number of iterations $r$ and obtain an empirical scaling from approximately $\mathcal{O}(r^{-1.19})$ to $\mathcal{O}(r^{-1.3})$. We also note that the \emph{absolute} error of our procedure is below our error bound
by \textit{many} orders of magnitude, suggesting reasonably low absolute factors in practice.

\section{Further applications}

The techniques we have developed in this paper allow for a number of promising further applications. Some of these applications are related to the state preparation procedure, and some are consequences of the techniques we have developed. Of course, the adiabatic state preparation procedure that we present immediately enables applications that rely on continuous input states, such as those we have outlined in \cref{sec:background}.

However, the state preparation procedure can actually be an application in its own-right. For instance, an important problem is to draw samples following a particular given continuous probability density function. Classically, techniques to do so, such as inverse transform sampling, generally require integrating the probability density function in question, and thus have complexity that scales with the resolution of the grid discretization~\cite{olver2013fast}. 
Our state-preparation approach immediately enables such sampling, simply by preparing the square-root of the desired density function and then measuring the resulting state. Our approach has overall logarithmic complexity in doing so (and is thus an instance of exponential quantum advantage), and moreover produces \textit{truly random samples} from the density function, rather than just pseudo-random samples.  

Furthermore, given that our approach allows us to prepare a quantum state whose amplitudes are proportional to the samples $f(x_j)$, it follows
that the measurement outcome of the bitstring $j \in \{0,1\}^n$ has a probability proportional to $|f(x_j)|^2$.
As a result, we can directly use our approach to efficiently draw importance samples from efficiently computable functions, and if the function has one dominant maximum then importance sampling can even find this global maximum with a high probability. 

An alternative proposal for using our techniques for global optimization follows. Given a function $g:\mathbb{R}\mapsto \mathbb{C}$, preparing $f(x) = g(x)^k$ for some constant $k>1$ allows us to find the \textit{global maximum} of $g(x)$  with probability that improves exponentially in $k$ through direct sampling.
However, the filling ratio of $f(x)$ also decreases exponentially in $k$, and so the complexity of preparing the state negates any benefit of using $k>1$. On the other hand, if the algorithm's asymptotic dependence of $\fillr$ can be improved, then such an approach could be an extremely promising global optimization procedure.

Additionally, as discussed in the Appendix, the simulation technique we have developed enables the efficient integration of Lipschitz continuous probability density functions through the use of phase estimation. In particular, we can prepare the state $| \psi \rangle$ with our adiabatic approach and apply time evolution under $\H(s)$. By implementing phase estimation we can estimate the only non-zero eigenvalue of $\H(s)$ with near-perfect success rate and thus obtain an estimate of the $L^2$ norm of the encoded function.

\section{Discussion and Conclusion}
In this work we introduced a family of techniques that allow the \emph{efficient} loading of any efficiently computable function onto a quantum register such that amplitudes of the state vector $\ket{\psi}$ are proportional to samples of the function. Moreover, the state preparation technique is actually more general, and also applies to certain discontinuous functions.

Motivating the power of our techniques, the number of states in an $n$-qubit Hilbert space grows exponentially as $2^n$, so a quantum computer allows us to efficiently ``compute'' an exponential number of samples from our function with one oracle query. This property enables, amongst other things, efficient state preparation which is essential to a large number of key applications of quantum computers. For example, applications in fields ranging from the simulation of quantum physics and chemistry~\cite{chan2022grid} to financial forecasting~\cite{stamatopoulos2020option} will all benefit from our techniques.

A central idea in this paper is that we can embed the desired quantum state $\ket{\psi}$ into a rank-1 Hamiltonian $\mathcal{H} \propto \ket{\psi} \bra{\psi}$ whose ground state is exactly $\ket{\psi}$, and with all other orthogonal states having an eigenvalue of $0$. We then utilized powerful low-rank~\cite{rebentrost2018quantum} and 1-sparse~\cite{childs2004quantum,aharonov2003adiabatic,childs2003exponential,berry2007efficient}  Hamiltonian simulation techniques from the literature, allowing us to efficiently simulate $\mathcal{H}$. Among other variants, we presented an adiabatic-evolution approach that enables quasi-deterministic preparation of the ground state of the rank-1 matrix $\mathcal{H}$. Central to the argument of this procedure's efficiency is that the spectral gap of such a rank-1 matrix is the same as its spectral norm, leading to a rare case where an adiabatic algorithm has a provably bounded spectral gap (which is essential to proving the efficiency of the adiabatic procedure). We rigorously proved that the query complexity of our adiabatic
approach is a constant independent of the number of qubits. That is, we need only make a constant number of calls to the oracle computing the function amplitudes, regardless of the number of qubits the state is being prepared on. Moreover, we also prove that the state-preparation error and failure probability can be suppressed arbitrarily by increasing the temporal resolution (i.e. the number of steps in the adiabatic evolution $r$). It is worth noting that even though we perform measurements to reset the state of the ancilla register, our approach becomes reversible for a sufficiently high state-preparation precision given that the measurement outcomes are then
near-deterministic as per \cref{theorem:total_query_complexity_informal}. Moreover, we stress that the non-deterministic component of the algorithm is essentially an analytical convenience, and that in practice one can accept any outcome of the ancilla measurement. 

We would additionally like to emphasize that the presented ideas are ripe for future modification, optimization, and generalization. For example, in practice one can use our adiabatic approach to efficiently and quickly prepare a crude approximation of $\ket{\psi}$ (via low precision $\epsilon$) and then either refine it via our phase estimation or destructive interference procedures (with high probability of success). Moreover, further optimizing the algorithm's asymptotic dependence on both $\epsilon$ and $\fillr$ is a promising area of future exploration, with potentially significant ramifications. Indeed we expect with minor modifications one may be able to reduce the $O(\Delta t^2)$ error of the low-rank simulation approach -- which is currently a bottleneck.
As such, we expect our work will spark further developments, and can potentially lead to a rich set of quantum algorithms for manipulating and dealing with continuous functions by utilizing quantum parallelism.

Our work almost immediately allows for the preparation of continuous multivariate functions, albeit with the main limitation that the filling ratio may decrease asymptotically in the dimension of the function being prepared, as we discuss in \cref{sec:multivar}.

The presented algorithms (and related ones in the literature) assume execution on fault-tolerant quantum computers. However, the quantum devices of the near-future are error prone~\cite{preskill2018quantum} and there is limited work on preparing quantum states for NISQ devices. One possible approach utilizes quantum generative adversarial networks (qGAN) to tune the parameters in variational circuits to load a given distribution~\cite{zoufal2019quantum}. Some work on using qGANs to produce multivariate quantum states has also been conducted~\cite{zhu2021generative, agliardi2022optimal}. Alternatively, Rattew \textit{et al.} demonstrate how normal distributions may be efficiently produced in a manner resistant to most hardware errors~\cite{rattew2021efficient}. The techniques presented in that work can potentially be directly applied to the work presented in this paper to discard state preparation attempts in which hardware errors likely occurred. Furthermore, combination with powerful error mitigation strategies may also be possible, should the states produced be used to compute the expectation values of observables~\cite{koczor2021exponential}. However, to be useful with the algorithms presented in this work the infidelity of (uncorrected) hardware operations would likely need to be orders of magnitude lower than current state-of-the-art approaches (e.g. refs~\cite{pino2020demonstration, pelofske2022quantum}).

\section*{Acknowledgments}
The authors would like to thank Simon Benjamin, Marco Pistoia, Dylan Herman, Yue Sun, and the entire FLARE group at JPMorgan Chase Bank, N.A. for numerous insightful and helpful conversations.
The authors would like to give special thanks to Yue Sun and Dylan Herman for carefully checking the proofs.
This work was supported in part by JPMorgan Chase Bank, N.A. through the Future Lab for Applied Research and Engineering (FLARE) Ph.D. Fellowship program.
B.K. thanks the University of Oxford for
a Glasstone Research Fellowship and Lady Margaret Hall, Oxford for a Research Fellowship.
The authors would like to acknowledge the use of the University of Oxford Advanced
Research Computing (ARC) facility in carrying out this work.

\section*{Disclaimer}
This paper was prepared for information purposes by the teams of researchers from the various institutions identified above, including the Future Lab for Applied Research and Engineering (FLARE) group of JPMorgan Chase Bank, N.A..  This paper is not a product of the Research Department of JPMorgan Chase \& Co. or its affiliates.  Neither JPMorgan Chase \& Co. nor any of its affiliates make any explicit or implied representation or warranty and none of them accept any liability in connection with this paper, including, but not limited to, the completeness, accuracy, reliability of information contained herein and the potential legal, compliance, tax or accounting effects thereof.  This document is not intended as investment research or investment advice, or a recommendation, offer or solicitation for the purchase or sale of any security, financial instrument, financial product or service, or to be used in any way for evaluating the merits of participating in any transaction.

\bibliography{bibliography}

\clearpage
\onecolumngrid
\appendix

\begin{center}
     \Huge \textbf{Appendix}
\end{center}
\normalsize 

The structure of this appendix is as follows. We first introduce basic definitions and notation in \cref{sec:prelim}. In \cref{section:hamiltonian_simulation_techniques} we review relevant Hamiltonian simulation techniques of crucial importance which form the basis of
our state-preparation procedure. These include 1-sparse in \cref{sec:1sparse} and low-rank in \cref{subsection:lowrank_hamiltonian_simulation} simulation techniques as well as adiabatic simulation in \cref{section:adiabatic_background}.

We detail our adiabatic state-preparation approach in \cref{sec:direct_state_prep}; First we state a compact form of our main result in \cref{theorem:total_query_complexity} and then describe details of the procedure in \cref{sec:direct_prep_procedure}.
We then systematically derive our upper bounds, first for an arbitrary finite qubit count in \cref{sec:non-asymptotic}
and then we prove existence of constant asymptotic bounds in \cref{sec:asymptotic} for continuous functions. We then outline generalisations
to arbitrary functions (mappings) in \cref{sec:generalisation}.

In the next three sections we detail further applications beyond our main result: we propose a `self-verifying' state-preparation approach 
based on a phase estimation protocol in \cref{section:qpe_state_prep} as well as a more compact variant of it based on a Hadamard test in \cref{section:prep_via_destructive_interference}. We also discuss that these two variants can be used to efficiently
integrate functions while we point out the phase-estimation approach has a superior convergence rate. We finally discuss generalisation
to continuous multi-variate functions. The last section \cref{sec:further_technical} contains further technical details.

\section{Preliminaries and definitions \label{sec:prelim}}

\subsection{Continuous and efficiently computable functions}

While our approach naturally applies to arbitrary functions (mappings) as we detail in \cref{sec:generalisation},
we mostly focus on the pivotal case of continuous functions given these contain most practically important 
cases and they naturally admit convenient asymptotic properties as we show later.

In this work we refer to a continuous function $f: [a,b] \mapsto \mathbb{C}$ on the  bounded interval
$a \leq x \leq b$ as $f(x)$ whose general definition is $f\in \mathcal{C}$  as $\lim_{c \rightarrow x} f(c) = f(x)$.
Recall that due to the boundedness theorem all such functions are bounded as $|f(x)|\leq \lVert f \rVert_{max}$.
While we keep our results as general as possible and consider arbitrary continuous functions (unless stated otherwise), 
we note that so-called Lipschitz continuous functions cover nearly all instances one can encounter in practice,
and we detail in \cref{sec:further_asymptotic} that these Lipschitz continuous functions
additionally satisfy powerful exponential scaling properties.

In the rest of this work we will refer to continuous functions and their discretisations intergchangebly;
in a digital (quantum) computer we aim to represent both the argument $x$ and the function value $f(x)$ in a
binary encoding. While we explictly discuss our discretisations below, let us now define a subset of 
general continuous functions which additionally satisfy efficient computability constraints whith respect to some suitable
binary encoding.
\begin{definition}[efficiently computable functions]\label{def:eff_comp}
Given the above families of continuous functions, we consider $f$ to be efficiently computable if there exists an efficient quantum algorithm $O_f$ performing the mapping,
\begin{align}
    O_f\ket{x}\ket{0} = \ket{x}\ket{f(x)},
\end{align}
for any computational basis state $\ket{x}$
given a suitable binary encoding of the real and complex numbers x  and f(x), respectively. Such an oracle may be efficiently constructed for any function where an efficient classical circuit allows the evaluation of $f$ for any possible input in the domain. Intuitively, if there is a classical algorithm allowing the evaluation of $f$, then it is possible to construct $O_f$.
\end{definition}

Let us consider a simple example. Consider the quantum circuit $O_f$ implementing the function $f(x)=x$. This can be implemented by a constant-depth circuit as a sequence of CNOT gates that are controlled on the individual bits in the $\ket{x}$ register. Suppose the first register has $n$ bits and the second has $d$, if $n>d$, we can simply just control on the first $n$ bits of the first register. Furthermore, more involved functions can be computed efficiently by invoking addition and other arithmetic operations; these can be implemented as discussed in refs~\cite{vedral1996quantum, draper2000addition, draper2004logarithmic, takahashi2009quantum, bhaskar2015quantum}.

A large number of applications require that a quantum state encodes a probability distribution, i.e.,  a non-negative function.
Given a non-negative function, its integral over any domain is also non-negative and we will see later
that the ability to efficiently integrate a non-negative function greatly simplifies our general problem, i.e., through integration we effectively broaden possibly narrow peaks.
\begin{definition}[efficiently integrable functions] \label{def:integrable_fct}
We define the non-negative function $g: [a,b] \mapsto  \mathbb{R}_+ $ such that $g(x_j)$ gives the density in the grid interval between state $\ket{x_j}$ and $\ket{x_{j+1}}$,
\begin{align}
    g(x_j) = \left[ \int_{x_j}^{x_j + \step} f(x) \, \mathrm{d} x \right]^{1/2},
\end{align}
where $f:[a,b] \mapsto \mathbb{R}_+$ is a non-negative function and we assume that the resulting integral function
$g(x_j)$ is efficiently computable as in Definition~\ref{def:eff_comp}. Here, $\Delta_N \equiv \frac{b-a}{N}$.
As such, we can efficiently enforce the normalisation 
$\sum_j |g(x_j)|^2 = \lVert f \rVert_1 = 1 $ which we will assume is always the case.
\end{definition}

Let us consider the function $f(x)=x$ for example. We can analytically compute its integral as $g(x_j) =  \sqrt{\Delta_N (\Delta_N/2 +  x_j)}$ and implement the corresponding oracle by using arithmetics operations.

\begin{definition}[filling ratio]\label{def:peakedness}
Given the usual definition for $L^p$ function norms of a Riemann integrable, bounded function $f$, $\lVert f \rVert_p=\left[\int_{a}^{b} |f(x)|^p \, \mathrm{d} x\right]^{1/p}$, we define the filling ratio of $f$, 
\begin{equation*}
    \fillr :=  \lVert f \rVert_1 / \lVert f \rVert_{max}.
\end{equation*}
Note that for continuous functions over a closed interval $\lVert f \rVert_\infty =: \lVert f \rVert_{max}$ represents the absolute largest value of the function.
\end{definition}
Above the filling ratio $0 < \fillr \leq 1$ actually quantifies
the absolute area under the function relative to a rectangle of the same maximal height; the upper bound ($1$) is saturated by a maximum entropy uniform function $f(x) =\mathrm{const}$ for which
our presented techniques are most efficient. On the other hand, the filling ratio $\fillr$ can be arbitrarily small but is never zero --
the lower bound is approached by functions that consist of arbitrarily narrow peaks.

\subsection{Encoding functions into quantum states}

We are generally concerned with the preparation of quantum states following continuous functions. We now introduce some of the notation used throughout the remainder of this paper.
First, we assume that we have an $n$ qubit system, with $N=2^n$ computational basis states.
Second, as explained above we consider continuous functions $f(x)$ on some real interval $x \in [a, b]$ and we discretise the interval into $N$ equally spaced samples with $a=x_0$ is the leftmost value in the interval, and $x_{N-1}$ is the rightmost value via $b=x_{N}$.
The step size is thus given by $\step = (x_{N}-x_0)/N = (b-a)/N$.
Note that the binary integers $j \in \{0,1\}^n$ index the standard standard basis states $\{\ket{j}\}_j$ of our quantum computer
and we can equivalently refer to these in terms of the x values $\{\ket{x_j}\}_j$
given the relation 
$x_j = x_0 + j \step $
admits the inverse mapping 
$
    j = (x_j - x_0)\step^{-1}
$.

Suppose we have an $n$ qubit quantum system, prepared in the state $\ket{\psi}$.
We define two kinds of encodings as
\begin{equation}\label{eq:encoding}
 \textbf{pointwise:}   \quad \ket{\psi} = (\mathcal{N}_N)^{-1}  \sum_{j\in \{0,1\}^n}    f(x_j) \ket{j},
 \quad \quad
 \textbf{integral:} \quad   \ket{\psi} = \sum_{j\in \{0,1\}^n}    g(x_j)  \ket{j},
\end{equation}
where the integral encoding is identical to that in the Grover-Rudolph approach detailed in \cref{sec:further_technical}.
Here we assume that $f(x)$ is an efficiently computable function from Definition~\ref{def:eff_comp} and $\mathcal{N}_N :=  \sum_{j\in \{0,1\}^n}  |f(x_j)|^2$ is a normalisation constant which we cannot necessarily compute efficiently.
Here $g(x_j)$ are piecewise integrals of an efficiently integrable function from Definition~\ref{def:integrable_fct} and, as such, we can efficiently enforce normalisation.

While the two encodings may have very different behaviours for finite $N$ we
show in Property~\ref{prop:encodig_equiv} that for large $N$ asyomptotically the integral encoding becomes identical to the corresponding pointwise encoding of a non-negative function.

	\begin{remark}
		While in the present work we restrict ourselves to a standard, uniform grid representation as per \cref{eq:encoding},
		it is straightforward to implement non-uniform grids that are required for, e.g., computing quadrature integrals.
		One would then proceed by efficiently computing $f(q(x_j))$ where $q(x_j)$ efficiently maps the uniform grid
		to a non-uniform one and thus the composite function is efficiently computable.
	\end{remark}

Note that later we will also make use of a binary encoding of the function values $f(x)$, which we we discuss in detail in Sec.~\ref{sec:1sparse}.

\section{Hamiltonian Simulation Techniques}\label{section:hamiltonian_simulation_techniques}
\subsection{1-Sparse Hamiltonian Simulation \label{sec:1sparse}}

Let us now recollect the approach described in refs.~\cite{childs2004quantum,aharonov2003adiabatic,childs2003exponential,berry2007efficient}
for efficiently simulating the time evolution under 1-sparse Hamiltonians.

\begin{lemma}\label{lemma:one_sparse_simulation}
Suppose we are given a one-sparse Hermitian Hamiltonian $H$ acting on $n$ qubits. $H$ is specified by the oracles $O_f$ and $O_H$, defined such that 
\begin{align}
    O_f\ket{x}_n\ket{0}_n &= \ket{x}_n\ket{f(x)}_n,\\
    O_H\ket{x}_n\ket{y}_n\ket{0}_d &= \ket{x}_n\ket{y}_n\ket{H_{x,y}}_d,
\end{align}
where the function $f(x)$ gives the column of the only non-zero matrix element of $A$ in row $x$, $H_{x,y}$ corresponds to the element in the $x^{th}$ row and $y^{th}$ column of $H$, and we use $d$-bits of resolution to encode the matrix elements of $H$. Then, an initial state $\ket{\psi_0}$ can be simulated by $e^{-i H t}$ for an arbitrary $t$ with $0$ error using $O(1)$ oracle calls (and circuit depth equal to the circuit depth of the oracle implementations).
\end{lemma}

\begin{proof}
As we include this proof simply to help the reader build intuition, we only consider the case where the matrix elements of $H$ are real, and defer the reader to the existing literature to see how this technique is adapted for the general case of a Hermitian $H$~~\cite{childs2004quantum,aharonov2003adiabatic,childs2003exponential,berry2007efficient}(noting that the overall asymptotic complexity remains the same)
\footnote{
	\begin{minipage}{0.4\textwidth}
	Nathan Weibe has an excellent video lecture covering 1-sparse Hamiltonian simulation in \href{https://youtu.be/tllz6y7WUUs}{[link]}.
	\end{minipage}
}. 
First, observe that for a one-sparse $H$ satisfying $H=H^{T}$, $H\ket{x} = \ket{f(x)}$ as $H$ maps $\ket{x}$ to the only non-zero row, which must be the same as the non-zero column in row $x$ since $H$ is symmetric. Moreover, it must also be the case that $f(f(x)) = x$, otherwise $H$ would not be one-sparse. As a result, $H$ can be thought of as an adjacency matrix describing a graph with $2^n$ vertices, such that the graph consists of a set of connected components with each component containing either one or two vertices. Therefore, $H$ acts on $1\times 1$ and $2\times 2$ invariant subspaces. Consequently, simulation of $e^{-i H t}$ is equivalent to simply performing simulations in each of these $1\times 1$ and $2\times 2$ subspaces. Thus, our simulation circuit needs the logic to identify if a given basis vector is in a $1\times 1$ or a $2\times 2$ block, and then needs to perform the corresponding simulation. 

First, we write our initial quantum state in the general form $\ket{\psi_0}=\sum_{x = 0}^{2^n -1}\alpha_x\ket{x}$. We then initialize our quantum state coupled to an $n$-qubit ancillary register, a single qubit ancillary register, a $d$-qubit ancillary register, and a final $1$-qubit ancillary register. Note that a more efficient implementation of this procedure exists, but we use the following for expository purposes. Our joint quantum state is thus given by,
\begin{align}
    \ket{\psi_0}_n\ket{0}_n\ket{0}_1\ket{0}_d\ket{0}_1 = \sum_{x = 0}^{2^n -1}\alpha_x\ket{x}_n\ket{0}_n\ket{0}_1\ket{0}_d\ket{0}_1. 
\end{align}
We begin the simulation procedure by applying the $O_f$ oracle to the first and second registers,
\begin{align}
    \ket{\psi_0}_n\ket{0}_n\ket{0}_1\ket{0}_d \ket{0}_1
    \xrightarrow{O_f}
    \sum_{x = 0}^{2^n -1}\alpha_x\ket{x}_n\ket{f(x)}_n\ket{0}_1\ket{0}_d\ket{0}_1.
\end{align}
We now define the comparison operator, $\text{CMP}$, which acts on the first three registers with the mapping $\text{CMP}\ket{x}_n\ket{y}_n\ket{0}_1 = \ket{x}_n\ket{y}_n\ket{x > y}_1$ (where $x>y$ is $1$ if true, and $0$ if false). Applying $\text{CMP}$ yields,
\begin{align}
    \xrightarrow{\text{CMP}}
    \sum_{x = 0}^{2^n -1}
    \begin{cases}
    \alpha_x\ket{x}_n\ket{f(x)}_n\ket{0}_1\ket{0}_d\ket{0}_1 & \text{ if } x \le f(x)\\
    \alpha_x\ket{x}_n\ket{f(x)}_n\ket{1}_1\ket{0}_d\ket{0}_1 & \text{ if } x > f(x).
    \end{cases}
\end{align}
We now define the controlled-swap operation, $\text{CSWAP}$, acting on the first three registers such that $\text{CSWAP}\ket{x}_n\ket{y}_n\ket{0}_1=\ket{x}_n\ket{y}_n\ket{0}_1$ and $\text{CSWAP}\ket{x}_n\ket{y}_n\ket{1}_1=\ket{y}_n\ket{x}_n\ket{1}_1$. Applying $\text{CSWAP}$ yields,
\begin{align}
    \xrightarrow{\text{CSWAP}}
    \sum_{x = 0}^{2^n -1}
    \begin{cases}
    \alpha_x\ket{x}_n\ket{f(x)}_n\ket{0}_1\ket{0}_d\ket{0}_1 & \text{ if } x \le f(x)\\
    \alpha_x\ket{f(x)}_n\ket{x}_n\ket{1}_1\ket{0}_d\ket{0}_1 & \text{ if } x > f(x).
    \end{cases}
\end{align}
We now apply $O_H$, acting on the first, second, and fourth registers.
\begin{align}
    \xrightarrow{O_H}
    \sum_{x = 0}^{2^n -1}
    \begin{cases}
    \alpha_x\ket{x}_n\ket{f(x)}_n\ket{0}_1\ket{H_{x, f(x)}}_d\ket{0}_1 & \text{ if } x \le f(x)\\
    \alpha_x\ket{f(x)}_n\ket{x}_n\ket{1}_1\ket{H_{f(x), x}}_d\ket{0}_1 & \text{ if } x > f(x).
    \end{cases}
\end{align}

\begin{figure*}[tb]\label{fig:crx_gate_definition}
    \[
        \Qcircuit @C=1em @R=0.7em {
            \lstick{\ket{a}_d}    & \qw & \multigate{1}{\text{CRx}(t)} & \qw & \rstick{\ket{a}_d} \qw \\
            \lstick{\ket{\psi}_1} & \qw & \ghost{\text{CRx}(t)}        & \qw & \rstick{e^{-i a t X}\ket{\psi}_1} \qw 
        }
    \]
    \caption{Quantum circuit representation of the $\text{CRx}(t)$ gate.}
\end{figure*}
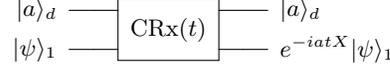
We must now take a brief interlude, and define two operations: $\text{CRx}(t)$ and $\text{CRz}(t)$ such that the gates act on a $d$-qubit control register and a $1$-qubit target register, as shown in Figure~\ref{fig:crx_gate_definition}. In particular, we define the mapping when the input state is some basis vector $\ket{a}_d$, and clearly this definition holds in general by linearity. 
First, allowing the binary expansion $a = a_1a_2...a_d$ we note,
\begin{align}
    e^{-i a t X}
    =
    e^{-i (a_1 2^{d-1} + a_2 2^{d-2} + ... + a_d 2^{d-d}) t X}
    =
    e^{-i a_1 2^{d-1} t X}e^{-i a_2 2^{d-2} t X} ... e^{-i a_d 2^{0} t X}.
\end{align}
Thus, we can clearly implement this gate by a sequence of $d$ single-qubit controlled $\text{Rx}$ gates, with $\text{Rx}(t) = \cos(\frac{t}{2})I-i\sin(\frac{t}{2})X$, where the $j^{th}$ gate is controlled on bit $a_j$ and the time is set to $t 2^{d - j}$. Thus, the cost of this operation scales as $O(d)$ in terms of circuit depth and required parameterized gates. The same argument applies to creating the $\text{CRz}(t)$ gate (although it can trivially be created from the $\text{CRx}(t)$ gate by applying a Hadamard before and after the gate on the target qubit).

We now must determine whether a given state is in a $1\times 1$ or $2\times 2$ invariant subspace, to do this we use the final ancillary qubit (note that this qubit is not necessary, and this operation can be done in place, however we include it for notational clarity). Define a new operator, $\text{EQ}$, which acts on the first two registers and the fifth register, with the mapping $\text{EQ}\ket{x}_n\ket{y}_n\ket{0}_1 = \ket{x}_n\ket{y}_n\ket{x = y}$, where $x=y$ is $1$ if true, and $0$ otherwise. Then,
\begin{align}
    \xrightarrow{\text{EQ}}
    \sum_{x = 0}^{2^n -1}
    \begin{cases}
    \alpha_x\ket{x}_n\ket{x}_n\ket{0}_1\ket{H_{x, f(x)}}_d\ket{1}_1 & \text{ if } x = f(x) \\ 
    \alpha_x\ket{x}_n\ket{f(x)}_n\ket{0}_1\ket{H_{x, f(x)}}_d\ket{0}_1 & \text{ if } x < f(x)\\
    \alpha_x\ket{f(x)}_n\ket{x}_n\ket{1}_1\ket{H_{f(x), x}}_d\ket{0}_1 & \text{ if } x > f(x).
    \end{cases}
\end{align}
Since the first two registers are only equal if $f(x)=x$, they are only equal if $\ket{x}$ is in a $1\times 1$ invariant subspace, and as such the correct evolution prescribes a rotation by $\text{CRz}$. If the first two registers are not equal, then we must be in a $2\times 2$ invariant subspace, and so the correct evolution follows a rotation by $\text{CRx}$.
Noting that both the $\text{CRx}$ and $\text{CRz}$ gates use the fourth register as their ``pass-through`` input, and the third register as their target input, we can then apply $\text{CRx}$ gate conditioned on the $\ket{0}_1$ state of the fifth register followed by a $\text{CRz}$ gate conditioned on the $\ket{1}_1$ state of the fifth register (we refer to these controlled-controlled rotation gates as $\text{CCRx}$ and $\text{CCRz}$). Note that this is equivalent to only performing an uncontrolled $\text{CRx}$ gate, and then applying Hadamard gates on the third register before and after the $\text{CRx}$ gate conditioned on the fifth register, however the preceding description yields a clearer analysis. We then obtain,
\begin{align}
    \xrightarrow{\text{CCRx(t) CCRz(t)}}&
    \sum_{x = 0}^{2^n -1}
    \begin{cases}
    \text{CRz}(t)\alpha_x\ket{x}_n\ket{x}_n\ket{0}_1\ket{H_{x, f(x)}}_d\ket{1}_1 & \text{ if } x = f(x) \\ 
    \text{CRx}(t)\alpha_x\ket{x}_n\ket{f(x)}_n\ket{0}_1\ket{H_{x, f(x)}}_d\ket{0}_1 & \text{ if } x < f(x)\\
    \text{CRx}(t)\alpha_x\ket{f(x)}_n\ket{x}_n\ket{1}_1\ket{H_{f(x), x}}_d\ket{0}_1 & \text{ if } x > f(x).
    \end{cases}\\
    &=
    \sum_{x = f(x)}\text{CRz}(t)\alpha_x\ket{x}_n\ket{x}_n\ket{0}_1\ket{H_{x, f(x)}}_d\ket{1}_1\\
    &\ \ +
    \sum_{x < f(x)}
    \text{CRx}(t)\left( 
    \ket{x}_n\ket{f(x)}(\alpha_x\ket{0}_1 + \alpha_{f(x)}\ket{1}_1)\ket{H_{x, f(x)}}_d\ket{0}_1
    \right)\\
    &=
    \sum_{x = f(x)}e^{-i H_{x, x} t}\alpha_x\ket{x}_n\ket{x}_n\ket{0}_1\ket{H_{x, f(x)}}_d\ket{1}_1 \\
    &\ \ +
    \sum_{x < f(x)}
    \ket{x}_n\ket{f(x)}\left(e^{-i H_{x, f(x)} Xt}(\alpha_x\ket{0}_1 + \alpha_{f(x)}\ket{1}_1)\right)\ket{H_{x, f(x)}}_d\ket{0}_1
\end{align}
where the second last equation follows from the fact that if $x < f(x)$, then $\ket{x}_n\ket{f(x)}_n$ is equal to $\ket{f(x)}_n\ket{x}_n$ when $f(x) > x$, and the notation $\sum_{x=f(x)}$ means sum over all $x$ such that $x=f(x)$, and similarly $\sum_{x<f(x)}$ means sum over all $x$ such that $x < f(x)$.
Finally, we uncompute all the ancillary operations (EQ, $O_H$, CSWAP, CMP, $O_f$), obtaining the state
\begin{align}
    \sum_{x = f(x)}e^{-i H_{x, x} t}\alpha_x\ket{x}_n\ket{0}_n\ket{0}_1\ket{0}_d\ket{0}_1 +
    \sum_{x \neq f(x)}
    \alpha_x\left(e^{-iH t} \ket{x}_n\right)\ket{0}_n\ket{0}_1\ket{0}_d\ket{0}_1\\
    =
    \sum_{x=0}^{2^n - 1}\alpha_x\left(e^{-iH t} \ket{x}_n\right)\ket{0}_n\ket{0}_1\ket{0}_d\ket{0}_1
    =
    \left(e^{-iH t} \ket{\psi_0}_n\right)\ket{0}_n\ket{0}_1\ket{0}_d\ket{0}_1.
\end{align}
where this simply follows from the observation that $e^{-i H_{x, x} t}$ correctly describes the evolution of a $1\times 1$ invariant subspace, while $e^{-i H_{x, f(x)} Xt}$ correctly describes the evolution of a $2\times 2$ invariant subspace, and then we simply apply linearity of $e^{-iH t}$.

It is worth briefly mentioning that if the matrix elements of $H$ can be exactly encoded in $d$ bits, then this procedure incurs $0$ error, and if they cannot, then the error vanishes exponentially in $d$ and so is effectively negligible. 
\end{proof}
Note that if the elements of $H$ are efficiently computable, then a quantum circuit implementing $O_H$ can clearly be implemented efficiently (i.e. with polynomial bounded complexity). Similarly, if the location of the non-zero element in any given row can also be computed efficiently, then the overall one-sparse simulation circuit can also be implemented efficiently.

\subsection{Low-rank Hamiltonian Simulation}\label{subsection:lowrank_hamiltonian_simulation}

As low-rank Hamiltonian simulation is an essential subroutine in our algorithm, we now recollect results of Rebentrost \textit{et al.} and summarize their proof~\cite{rebentrost2018quantum}.

\begin{statement}\label{statement:low_rank_simulation}
The procedure presented in~\cite{rebentrost2018quantum} allows the dynamics under an arbitrary low-rank dense matrix $A \in \mathbb{C}^{n\times n}$ to be simulated efficiently if the matrix satisfies the two conditions: $\lVert A \rVert_{max} = \Theta(1)$ and 
$\lVert A \rVert_{2} = \Theta(N)$ as we increase the dimension $N$. 
Moreover, we suppose we have oracular access to the elements of $A$.
Then, for time $\Delta t$, we can simulate the dynamics under $e^{-i \Delta t A}$ using $4$ oracle queries, with trace-norm error bounded by $2\lVert A \rVert_{max}^2 \Delta t^2$.
\end{statement}
\begin{proof}
We now summarize the results of~\cite{rebentrost2018quantum}.
First, we begin by embedding the elements of $A$ in a matrix $S_A\in \mathbb{C}^{N^2\times N^2}$ as follows,
\begin{align}
    S_A = \sum_{j=0}^{N - 1}\sum_{k=0}^{N-1} A_{jk} \op{k}{j}\otimes \op{j}{k}.
\end{align}
Clearly, $S_A$ is one-sparse. Moreover, element $x, y$ of $S_A$, $[S_A]_{xy}$ can be computed easily and is given by $A_{jk}$, where $x=jN + k$ and $y=jN + k$. Thus, $O_{S_A}$ is clearly efficiently implementable, given access to the elements of $A$. Moreover, given the structure of $S_A$ it is also straightforward to implement $O_f$, and as a result we can perform the simulation $e^{-i S_A \Delta t}$ using Lemma~\ref{lemma:one_sparse_simulation} with a total of $4$ oracle calls, and with $0$ error (assuming no discretization error in the matrix elements of $A$). We now assume that the state we wish to evolve has a density operator given by $\sigma$, and we have a separate $n$-qubit register in state $\rho = \op{\phi}{\phi}$ such that $\ket{\phi}=\frac{1}{\sqrt{N}}\sum_{j=0}^{N-1}\ket{j}$ (i.e. a uniform superposition).
The exact evolution for time $\Delta t$ under $S_A$ of state $\rho\otimes \sigma$ generates the following dynamics under the partial trace,
\begin{align} \label{eq:low_rank_partial_tr}
    \mathrm{tr}_1 [e^{-i \Delta t S_A} \rho \otimes \sigma e^{i \Delta t S_A}]
    =
    e^{-i \Delta t A/N} \sigma e^{i \Delta t A/N} + \epsilon_0.
\end{align}
The leading error in $\Delta t$ in terms of a trace distance was established in~\cite{rebentrost2018quantum}
as
\begin{equation} \label{eq:low_rank_error}
    \epsilon_0 \leq  2  [ \lVert A \rVert_{max} \,   \Delta t]^2,
\end{equation}
which depends on the absolute largest entry of the matrix $\lVert A \rVert_{max}$.
We will use this approach to simulate the dynamcis under our rank-1 matrix as
$A/N \propto | \psi \rangle \langle \psi | $.
\end{proof}

\subsection{Adiabatic Computation}\label{section:adiabatic_background}
Given an initial Hamiltonian $\H(0)=\H_0$ whose ground state is easy to prepare, and a final Hamiltonian $\H(1) = \H_1$ whose ground state we would like to prepare, we define the time-dependent Hamiltonian $\H(s)$, with $0\le s \le 1$.
If we evolve the state according to $\H(s)$ ``sufficiently slowly'', the adiabatic theorem guarantees that the state will remain in the instantaneous ground state of $\H(s)$ throughout the evolution, and so at the end of the procedure we will have obtained the ground state of $\H_1$. For an insightful discussion of adiabatic quantum computation, see the paper by van Dam \textit{et al.}~\cite{van2001powerful}, or the comprehensive review by Albash \textit{et al.}~\cite{albash2018adiabatic}. As introduced by Farhi \textit{et al.}~\cite{farhi2000quantum}, $\H(s)$ is often defined as an interpolation between $\H_0$ and $\H_1$ in the form $\H(s)=(1-s)\H_0 + s\H_1$, however, this generally need not be the case.
Precisely, we perform the evolution obtained by solving the Shr\"odinger equation,
\begin{align}
    \frac{d}{ds}\ket{\psi(s)} = -i\H(s)\ket{\psi(s)},
\end{align}
where we select $\ket{\psi(s)}$ such that $\ket{\psi(0)}$ is the ground state of $\H_0$. Then, our objective is to obtain a circuit performing the evolution under the time-dependent Hamiltonian $\H(s)$. Of note, the adiabtic evolution is considered sufficiently slow if throughout the evolution the adiabatic schedule $\tau(s)$ satisfies,
\begin{align}\label{eq:gap_def}
    \frac{
    \lVert   \frac{d}{ds}\H(s)  \rVert_2}{g^2(s)}
    \ll
    \tau(s),
\end{align}
where $g(s)$ is the spectral gap of $\H(s)$ (i.e. the difference in the two lowest energy levels), and $\lVert \frac{d}{ds}\H(s)  \rVert_2$ is the spectral norm of the derivative of our time-dependent Hamiltonian. The spectral norm of a matrix $M$ is defined as $\max_{||x||_2=1}||Mx||_2$. In  the simplest implementation it is sufficient to use a constant schedule throughout, defined as
\begin{align}
    \max_{s\in [0, 1]}\frac{\lVert   \frac{d}{ds}\H(s)   \rVert_2}{g^2_{min}}
    \ll
    \tau
\end{align}
where $g_{min}$ is the minimum spectral gap of $\H(s)$ for any time $s$. Then, for the constant schedule, the delay factor is given by $T = \int_0^1 \tau ds = \tau$ and so we have the bound $T = O\left(\max_{s\in [0, 1]}\frac{\lVert  \frac{d}{ds}\H(s)   \rVert_2}{g^2_{min}}\right)$.
As described by van Dam \textit{et al.}, when using a constant schedule, a discretized unitary transformation induced by $\H(t)$ may be written as,
\begin{align}
    U'(T) =
    \exp\left[
    -i\frac{T}{r}\H(\frac{r}{r})
    \right]
    \cdot ... \cdot
    \exp\left[
    -i\frac{T}{r}\H(\frac{1}{r})
    \right],
\end{align}
That is, we approximate the overall transformation by a sequence of $r$ time-independent evolutions, each applied for an equal amount of time $\frac{T}{r}$.
An important fact we utilize in bounding the error from such a discretization comes from Lemma 1 presented by van Dam \textit{et al.}~\cite{van2001powerful}. We summarize their lemma and describe its proof for completeness.

\begin{lemma}\label{lemma:implied_unitary_deviation_bound}
    Suppose you have two time-dependent Hamiltonians, $\H(t)$ and $\H'(t)$ such that they induce the unitaries $U(T)$ and $U'(T)$ respectively. If the spectral norm of the difference of the Hamiltonians is bounded by at most $\delta$ for all times $t$, i.e. $
    \lVert \H(t) - \H'(t)  \rVert_2 \le \delta$, then $\lVert U(T) - U'(T) \rVert_2 \le \sqrt{2T\delta}$.
\end{lemma}
\begin{proof}
Define the states $\ket{\psi'(t)}$ and $\ket{\psi(t)}$, such that $\ket{\psi'(0)}=\ket{\psi(0)}$ which are respectively obtained by evolving the starting state according to $\H'(t)$ and $\H(t)$.
We then start by computing the operator distance via the (maximal) vector norm $\lVert \left(U(T) - U'(T)\right)\ket{\psi(0)} \rVert_2 = \sqrt{\bra{\psi(0)}\left(U(T)-U'(T) \right)^{\dagger}\left(U(T) - U'(T)\right)\ket{\psi(0)}}$, we have
\[
\lVert U(T) - U'(T) \rVert_2  \le \sqrt{2 - 2\mathrm{Re} \ip{\psi'(T)}{\psi(T)}}.
\]
Here we consider how the inner product of the two states change with respect to time (and note the parametrisation
is smooth in our case), obtaining
\begin{align}
    \frac{d}{dt}\ip{\psi'(t)}{\psi(t)} = -i\bra{\psi'(t)}
    \left(\H(t) - \H'(t) \right)\ket{\psi(t)}.
\end{align}
We can then integrate both sides with respect to time from $0$ to $T$ and lower bound its absolute value
as
\[
|\ip{\psi'(T)}{\psi(T)}| 
= |1 - i\int_{0}^T \bra{\psi'(t)}\left(\H(t) - \H'(t) \right)\ket{\psi(t)}\, \mathrm{d}t|
\geq 
1- 
|\int_{0}^T \bra{\psi'(t)}\left(\H(t) - \H'(t) \right)\ket{\psi(t)} \, \mathrm{d}t|
\]
Of course, by assumption the spectral norm of the difference of the two Hamiltonians is bounded by $\delta$, and thus $|\int_{0}^T \bra{\psi'(t)}\left(\H(t) - \H'(t) \right)\ket{\psi(t)}  \, \mathrm{d}t| \le \delta T$. Substituting this bound back
we obtain
\begin{align}\label{eqn:evolution_inner_product_ldeviation_lowerbound}
|\ip{\psi'(T)}{\psi(T)}| \ge 1 - \delta T,
\end{align}
and we thus immediately obtain the bound $\lVert U(T) - U'(T) \rVert_2 \le \sqrt{2\delta T}$.
\end{proof}

\begin{figure*}[tb]
	\begin{centering}
		\includegraphics[width=0.7\textwidth]{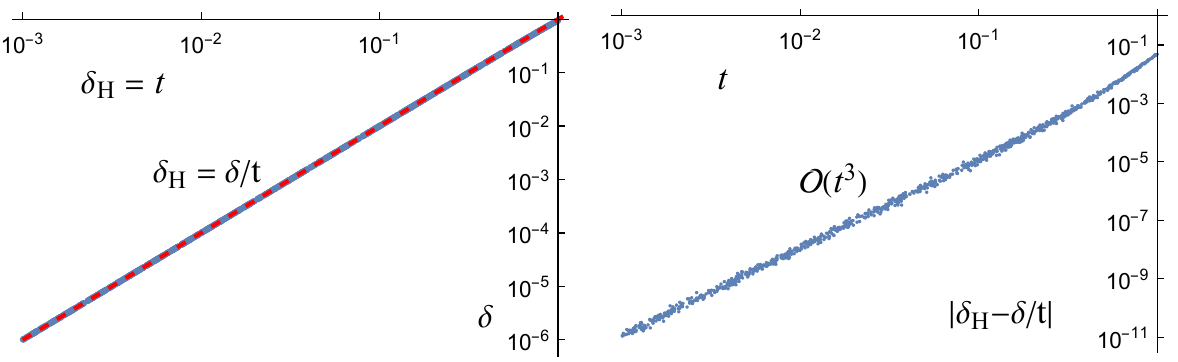}
		\caption{
		Verifying the bounds in \cref{corollary:implied_unitary_deviation_bound}.
		We randomly generate 1000 matrices (blue dots) $H'$ and $H''$ such that their distance is $t$ and we randomly select $t$ between $10^{-3}$ and $10^0$, and compute the generated unitaries $U'$ and $U''$.
		(left) The distance of two unitaries $\lVert U'- U''\rVert_2  \equiv \delta$
		is related to the  distance of the generator Hamiltonians  $\delta_H \equiv 	\lVert H'- H''\rVert_2$
via \cref{eq:norm_equivalence}
as $\delta_H   = \delta/t + O(t^2) $ (equation shown with red dashed line).
The error term in our equation is computed as
the distance $|\delta_H   - \delta/t|$
and we empirically find a better error-term scaling of $O(t^3)$.
}
	\end{centering}
\end{figure*}

\begin{corollary}\label{corollary:implied_unitary_deviation_bound}
Suppose we have the ideal unitary evolution $U':=e^{-it H'}$ for fixed time $t$ generated by the time-independent Hamiltonian
	$H'$ and we have an approximation to this unitary as $U''$
	with the error bound $\lVert U' - U'' \rVert_2 \leq \kappa t^2 $ 	for some $\kappa \geq 0$.
	The error bound on the unitaries implies that any effective Hamiltonian $H''$ that
	generates the approximate unitary $U'':=e^{-it H''}$ 
	satisfies
	$\lnorm{H' - H''} \leq \kappa t  + O(t^2)$.
\end{corollary}
\begin{proof}
	
We define the distance $\delta$ as $\delta := \lVert U' - U'' \rVert_2$,
where by formulation we know that $\delta \leq   \kappa t^2 $ for some $\kappa \geq 0$.
Given a unitary $U''$ and the fixed time $t$ there exists an effective Hamiltonian 
$\H''$ that generates $U'':=e^{-it H''}$. Given the two unitaries $U':=e^{-it H'}$ and $U'':=e^{-it H''}$ (satisfying the property $\lVert U' - U'' \rVert_2 = \delta$) generated by the time-independent Hamiltonians $H'$ and $H''$ respectively, we aim to bound the distance $\delta_H := \lVert H'- H'' \rVert_2$.
We can express this operator distance in terms of the (maximum of the) vector norm 
\begin{equation}\label{eq:scal_prod}
    \lVert U' - U'' |\psi \rangle  \rVert^2
    =
    \langle \psi | (U' - U'')^\dagger (U' - U'') |\psi \rangle 
    = 2 -2 \mathrm{Re}  \langle  \psi | (U')^\dagger U'' | \psi \rangle.
\end{equation}
We now use the Baker-Campbell-Hausdorff formula, in the form of the Zassenhaus formula~\cite{casas2012efficient} $e^{t (X+Y)}  =  e^{tX} e^{tY}   e^{-\tfrac{t^2}{2} [X,Y] } e^{\tfrac{t^3}{6} Q} \cdots$ (with $Q=2 [Y,[X,Y]]+ [X,[X,Y]]]$), which we can rearrange to obtain
\begin{equation*}
e^{tX} e^{tY}   =  	e^{t (X+Y)} \cdots  e^{  -\tfrac{t^3}{6} Q }   e^{\tfrac{t^2}{2} [X,Y] }.
\end{equation*}
We apply this formula by substituting 
$X := i H'$ and $Y:= -i H''$ as
\begin{equation*}
 (U')^\dagger U'' = e^{it H'} e^{-it H''}  =	   e^{ it (H'- H'')} \cdots e^{  -\tfrac{t^3}{6} Q }   e^{\tfrac{t^2}{2} [ H' , H''] }.
\end{equation*}
We now take the power series expansion of each exponential, keeping only terms of order $t^3$ and below,
\begin{align*}
	(U')^\dagger U''
	&=
	\Big(  1+ it (H'- H'') -\tfrac{t^2}{2} (H'- H'')^2 -i \tfrac{t^3}{6} (H'- H'')^3 \Big)
	\Big(  1 -\tfrac{t^3}{6} Q \Big) 
	\Big(  1+ \tfrac{t^2}{2} [ H', H''] \Big) 
	+\mathcal{O}(t^4)\\
	&=
	1 + 
	it (H'- H'') 
	+ \tfrac{t^2}{2} \Big( [ H' , H'']  - (H'- H'')^2  \Big)
	+
	t^3 \Xi
	+\mathcal{O}(t^4),
\end{align*}
where we have denoted the operator with factor $t^3$ as
\begin{equation*}
	\Xi:=
	-\tfrac{1}{6} Q 
	+ i \tfrac{1}{2} [ H' , H'']  (H'- H'')
	-i \tfrac{1}{6} (H'- H'')^3.
\end{equation*}
We now want to compute the expected value term from \cref{eq:scal_prod}
and  obtain
\begin{equation}
	\mathrm{Re}  \langle  \psi | (U')^\dagger U'' | \psi \rangle 
	=
	1
	- \tfrac{t^2}{2} \mathrm{Re}  \langle  \psi |  (H'- H'')^2 | \psi \rangle 
	+ t^3 \mathrm{Re}  \langle  \psi |  \Xi | \psi \rangle
		+\mathcal{O}(t^4), 
\end{equation}
using that the expected value of any Hermitian operator is real and thus
$\mathrm{Re}\, i \langle  \psi |  (H'- H'') | \psi \rangle = 0 $,
and similarly the expected value of any anti-Hermitian operator is imaginary and thus
$\mathrm{Re} \langle  \psi | [ H' , H'']  | \psi \rangle =0$.
Furthermore, rather than explicitly computing the value of the third-order term, we simply bound it as $|\mathrm{Re}  \langle  \psi |  \Xi | \psi \rangle| \in \mathcal{O}( \delta_H  )$.
Therefore, we obtain
\begin{equation*}
\mathrm{Re}  \langle  \psi | (U')^\dagger U'' | \psi \rangle
=
1- \frac{t^2}{2} \langle  \psi | (H'- H'')^2| \psi \rangle + O(t^3 \delta_H) + O(t^4)
=
1- \frac{t^2}{2} \lVert (H'- H'')| \psi \rangle\rVert^2 + O(t^3 \delta_H) + O(t^4).
\end{equation*}
We substitute this back and obtain our bound as
\begin{equation*}
    \lVert U' - U'' |\psi \rangle  \rVert^2
    =
    t^2 \lVert (H'- H'')| \psi \rangle\rVert^2
    + O(t^3 \delta_H) + O(t^4)
\end{equation*}
where we now divide by $t^2$ and take the square root 
\begin{equation*}
	\lVert U' - U'' |\psi \rangle  \rVert/t
	=
    \sqrt{\lVert (H'- H'')| \psi \rangle\rVert^2
    + O(t \delta_H) +O(t^2)}
\end{equation*}
which can be computed using the expansion $\sqrt{a + \epsilon } = \sqrt{a} + \mathcal{O}(\epsilon)$
as
\begin{equation*}
	\lVert U' - U'' |\psi \rangle  \rVert/t
	=
	\lVert (H'- H'')| \psi \rangle\rVert
		+ O(t \delta_H) +O(t^2).
\end{equation*}
Given that this holds for any state $|\psi \rangle$, the above equation implies the equivalence of the norms as
\begin{equation*}
	\delta_H  = \delta /t +  O(t \delta_H) + O(t^2),
\end{equation*}
where we have substituted our expressions for the norms as $\delta \equiv \lVert U'- U''\rVert_2$
and  $\delta_H \equiv 	\lVert H'- H''\rVert_2$.
Furthermore, we conclude by re-expressing the additive error as a multiplicative error $\delta_H +  O(t \delta_H) = \delta_H  (1 +  O(t))$
as
\begin{equation}\label{eq:norm_equivalence}
	\delta_H   = \Big( \delta/t + O(t^2) \Big)  \Big( 1 +  O(t) \Big)
	=
	 \delta/t + \delta + O(t^2) 
	\leq 
	\kappa t  + O(t^2) ,
\end{equation}
where in the last equation
we substituted our upper bound
on the unitary distance as
$\delta \leq   \kappa t^2 $.

\end{proof}

\begin{remark}
   \cref{corollary:implied_unitary_deviation_bound} bounds the maximum deviation of the Hamiltonians generating unitaries with bounded distance. However, this bound only holds for small $t<1$, and such a general statement is not possible for non asymptotic small $t$, as the following counterexample proves. Suppose we have $H' = \mathrm{diag}(1,\lambda_2, \dots \lambda_d )$ and $H'' = \mathrm{diag}(2,\lambda_2, \dots \lambda_d )$ which share the same eigenvalues except for one eigenvalue and thus the distance is $\lVert H' - H'' \rVert_2 = 1$. However, the induced unitaries at $t=2\pi$ are identical as $U' = \mathrm{diag}(1,e^{-i 2\pi \lambda_2}, \dots e^{-i 2\pi \lambda_3} ) =  U''$. 
\end{remark}

\begin{lemma}\label{lemma:double_deviation_induced_unitary_bound}
Suppose you have Hamiltonians $\H(t)$, $\H'(t)$ and $\H''(t)$, which induce the unitary transformations $U(T)$, $U'(T)$ and $U''(T)$ respectively, and that $\lnorm{\H(t)-\H'(t)} \le \delta_0$ and $\lnorm{\H'(t)-\H''(t)} \le \delta_1$. Then,
$\lnorm{U(T) - U''(T)} \le \sqrt{2T(\delta_0 + \delta_1)}$.
\end{lemma}
\begin{proof}
First, by the triangle inequality,
\begin{align}
    \lnorm{\H(t) - \H''(t)} = \lnorm{\H(t) - \H'(t) + \H'(t) - \H''(t)} \le \lnorm{\H(t) - \H'(t)} + \lnorm{\H'(t) - \H''(t)} \leq \delta_0 + \delta_1.
\end{align}
Applying Lemma~\ref{lemma:implied_unitary_deviation_bound} immediately then gives the result $\lnorm{U(T) - U''(T)} \le \sqrt{2T(\delta_0 + \delta_1)}$.
\end{proof}

\section{General distribution loading through adiabatic evolution: direct state preparation on $n$ qubits\label{sec:direct_state_prep}}

Let us now compactly present the main result of the present section only in the asymptotic limit $N \rightarrow \infty$.

\begin{theorem}\label{theorem:total_query_complexity}
Asymptotically in $N$, for an efficiently computable continuous function $f: [a,b] \mapsto \mathbb{C}$, the direct state preparation algorithm prepares the $n$ qubit state (with $N=2^n$), $\ket{\psi} = \frac{1}{\mathcal{N}(1)}\sum_{j=0}^{N-1}f(x_j)\ket{j}$ (with $f$ replaced with $g$ in the case of the integral encoding, as defined in \cref{def:integrable_fct}), up to error $\epsilon$ with query complexity bounded as
\begin{align*}
    O\left(
        \frac{\mathcal{F}^p}{\epsilon^2}
    \right),
\end{align*}
where $\mathcal{F}$ is the filling ratio from~\cref{def:peakedness}, and $p=-2$ if using the integral function encoding, and $p=-4$ if using the point-wise function encoding. The error $\epsilon$ is the deviation from an ideal unitary in terms of a spectral distance. Moreover, the algorithm uses $O(n + d)$ ancillary qubits (where $d$ is the number of digits used in the discretization of $f$), and has a probability of failure bounded by
\begin{align*}
    O\left(
        \epsilon^2
    \right).
\end{align*}
Thus the probability of failure decreases as we increase the target precision $\epsilon$, a feature that is highly advantageous.
\end{theorem}
\begin{proof}
From \cref{lemma:hamiltonian_simulation_query_complexity} we know that the algorithm has query complexity $O(r)$, and uses $O(n + d)$ ancillary qubits. 
From \cref{cor:total_error_deviaton}, we have that $\epsilon \in O (\sqrt{\mathcal{F}^p/r} )$, and so it immediately follows that to obtain an error $\epsilon$ we must have $r \in O(\mathcal{F}^p/\epsilon^2)$, yielding the query bound. Finally, from \cref{lemma:prob_failure_bound} we have that the probability of failure is bounded by $O(\mathcal{F}^p/r)$, which immediately gives the probability failure bound of $O(\epsilon^2)$ when substituting our above upper bound for $r$
as a \emph{lower bound}, i.e., we use $r$ larger than the upper bound to guarantee a worst-case error $\epsilon$.
Similar statements can be made without assuming large $N$ by directly combining the corresponding non-asymptotic lemmas
from \cref{sec:non-asymptotic}.
\end{proof}

Note that in \cref{sec:non-asymptotic} we will indeed derive explicit upper bounds for arbitrary finite qubit count $n$
and above we have only stated the main result in the limiting scenario for brevity.
Furthermore, we note that while the above  asymptotic limits exist for any continuous functions, in the specific case of Lipschitz continuous
functions (which cover nearly all scenarios in practice)
these limits are approached rapidly in exponential order in $n$. For commonly encountered non-pathological functions, the asymptotic limit is usually reached with $n<30$.

\begin{remark}\label{remark:comment_on_knowing_filling_ratio}
When implementing the approach in practice explicit knowledge of the filling ratio $\mathcal{F}$ is not actually required.
The above result merely establishes the scaling of the number of timesteps 
$r$ required to obtain an error $\epsilon$, and note that in practice a bound on $r$ is straight-forward to obtain by combining Lemma~\ref{lemma:simulation_algorithmic_error_bound} and Remarks~\ref{remark:asymptotically_constant_ratio} and~\ref{remark:easy_rescaling_remark}
as long as one can either bound or over-estimate the maximum value of $f$ -- which indeed bounds $\fillr$.
We present the bounds in terms of $\mathcal{F}$ because it makes the intuition clear as to when the constant cost of the algorithm will be large or small.
\end{remark}

\begin{remark}
    In general, while not necessary, QRAM can provide an asymptotic improvement in the circuit depth by alleviating the need for most of the quantum arithmetic involved in implementing the oracle (e.g. by storing $f_0$ and $f_1$ for all values in the grid, and then using basic addition and multiplication to compute $f_s$ as needed). 
    Moreover, in certain QRAM architectures, (namely where each quantum register in the architecture has its own controlling classical mini-cpu and batch instructions can be issued to all the classical controllers in parallel) the cost of the oracle access can be made $O(1)$, thus rendering the algorithm's circuit complexity the same as its query complexity. 
\end{remark}

\subsection{The procedure\label{sec:direct_prep_procedure}}
In the standard model of adiabatic quantum computation originally proposed by Farhi \textit{et al.} in 2000~\cite{farhi2000quantum}, one defines a time-dependent interpolated Hamiltonian $H(s)=(1-s)H_0 + s H_1$ where $H_0$ is some initial Hamiltonian with an easy to prepare ground state, and $H_1$ is a final Hamiltonian whose ground state we wish to prepare. One then evolves the initial state according to $H(s)$ to obtain the final ground state of the problem. However, a challenge with this approach is that the spectral gap can become exponentially small in general (thus necessitating exponential evolution time).

However, in some cases this problem can be circumvented by viewing the paradigm of adiabatic quantum computation more generally -- we need not limit ourselves to using a time-dependent Hamiltonian of the form $H(s)=(1-s)H_0 + s H_1$. 
In our case, we begin at $n$ qubits (i.e. use a resolution of $N=2^n$) and define a time-dependent \textit{rank-one} Hamiltonian. That is, rather than interpolating between two separate Hamiltonians, we define our Hamiltonian as proportional to a projector onto a discretized time-dependent quantum state the encodes a parametrised function. Precisely, we define
\begin{align}\label{eq:param_hamil_def}
\mathcal{H}(s) := - A(s)/N
\quad \quad
    \text{pointwise: \quad \quad  $A(s)$ encodes $f_s$ as }
    \quad \quad
    & f_s :=  (1-s) f_0
    + s f_1.
    \\
    \text{integral: \quad \quad  $A(s)$ encodes $g_s$ as }
    \quad \quad
     & g_s :=  (1-s) g_0
    + s g_1.\nonumber
\end{align}
Here $A(s)$ is a rank-one matrix with entries $[A(s)]_{kl}:= f_s(x_k) f^*_s(x_l)$ ($[A(s)]_{kl}:= N g_s(x_k) g_s(x_l)$).
Furthermore, $f_0(x)$ and  $g_0(x)$ are trivial functions, such as the uniform distribution, that we can easily prepare as our initial state
while  $f_1(x)$ and $ g_1(x)$ are our final desired functions.
This definition of $\H(s)$ has numerous beneficial properties, which enable general analytic performance guarantees, which we now derive. 
As a matter of notation, we define the normalization constant $\mathcal{N}(s)$ as
\begin{align}\label{eq:normalisation}
    \mathcal{N}(s) := \sqrt{\sum_{j=0}^{N-1}|f_s(x_j)|^2},
\end{align}
for the interpolated function $f_s$ (and similarly for the integration case $g_s$) as per \cref{eq:param_hamil_def}. $\mathcal{N}(1)$ is the main quantity we are concerned with in our proofs, and we derive all of our general bounds in terms of it (noting that as per Remark~\ref{remark:easy_rescaling_remark} it grows as $O(\sqrt{N})$).

In the case of the pointwise state preparation approach, while it is not a limitation of the presented algorithm, it is important to keep in mind that we define our problem such that we aim to represent a continuous function over a finite grid -- which is only meaningful if the grid is fine enough to not miss features of the function. While we prove asymptotic properties for an increasing resolution $N \rightarrow \infty$, the approach might fail in ill-defined cases when $N$ is too small. In particular, if the function $f_1$ consists of only narrow peaks of width $\epsilon$ then we need at least $n \in O[ \log_2(\epsilon^{-1}) ]$ qubits to resolve these peaks and for qubit counts below this threshold our discretisation may in a worst-case scenario correspond to the null-vector (i.e., the function value is $0$ at all grid points). A particular example could be the function $f(x) = \sin(2^{20} \pi x )$ which requires at least $21$ qubits of resolution in the range $0\leq x \leq 1$ using our standard grid.
It is important to stress that this is not a limitation of the algorithm being presented -- when preparing a function according to point-wise samples no algorithm can handle such ill-defined cases.
However, in practice this is unlikely to ever be a problem, because of the exponential growth of the grid in the number of qubits. For example, at $128$ qubits, $N$ is sufficiently large to resolve a function defined on a kilometer interval with grid points at the Planck scale.

Let us now summarize the algorithm, along with the main proofs. 
\begin{itemize}
	\item We begin by preparing the state $\ket{\psi_0} = \frac{1}{\mathcal{N}(0)}\sum_{j=0}^{N-1}f_0(x_j)\ket{j}$, which by definition of $f_0$ is easy to prepare.
	\item We then implement adiabatic evolution for a total time $T$ according to the time-dependent rank-one Hamiltonian $\H(s)$ as defined in Lemma~\ref{lemma:normalisation}.
	\item We prove that $T$ is constant bounded in Lemma~\ref{lemma:total_evolution_time} (to see that this is indeed a constant bound, consider Remark~\ref{remark:easy_rescaling_remark}).
	\item The adiabatic evolution is implemented by the sequence of time-independent unitary transformations given by $e^{-i \tfrac{T}{r}\H(\frac{r}{r})}e^{-i \tfrac{T}{r}\H(\frac{r-1}{r})}\cdots e^{-i \tfrac{T}{r}\H(\frac{1}{r})}$.
	\item We prove a bound on the error incurred by this discretization of the adiabatic evolution in Lemma~\ref{cor:discret_error}.
	\item In Lemmas~\ref{lemma:low_rank_sim_error} and~\ref{lemma:simulation_algorithmic_error_bound} we demonstrate how each of the time-independent terms $e^{-i \tfrac{T}{r}\H(\frac{j}{r})}$ may be implemented using the low-rank Hamiltonian simulation procedure, and bound the error and probability of failure incurred in the simulation.
	\item In \cref{lemma:prob_failure_bound}, we bound the cumulative probability of failure from implementing the low-rank Hamiltonian simulation technique $r$ times in a row (in $e^{-i \tfrac{T}{r}\H(\frac{r}{r})}e^{-i \tfrac{T}{r}\H(\frac{r-1}{r})}\cdots e^{-i \tfrac{T}{r}\H(\frac{1}{r})}$).
	\item In Lemma~\ref{lemma:joint_error_bound} we bound the total error incurred in the procedure. 
	\item While we present these bounds for an arbitrary, finite qubit count $n$ in \cref{sec:non-asymptotic} without assuming continuity of the functions, we establish asymptotic limits in \cref{sec:asymptotic} using properties of continuous functions.
	\item Finally, in \cref{sec:generalisation} we establish that our approach is efficient not only for continuous but also for arbitrary functions (mappings) that satisfy certain properties.	
	\end{itemize}

\subsection{Proofs for finite qubit count $n$ \label{sec:non-asymptotic}}

As we noted above we are primarily concerned with the normalisation factor from
\cref{eq:normalisation} when proving performance properties of the adiabatic approach.
Let us here point out that we can efficiently estimate this normalisation constant.

\begin{remark}[Efficiently computable normalisation]\label{remark:easy_rescaling_remark}
As demonstrated in Appendices~\ref{subsec:estimating_normalization_factor} and~\ref{subappendix:estimating_normalization_factor_v2}, it is possible to efficiently compute the normalization factor for an arbitrary $f_1$, and thus to re-scale $f_1$ such that $\mathcal{N}(1) = \sqrt{N}$. 
\end{remark}
The above property will ensure that our following upper bounds on the performance of adiabatic evolution are simplified and tight. In the following subsection we first prove that the adiabatic evolution needs to run only for a constant bounded time $T$ by bounding the delay factor. The main results are illustrated in \cref{fig:adiabatic_bound_plot} (left). In in the second subsection we prove bounds on the algorithmic errors incurred by our procedure. The algorithmic error of our adiabatic evolution is illustrated in \cref{fig:adiabatic_bound_plot} (right).

\subsubsection{Proofs on boundedness of the spectral gap}

\begin{figure*}[tb]
	\begin{centering}
		\includegraphics[width=0.7\textwidth]{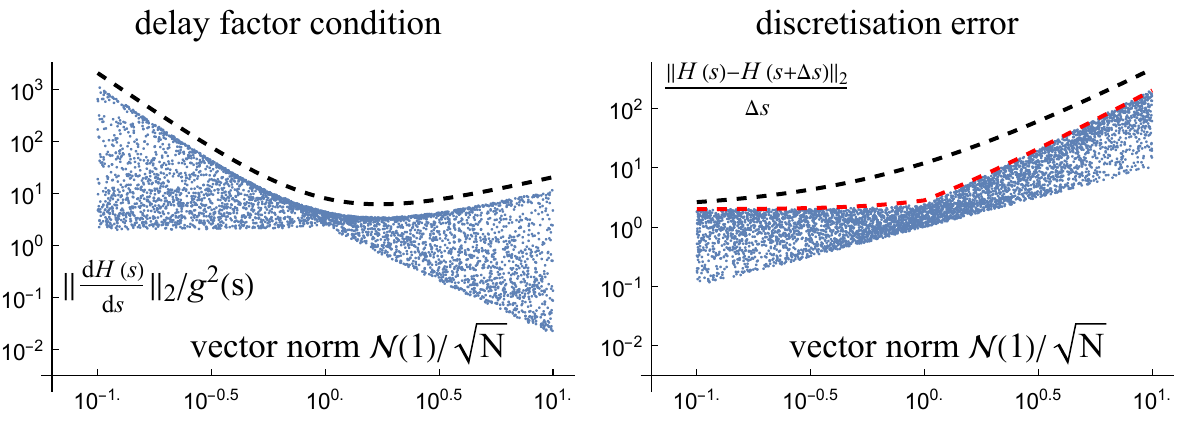}
		\caption{
	    Numerically verifying bounds that guarantee the efficient performance of adiabatic evolution from \cref{sec:non-asymptotic}.
	    (left)
	    The energy gap $g(s)$ of our rank-1 Hamiltonian $H(s)=-A(s)/N$ from \cref{eq:param_hamil_def} is determined by the
	    normalisation $\mathcal{N}(s)$ from \cref{eq:normalisation} and is bounded via \cref{lemma:normalisation}.
		The 
	    matrix norm of the derivative $\lVert \frac{d}{ds}\H(s) \rVert_2/g^2(s)$ relative to this spectral gap determines the overall evolution time $T$ required for adiabatic evolution.
	    Upper bound (dashed black line) from \cref{lemma:total_evolution_time}.
	    (right)
	    The matrix distance $ \lVert H(s) - H'(s) \rVert_{2} \propto \Delta s$ determines the error when approximating a continuous adiabatic evolution via a piecewise constant evolution.
	    We note that with similar proof techniques we could derive tighter asymptotic (in $\Delta s$, not presented in this work) upper bounds (dashed red line) but we prefer the upper bound (dashed black line) in \cref{cor:discret_error} due to its more compact expression.
	    Given \emph{all bounds} in \cref{sec:non-asymptotic} are valid to arbitrary rank-1 matrices $A$ (i.e., we need only invoke continuity of functions when computing asymptotic limits in \cref{sec:asymptotic}), we simulated $5000$ unitary Haar-random discretisation 
	    vectors $v_1$ at $n=10$ qubits for uniformly randomly selected $0 \leq s \leq 1$ and logarithmically uniformly randomly selected vector norms $\mathcal{N}$.	\label{fig:adiabatic_bound_plot}
		}
	\end{centering}
\end{figure*}

\begin{lemma}[Spectral gap bound]\label{lemma:normalisation}
We choose the phase of our trivial initial function  properly via the efficient verification protocol in \cref{sec:verification} as one of the four possibilities $f_0(x) = (\pm 1 \pm i)/\sqrt{2}$.
Given a parameterised rank-1 encoding $A(s)$ from Eq.~\eqref{eq:param_hamil_def}, the spectral gap of our Hamiltonian $\mathcal{H}(s) = - A(s)/N =: |v_s \rangle \langle v_s|$ is determined by the vector norm $g(s) = ||\ket{v_s}||_2^2$ and is bounded as
\begin{equation*}
    g(s) \ge \frac{\mathcal{N}(1)^2/N}{\mathcal{N}(1)^2/N + 1}.
\end{equation*} 
\end{lemma}
\begin{proof}
First, observe that since $\H(s)$ is rank-one, its spectral gap is the same as its spectral norm, and thus
\begin{align}
    g(s) = \lVert H(s) \rVert_2
    =
    \left\lVert \frac{A(s)}{N} \right\rVert_2 
    =
    \left\lVert \op{v_s}{v_s} \right\rVert_2 
    =
    \left\lVert \ket{v_s} \right\rVert_2^2
    .
\end{align}
Moreover, the definition for $\H(s)$ in Equation~\ref{eq:param_hamil_def} is equivalent to the following,
\begin{align}\label{eqn:expanded_state_from_projector}
    \ket{v_s} = (1-s)\ket{v_0} + s\ket{v_1},
\end{align}
where $\ket{v_0} = \frac{1}{\sqrt{N}}\sum_{j=0}^{N-1}f_0(x_j)\ket{j}$ and $\ket{v_1} = \frac{1}{\sqrt{N}}\sum_{j=0}^{N-1}f_1(x_j)\ket{j}$ (and thus $\ket{v_0}$, $\ket{v_1}$ and $\ket{v_s}$ are not normalized). In the integration case, simply replace each $f(x)$ with the appropriate integral of $f$ (as per Equation~\ref{eq:encoding}). Then, using the Equation~\ref{eqn:expanded_state_from_projector},
\begin{align}\label{eq:spec_gap}
    g(s) = \left\lVert \ket{v_s} \right\rVert_2^2
    = (1-s)^2\ip{v_0}{v_0} + s^2 \ip{v_1}{v_1} + 2s(1-s)\text{Re}\ip{v_0}{v_1}.
\end{align}
Noting that $\ip{v_0}{v_0} = \frac{1}{N}\mathcal{N}(0)^2$, $\ip{v_1}{v_1} = \frac{1}{N}\mathcal{N}(1)^2$, and that by selecting the appropriate phase and sign in $f_0(x) = (\pm 1 \pm i)/\sqrt{2}$ we ensure that $2\text{Re}\ip{v_0}{v_1} \ge 0$, we get
\begin{align}
    g(s) \ge (1-s)^2\frac{\mathcal{N}(0)^2}{N} + s^2 \frac{\mathcal{N}(1)^2}{N}.
\end{align}
Noting that it is trivial to select an $f_0$ (or $g_0$) such that $\frac{\mathcal{N}(0)}{\sqrt{N}}=1$, the above simplifies to
\begin{align}
    g(s) \ge (1-s)^2 + s^2 \frac{\mathcal{N}(1)^2}{N}.
\end{align}
Minimizing the expression with respect to $s$, we find the minimum at $s=\frac{1}{\mathcal{N}(1)^2/N + 1}$. Substituting this value for $s$, we then obtain the bound
\begin{align}\label{eqn:direct_spectral_gap}
    g(s) \ge \frac{\mathcal{N}(1)^2}{\mathcal{N}(1)^2 + N}.
\end{align}
\end{proof}

As an aside, $2\text{Re}\ip{v_0}{v_1} \ge 0$ only holds if we have properly chosen the phase of our initial trivial function $f_0(x) = (\pm 1 \pm i)/\sqrt{2}$. If we select the phase incorrectly, it is possible that $\text{Re}\sum_{j}f_0(x_j)f_1(x_j) < 0$, which negative term can decrease the spectral gap in \cref{eq:spec_gap} below our lower bound (which assumes a non-negative $\text{Re}\sum_{j}f_0(x_j)f_1(x_j)$). 
Indeed for real functions we need only consider the sign $f_0(x) = \pm 1$ while for the integral encoding we trivially set $f_0(x) =1$. Note that we can efficiently decide the sign (and phase) of the trivial function given we need only test the four possibilities for a phase  and we can efficiently verify the result of our state-preparation procedure via the destructive interference circuit in Fig.~\ref{fig:destructive} as discussed in \cref{sec:verification}. 

\begin{corollary}
By Lemma~\ref{lemma:normalisation}, if we re-scale $f_1$ as per Remark~\ref{remark:easy_rescaling_remark}, the bound on the spectral gap of the Hamiltonian may be written as,
\begin{align*}
    g(s) \ge \frac{1}{2},
\end{align*}
which is indeed a constant. 
\end{corollary}

\begin{lemma}[Delay factor bound]\label{lemma:total_evolution_time}
Given the Hamiltonian $\H(s)$ as defined in Lemma~\ref{lemma:normalisation}, we have the bound
\begin{align}
    \frac{\left\lVert \frac{d}{ds}\H(s) \right\rVert_2}{g(s)^2}  \le 2\frac{\left(\mathcal{N}(1)^2/N + 1\right)^2}{\left( \mathcal{N}(1)/\sqrt{N}\right)^3}.
\end{align}
This implies that the total adiabatic evolution time required for the procedure is asymptotically given by $T\in O\left(\frac{\left(\mathcal{N}(1)^2/N + 1\right)^2}{\left(\mathcal{N}(1)/\sqrt{N}\right)^3} \right)$.
\end{lemma}

\begin{proof}
We begin by providing an upper-bound for $\left\lVert \frac{d}{ds}\H(s) \right\rVert_2$.
\begin{align}
    \left\lVert \frac{d}{ds}\H(s) \right\rVert_2
    =
    \left\lVert \left(\frac{d}{ds} \ket{v_s}\right)\bra{v_s} + \ket{v_s}\left(\frac{d}{ds} \bra{v_s}\right)\right\rVert_2
    =
    \left\lVert \left(\ket{v_1}-\ket{v_0}\right)\bra{v_s} + \ket{v_s}\left(\ket{v_1}-\ket{v_0}\right)\right\rVert_2.
\end{align}
For ease of notation, allow $\ket{v'} = \ket{v_1} - \ket{v_0}$.
Using the triangle inequality,
\begin{align}\label{eqn:bounding_derivative_of_hs_intermediate}
    \left\lVert \op{v'}{v_s} + \op{v_s}{v'}\right\rVert_2
    \le 
    \left\lVert \op{v'}{v_s} \right\rVert_2
    +
    \left\lVert \op{v_s}{v'}\right\rVert_2
    =
    2\left\lVert \op{v'}{v_s} \right\rVert_2.
\end{align}
Suppose we have two normalized vectors $\ket{a'} = \frac{\ket{a}}{\sqrt{\ip{a}{a}}}$ and  $\ket{b'} = \frac{\ket{b}}{\sqrt{\ip{b}{b}}}$. Then, 
\begin{align}
    \lVert \op{a}{b} \rVert_2 = \sqrt{\ip{a}{a}} \sqrt{\ip{b}{b}} \lVert\op{a'}{b'} \rVert_2 \le \lVert \ket{a} \rVert_2 \lVert \ket{b} \rVert_2,
\end{align}
where the inequality follows from $\lVert  \op{a'}{b'} \rVert_2 \le 1$ for any normalized vectors $\ket{a'}$ and $\ket{b'}$. Applying this to Equation~\ref{eqn:bounding_derivative_of_hs_intermediate} yields,
\begin{align}\label{eqn:bounding_derivative_of_hs_intermediate2}
    \left\lVert \frac{d}{ds}\H(s) \right\rVert_2
    \le 
    2\lVert \ket{v'} \rVert_2 \lVert \ket{v_s} \rVert_2.
\end{align}
Of course,
\begin{align}
    \lVert \ket{v'} \rVert_2^2
    = \ip{v'}{v'}
    = \ip{v_0}{v_0} + \ip{v_1}{v_1} - 2\text{Re}\ip{v_0}{v_1}
    \le 
    \lVert \ket{v_0} \rVert_2^2 + \lVert \ket{v_1} \rVert_2^2,
\end{align}
where the inequality follows from the fact that $\text{Re}\ip{v_0}{v_1}\ge 0$ given we have
properly chose the phase of the trivial function. Thus,
\begin{align}
    \lVert \ket{v'} \rVert_2 \le \sqrt{\frac{\mathcal{N}(0)^2}{N}+\frac{\mathcal{N}(1)^2}{N}}
    =
    \sqrt{1 + \frac{\mathcal{N}(1)^2}{N}},
\end{align}
where we again used the fact that its trivial to select an $f_0$ (or $g_0$) such that $\mathcal{N}(0) = \sqrt{N}$. Equation~\ref{eqn:bounding_derivative_of_hs_intermediate2} then becomes,
\begin{align}
    \left\lVert \frac{d}{ds}\H(s) \right\rVert_2
    \le 
    2\sqrt{1 + \frac{\mathcal{N}(1)^2}{N}} \lVert \ket{v_s} \rVert_2.
\end{align}
Noting that $g(s) = ||\ket{v_s}||_2^2$, we can then write
\begin{align}
    \frac{\left\lVert \frac{d}{ds}\H(s) \right\rVert_2}{g(s)^2}
    \le 
    \frac{2\sqrt{1 + \frac{\mathcal{N}(1)^2}{N}} }{||\ket{v_s}||_2^3}.
\end{align}
From Lemma~\ref{lemma:normalisation}, we directly obtain,
\begin{align}
    \lVert \ket{v_s} \rVert_2^3
    \ge 
    \frac{\mathcal{N}(1)^3}{(\mathcal{N}(1)^2 + N)^{3/2}}.
\end{align}
Thus,
\begin{align}
    \frac{\left\lVert \frac{d}{ds}\H(s) \right\rVert_2}{g(s)^2}
    \le 
    2\frac{\sqrt{N + \mathcal{N}(1)^2}}{\sqrt{N}}
    \frac{(\mathcal{N}(1)^2 + N)^{3/2}}{\mathcal{N}(1)^3}
    =
    2\frac{\left(\mathcal{N}(1)^2 + N\right)^2}{\sqrt{N}\mathcal{N}(1)^3}.
\end{align}
As discussed in Section~\ref{section:adiabatic_background}, a sufficiently slow Hamiltonian evolution schedule is given by 
\begin{align}
    \tau(s) \gg \frac{\left\lVert \frac{d}{ds}\H(s) \right\rVert_2}{g(s)^2}.
\end{align}
Thus, $T=\int_0^1 \tau(s) ds$, and so asymptotically we also have the bound $T\in O\left(2\frac{\left(\mathcal{N}(1)^2 + N\right)^2}{\sqrt{N}\mathcal{N}(1)^3} \right)$.
\end{proof}

\begin{corollary}
By Lemma~\ref{lemma:total_evolution_time}, if we re-scale $f_1$ as per Remark~\ref{remark:easy_rescaling_remark}, the bound on the delay factor of $\H(s)$ is given by,
\begin{align*}
    \frac{\left\lVert \frac{d}{ds}\H(s) \right\rVert_2}{g(s)^2}  \le 8.
\end{align*}
Thus, the total adiabatic evolution time required by the procedure is asymptotically constant, i.e. $T\in O(1)$.
\end{corollary}

\subsubsection{Time evolution error bounds}

\begin{lemma}[Adiabatic discretization error bound]\label{cor:discret_error}
Given the parameterized Hamiltonian $\mathcal{H}(s)$ from Eq.~\eqref{eq:param_hamil_def} and its discretised approximation $\mathcal{H}'(s) = \mathcal{H}(\ceil{s r}/r)$ with a resolution $r\in \mathbb{N}$, the distance between the approximate and exact Hamiltonians is bounded as (in terms of the spectral norm) 
\begin{align*}
    \lVert \mathcal{H}(s) - \mathcal{H}'(s) \rVert_{2}
    &\leq \delta_0 \equiv 
	\frac{2}{r}\left(1 + 3\frac{\mathcal{N}(1)}{\sqrt{N}} + 2\frac{\mathcal{N}(1)^2}{N} \right) .
\end{align*}
This may be equivalently written as,
\begin{align*}
    \lVert \mathcal{H}(s) - \mathcal{H}'(s) \rVert_{2}
    &\leq
	\frac{2}{r}\left(1 + 3\lnorm{A(1)/N}^{1/2} + 2\lnorm{A(1)/N} \right).
\end{align*}
\end{lemma}
\begin{proof}
First, observe that the deviation between the exact and approximate Hamiltonians is maximized when $s=\frac{j+1}{r}$ (for $j\in \mathbb{Z}$ and $1\le j \le r$), as $\H'(\frac{j+1}{r}) = \H(\frac{j}{r})$.
Thus,
\begin{align}
    \lnorm{\H'(s) - \H(s)}
    \le 
    \lnorm{\H\left(\frac{j}{r}\right) - \H\left(\frac{j+1}{r}\right)}.
\end{align}
For notational brevity, allow $X_0 = \op{v_0}{v_0}$, $X_1 = \op{v_0}{v_1} + \op{v_1}{v_0}$, and $X_2 = \op{v_1}{v_1}$. Then,
\begin{align}
    \H(s) = \op{v_s}{v_s} = (1-s)^2X_0 + s(1-s)X_1 + s^2X_2.
\end{align}
Allowing $s=\frac{j+1}{r}$ and $\gamma = \frac{j}{r}$, simple computer algebra shows
\begin{align}
    \lnorm{\H(\gamma) - \H(s)}
    &=
    \lnorm{\left(\frac{1 + 2j - 2r}{r^2}\right)
    \left(
    X_0 - X_1 + X_2
    \right) 
    + \frac{1}{r}\left(X_1 + 2X_2 \right)}.
\end{align}
Then,
\begin{align}
    \lnorm{\H'(s) - \H(s)}
    &\le
    \frac{1}{r^2}\lnorm{X_0 - X_1 + X_2}
    + 
    \frac{2}{r}\lnorm{X_0 - X_1 + X_2}
    +
    \frac{1}{r}\left(\lnorm{X_1} + 2\lnorm{X_2} \right)\\
    &\le
    \frac{1}{r^2}\left(\lnorm{X_0}
    +\lnorm{X_1} + \lnorm{X_2}
    \right)
    + 
    \frac{1}{r}\left(2\lnorm{X_0} + 3\lnorm{X_1} + 4\lnorm{X_2} \right)
\end{align}
where we used the triangle inequality, and the fact that $|\frac{1 + 2j - 2r}{r^2}| \leq -(1-2r)/r^2$ (since $r\ge 1$).
Noting that $\lnorm{X_0} = 1$  , $\lnorm{X_1} \le 2\frac{\mathcal{N}(1)}{\sqrt{N}}$, $\lnorm{X_2} = \frac{\mathcal{N}(1)^2}{N}$, we obtain
\begin{align}
    \lnorm{\H'(s) - \H(s)}
    &\le
    \frac{2}{r}\left(1 + 3\frac{\mathcal{N}(1)}{\sqrt{N}} + 2\frac{\mathcal{N}(1)^2}{N} \right)
    -
    \frac{1}{r^2}\left(
        1
        +
        2\frac{\mathcal{N}(1)}{\sqrt{N}}
        +
        \frac{\mathcal{N}(1)^2}{N}
    \right).
\end{align}
Finally, using $\lnorm{A(1)/N} = \mathcal{N}(1)^2/N$, we obtain
\begin{align}
    \lnorm{\H'(s) - \H(s)}
    &\le
    \frac{2}{r}\left(1 + 3\lnorm{A(1)/N}^{1/2} + 2\lnorm{A(1)/N} \right)
    -
    \frac{1}{r^2}\left(
        1
        +
        2\lnorm{A(1)/N}^{1/2}
        +
        \lnorm{A(1)/N}
    \right).
\end{align}
Since all of the terms quadratic in $r$ are negative, we can then directly drop them to obtain the simplified expression for the upper-bound presented in the lemma.
\end{proof}

\begin{corollary}
By Lemma~\ref{cor:discret_error}, if we re-scale $f_1$ as per Remark~\ref{remark:easy_rescaling_remark}, the bound on the bound on the maximum deviation of the discretized and exact Hamiltonians is given by,
\begin{align*}
    \lVert \mathcal{H}(s) - \mathcal{H}'(s) \rVert_{2} \le \frac{8}{r}.
\end{align*}
\end{corollary}

So far, we have shown that the adiabatic evolution need only run for a constant period of simulation time $T$ (with the constant factor depending on properties of the function), and that the discretization of the adiabatic Hamiltonian may be performed with arbitrarily small error. We now show that the simulation of each $e^{-\H(s) \frac{T}{r}}$ term may also be performed efficiently (and with bounded error) and thus that the overall algorithm is efficient.

\begin{lemma}[Error of low-rank simulation]	\label{lemma:low_rank_sim_error}
	We apply the low-rank simulation from \cref{statement:low_rank_simulation}
	for time $\Delta t$.
	The quantum state after measuring the ancilla in the $|+^n\rangle$ state and discarding the ancilla is
	$|\psi'\rangle =	e^{-i \Delta t A/N}|\psi\rangle  + |\mathcal{E}\rangle$
	with $\lVert \mathcal{E} \rVert \leq 5/2 \Delta t^2 \lVert A \rVert_{max}^2$  and with bounded probability
	$	\mathrm{Prob} 
	\geq 
	1 - \Delta t^2 
	\lVert A \rVert_{max}^2$ expressing only leading terms in $\Delta t$.
	Given the above expression holds for any input state $|\psi\rangle$ it follows that the
	algorithmic error 
	$\lVert U(\Delta t)-U(\Delta t)' \rVert := \epsilon_0$ of the approach is bounded as
	\begin{equation*}
		\epsilon_0(\Delta t) \leq 5/2 \Delta t^2 \lVert A \rVert_{max}^2
	\end{equation*} 
	where $U'$ is the ideal, piecewise constant unitary $e^{-i\Delta t A/N}$ and $U''$ represents the mapping
	of our procedure (1-sparse evolution then Hadamard measurement). 
\end{lemma}
\begin{proof}
	Using the notations from \cref{sec:numerics}
	we explicitly compute the joint state of the ancilla and the main register
	$|\Psi(t) \rangle =
	e^{- i t S_A} | +^n \rangle |\psi\rangle$
	 up to order $t^2$ and neglecting possible digitisation error in computing entries of $S_A$
	 via an oracle (which can be suppressed arbitrarily) as
	\begin{equation*}
		|\Psi(t) \rangle 
		=  | +^n \rangle |\psi\rangle -i t S_A |\Psi(0) \rangle + \mathcal{O}(t^2)
			=
		| +^n \rangle |\psi\rangle 
		 - i t \sum_{jk} [A]_{kj} \Psi_{kj} \ket{j}\ket{k} + \mathcal{O}(t^2)
	\end{equation*}
	where we have  used the action of $S_A$ on the amplitudes $\Psi_{jk}:=( \langle j | \langle k |) (| +^n \rangle |\psi\rangle) $
	from Eq.~\eqref{eq:apply_sa_matrix}. Substitute the explicit
	form of the matrix $A_{kI} = \phi_k \phi_I^* $ and thus find the composite state as
	\begin{equation}\label{eq:joint_state}
		|\Psi(t) \rangle   =
	| +^n \rangle |\psi\rangle - i t \sum_{jk} \phi_k \phi_j^* \Psi_{kj} \ket{j}\ket{k} + \mathcal{O}(t^2)
	= | +^n \rangle |\psi\rangle -i t |\chi \rangle |\phi \rangle
	+ 
	|\epsilon\rangle,
	\end{equation}
	where we used that $\Psi_{jk} = \psi_k /\sqrt{N}$ and $\Psi_{kj} = \psi_j /\sqrt{N}$,
	and we have denoted the ancilla state
	$  |\chi \rangle  := \sum_j \phi_j^*\psi_j/\sqrt{N}
	\ket{I}$.
In \cref{eq:joint_state} we can explicitly write the residual term
	 as $|\epsilon\rangle:=  -\tfrac{t^2}{2} S_A^2| +^n \rangle |\psi\rangle +\mathcal{O}(t^3)$ and we can bound
	its vector norm as $\lVert \epsilon \rVert \leq \tfrac{t^2}{2} \lVert S_A \rVert_2^2 +\mathcal{O}(t^3)$.

	Let us now compute the quantum state that we obtain after projecting onto the $| +^n \rangle$ state on the ancilla qubit
	with $\mathcal{P}:=  |  +^n  \rangle \langle +^n| \otimes \mathrm{Id} $
	as
	\begin{equation*}
		\mathcal{P} |\Psi(t) \rangle 
		= | +^n \rangle |\psi\rangle -\frac{i t}{N} |+^n \rangle |A\psi \rangle
		+ \mathcal{P} |\epsilon \rangle,
	\end{equation*}
	where we have used that $\langle +^n | \chi \rangle = \sum_k \phi_k^*\psi_k/N 
	=
	\langle \phi | \psi \rangle/N
	$
	and note that $|A\psi\rangle := A|\psi\rangle = \langle \phi | \psi \rangle |\phi\rangle$.

	We compute the probability of this measurement by computing the norm of the projected vector as
	\begin{align*}
		\mathrm{Prob} =   \Big(\langle \Psi(t) | \mathcal{P}  \Big)   \Big(  \mathcal{P}  |\Psi(t) \rangle \Big)
		=&
		\Big(
		\langle  +^n  |\langle \psi |+  \frac{i t}{N} \langle  +^n |  \langle A \psi |
		+  \langle  \epsilon | \mathcal{P}
		\Big)
		\Big(   | +^n \rangle |\psi\rangle -\frac{i t}{N} |+^n \rangle |A \psi \rangle
		+ \mathcal{P} |\epsilon \rangle
		\Big)\\
		=&
		1 
		+ \frac{ t^2 }{N^2} \Big( \langle  +^n |   \langle A \psi |\Big) \Big( |+^n \rangle | A \psi \rangle\Big)
		+
		2 \mathrm{Re}[ \Big( \langle  +^n  |\langle \psi | \Big)  \mathcal{P} |\epsilon \rangle ]
		+\mathcal{O}(t^3).
	\end{align*}
	We will drop the second term in our lower bound given it is non-negative (as it expresses the vector norm $\Vert |+^n \rangle | A \psi \rangle \rVert^2$).
	
	We can lower bound the third term by first explicitly expressing it
	\begin{equation*}
		-	\Big( \langle  +^n  |\langle \psi | \Big)  \mathcal{P} |\epsilon \rangle
		=
		\tfrac{t^2}{2}\Big( \langle  +^n  |\langle \psi | \Big) 
		S_A^2 \Big( | +^n \rangle |\psi\rangle \Big) \leq 
		\tfrac{t^2}{2} \lVert  S_A \rVert_2^2
		=
		\tfrac{t^2}{2} \lVert A \rVert_{max}^2,
	\end{equation*}
	where we have used the projector acts trivially on $ | +^n \rangle |\psi\rangle$ and used
	from ref.~\cite{rebentrost2018quantum} that the absolute largest eigenvalue of $S_A$
	is $\lVert A \rVert_{max}$.
	We finally conclude that the probability is lower bounded as
	\begin{equation}\label{eq:prob_bound}
		\mathrm{Prob} 
		\geq 
		1 - t^2 
		\lVert A \rVert_{max}^2
		+\mathcal{O}(t^3).
	\end{equation}

	We can therefore discard the ancilla register (given it is separable and in the $|+^n\rangle$ state)
	and we obtain the quantum state $|\psi'\rangle$ as
	\begin{align}
		|+^n\rangle |\psi'\rangle
		=
		|+^n\rangle\frac{ 	e^{-it A/N}|\psi\rangle 
		}
		{\sqrt{\mathrm{Prob}}}
		+\frac{ \mathcal{P} |\epsilon \rangle +\frac{t^2}{2}  |+^n\rangle (\frac{A^2}{N^2}|\psi \rangle  ) }{\sqrt{\mathrm{Prob}}}
		+\mathcal{O}(t^3),
	\end{align}
	The vector norm of the error term is upper bounded via the triangle inequality
	as
	\begin{equation*}
		\lVert \, 
		\frac{ \mathcal{P} |\epsilon \rangle +\frac{t^2}{2}  |+^n\rangle (\frac{A^2}{N^2}|\psi \rangle  ) }{\sqrt{\mathrm{Prob}}}
		\,	\rVert
		\leq 
		\frac{
			\lVert \epsilon \rVert + \frac{t^2}{2}	\lVert A \rVert_2^2/N^2
		}{\sqrt{\mathrm{Prob}}} 
		\leq 3/2 t^2 \lVert A \rVert_{max}^2	+\mathcal{O}(t^4)
	\end{equation*}
	where we have used the inequality  $\lVert A \rVert_2^2/N^2 \leq \lVert A \rVert_{max}^2$, the Taylor expansion of the probability $1/\sqrt{\mathrm{prob}} = 1/\sqrt{1 - x} = 1 + x + \mathcal{O}(x^2)$
	with $x = \mathcal{O}(t^2)$ and the norm of the residual term  $\lVert \epsilon \rVert \leq \tfrac{t^2}{2} \lVert A \rVert_{max}^2 +\mathcal{O}(t^3)$. 
	
	Furthermore, note that $$
		|+^n\rangle\frac{ 	e^{-it A/N}|\psi\rangle 
}
{\sqrt{\mathrm{Prob}}}
	=
	|+^n\rangle  	e^{-it A/N}|\psi\rangle  
	+ |\epsilon'\rangle  
	$$
	with $\lVert |\epsilon'\rangle \rVert \leq t^2 
	\lVert A \rVert_{max}^2 +\mathcal{O}(t^3)$ via the previous expansion of the square root function.
	As such, after measuring the ancilla in the $|+^n\rangle$ state we obtain the
	state
	\begin{equation*}
		|\psi'\rangle =	e^{-it A/N}|\psi\rangle  + |\mathcal{E}\rangle
	\end{equation*}
	with $\lVert \mathcal{E} \rVert \leq 5/2 t^2 \lVert A \rVert_{max}^2$
	where we have again used the Taylor expansion of the square root function.
	
\end{proof}

\begin{lemma}[Total success probability]\label{lemma:prob_failure_bound}
	We use the low-rank approach to simulate an evolution for overall
	time $T$ by applying $r$ piece-wise constant evolutions with time $T/r$.
	The probability that throughout $r$ consecutive iterations
	we always measure the all plus state is lower bounded by
	\begin{equation*}
		\mathrm{prob}_{tot} \geq 1-T^2 	  \lVert A \rVert_{max}^2/r +\mathcal{O}(T^4/r^2),
	\end{equation*}
	and we discard all other measurement outcomes
	for the sake of analytical simplicity. 
\end{lemma}
\begin{proof}
		We can compute the overall probability for $r$ iterations using \cref{eq:prob_bound}
		and substituting that $t \equiv \Delta t = T/r$
	as
	\begin{equation}
		\mathrm{Prob}^r \geq
		(1 - T^2 
		\lVert A \rVert_{max}^2/r^2)^r
		+\mathcal{O}(T^3/r^3)
		=
		1-T^2 
		\lVert A \rVert_{max}^2/r
		+\mathcal{O}(T^4/r^2).
	\end{equation}
\end{proof}

	Note that the above is a loose bound -- in practice one does not need to discard the non-$+$outcomes as they lead to nearly the same performance; but
	in \cref{lemma:low_rank_sim_error} we only took into account and bounded the $|+^n \rangle$ outcome fidelities, and bounding other measurement outcomes is beyond the scope of the present work. We note that the probabilistic aspect of the algorithm is simply an analytical convenience, rather than an intrinsic property of the approach.

\begin{lemma}[Piecewise simulation error]\label{lemma:simulation_algorithmic_error_bound}
For a fixed $s$, the algorithmic simulation error from using Hamiltonian simulation technique summarized in Statement~\ref{statement:low_rank_simulation} to simulate $\H(s)$ for a time $\frac{T}{r}$ is bounded by,
\begin{align*}
    \epsilon_0(s) 
    &\in 
    O\left(\frac{1}{r^2}\left(\frac{\lVert A(1) \rVert_{max}}{\lnorm{A(1)/N}} 
    \frac{\left(1 + \lnorm{A(1)/N} \right)^3}{\lnorm{A(1)/N}^{3/2}}
    \right)^2\right).
\end{align*}
\end{lemma}
\begin{proof}
As summarised in \cref{statement:low_rank_simulation}, the low-rank simulation approach of ref~\cite{rebentrost2018quantum} for time $\Delta t$ under $S_A$ generates the following dynamics under the partial trace
from Eq.~\eqref{eq:low_rank_partial_tr} as
\begin{align*}
    \sigma' = e^{-i \Delta t A(s)/N} \sigma e^{i \Delta t A(s)/N} + \mathcal{O}(\Delta t^2)
\end{align*}
for any pure initial state $\sigma = \op{\phi}{\phi}$.
In \cref{lemma:low_rank_sim_error} we assumed that measurements are performed on the 
ancilla register and only the $|+^n\rangle$ outcome is accepted after which we discard the ancilla register.
This leads to an error in terms of the deviation between the exact $U'(\Delta t, s):= e^{- \Delta t A(s)/N}$
and obtained unitary mapping $\lnorm{U'(\Delta t,s) - U''(\Delta t,s)} = \epsilon_0(\Delta t,s)$.  
For ease of notation we denote the error $\epsilon_0(s) := \epsilon_0(\Delta t,s)$ as
\begin{align}
    \epsilon_0(s) \le \tfrac{5}{2} \lVert A(s) \rVert_{max}^2 \Delta t^2.
\end{align}
In the direct-state preparation algorithm, we must evolve for a total time $T$, which is constant bounded, as shown in Lemma~\ref{lemma:total_evolution_time}. We do so in $r$ evenly spaced intervals, resulting in $\Delta t = \frac{T}{r}$. As a result,
\begin{align}\label{eqn:epsilon_0_s_bound_intermediate}
    \epsilon_0(s) \le \tfrac{5}{2} \left(\lVert A(s) \rVert_{max} \frac{T}{r}\right)^2.
\end{align}
Noting that $f_s$ has a maximum filling ratio when $s=0$, for large $N$ we have,
\begin{align}\label{eqn:filling_ratio_bound}
    \frac{\lVert A(s) \rVert_{max}}{\lnorm{A(s)/N}} 
    \le 
    \frac{\lVert A(1) \rVert_{max}}{\lnorm{A(1)/N}}.
\end{align}
Observe that
\begin{align}
    \lnorm{A(s)/N}
    &\le 
    (1-s)^2\ip{v_0}{v_0} + s^2\ip{v_1}{v_1} + 2s(1-s)\text{Re}\ip{v_0}{v_1}\\
    &\le
    (1-s)^2 + s^2\frac{\mathcal{N}(1)^2}{N} + 2s(1-s)\frac{\mathcal{N}(1)}{\sqrt{N}}\\
    &=
    \left(
    (1-s) + s \frac{\mathcal{N}(1)}{\sqrt{N}}
    \right)^2\\
    &\le 
    \max\left(1, \frac{\mathcal{N}(1)^2}{N}\right)\\
    &\le\label{eqn:norm_a_s_over_N_bound}
    1+\frac{\mathcal{N}(1)^2}{N},
\end{align}
where the second inequality follows from the application of the Cauchy-Schwarz inequality on $\ip{v_0}{v_1}$. 
Combining Equations~\ref{eqn:epsilon_0_s_bound_intermediate} and \ref{eqn:filling_ratio_bound}, we obtain the bound
\begin{align}
    \epsilon_0(s) &\le \tfrac{5}{2}\left(\frac{\lVert A(1) \rVert_{max}}{\lnorm{A(1)/N}} \lnorm{A(s)/N} \frac{T}{r}\right)^2\\
    &\le
    \tfrac{5}{2}\left(\frac{\lVert A(1) \rVert_{max}}{\lnorm{A(1)/N}} \left(1 + \frac{\mathcal{N}(1)^2}{N} \right)
    \frac{T}{r}\right)^2
\end{align}
where the inequality follows from using Equation~\ref{eqn:norm_a_s_over_N_bound}. Using Lemma~\ref{lemma:total_evolution_time}, we have an upper-bound on the delay factor, and thus an \textit{asymptotic} upper-bound on $T$. To allow treatment as a strict upper-bound, we introduce the constant $k$, such that $T \le 2k\frac{\left(\mathcal{N}(1)^2 + N\right)^2}{\sqrt{N}\mathcal{N}(1)^3}$ (in practice, we find that $k\approx 100$). 
Of course, $\left(1 + \frac{\mathcal{N}(1)^2}{N} \right) = \frac{1}{N}\left(N + \mathcal{N}(1)^2 \right)$.
Then,
\begin{align}
    \epsilon_0(s) 
    &\le
     \frac{10 k^2}{r^2}\left(\frac{\lVert A(1) \rVert_{max}}{\lnorm{A(1)/N}} 
    \frac{1}{N^{3/2}}
    \frac{\left(\mathcal{N}(1)^2 + N\right)^3}{\mathcal{N}(1)^3}
    \right)^2\\
    &=
    \frac{10 k^2}{r^2}\left(\frac{\lVert A(1) \rVert_{max}}{\lnorm{A(1)/N}} 
    \left(\frac{\sqrt{N}}{\mathcal{N}(1)}\right)^3\left(1 + \frac{\mathcal{N}(1)^2}{N} \right)^3
    \right)^2\\
    &=
    \frac{10 k^2}{r^2}\left(\frac{\lVert A(1) \rVert_{max}}{\lnorm{A(1)/N}} 
    \frac{\left(1 + \lnorm{A(1)/N} \right)^3}{\lnorm{A(1)/N}^{3/2}}
    \right)^2,
\end{align}
where we used the fact that $\lnorm{A(1)/N} = \frac{\mathcal{N}(1)^2}{N}$.
\end{proof}

\begin{lemma}\label{lemma:hamiltonian_simulation_query_complexity}
The query complexity of simulating $e^{-i \H(s) \Delta t}$ as per Lemmas~\ref{lemma:low_rank_sim_error} and~\ref{lemma:simulation_algorithmic_error_bound} is simply $4$, i.e. is bounded by $O(1)$, while the total number of qubits required are $2n + d + 2$, and thus the algorithm requires $O(n + d)$ ancillary qubits.
Therefore, the total state preparation algorithm has query complexity $O(r)$.
\end{lemma}
\begin{proof}
As discussed at length, the simulation of $e^{-i \H(s) \Delta t}$ is implemented using the low-rank Hamiltonian simulation technique of ref~\cite{rebentrost2018quantum}, where they embed the $N\times N$ matrix $A_s \equiv \H(s)$ into the one-sparse $N^2 \times N^2$ matrix $S_{A_s}$. The one-sparse matrix can then be simulated as described in Lemma~\ref{lemma:one_sparse_simulation} using a total of 4 calls to the oracles $O_f$ and $O_H$, and with a total of $2n + d + 2$ qubits (i.e. $n + d + 2$ ancillary qubits). 
As the algorithm consists of performing $r$ consecutive such Hamiltonian simulations, the total query complexity is given by $O(r)$.
\end{proof}

\begin{lemma}\label{lemma:joint_error_bound}
The total error of the direct state preparation procedure, as measured by the deviation of the exact unitary evolution $U(T)$ from the worst-case unitary evolution we implement $U''(T)$, is given by
\begin{align}
    \lnorm{U(T) - U''(T)} \le \sqrt{2T(\delta_0 + \delta_1)},
\end{align}
where $\delta_0$ comes from Lemma~\ref{cor:discret_error} and $\delta_1=5/2 \lVert A \rVert_{max}^2\frac{T}{r}$.
\end{lemma}
\begin{proof}
First, we allow $\H(s)$ to represent the exact time-dependent Hamiltonian we wish to simulate, $\H'(s)$ the discretized Hamiltonian as per Lemma~\ref{cor:discret_error}, and $\H''(s)$ the discretized Hamiltonian which induces the unitary with algorithmic error obtained via the low-rank simulation technique as per Lemmas~\ref{lemma:low_rank_sim_error} and ~\ref{lemma:simulation_algorithmic_error_bound}. $\H(s)$ and $\H'(s)$ have both already been clearly defined, however, we have never explicitly considered $\H''(s)$ -- instead we have directly treated with the unitary evolution it induces and through this unitary mapping we can define $\H''(s)$ as its generator.  

As shown in~\cref{lemma:low_rank_sim_error} the operator deviation is bounded by,
\begin{align}
    \lnorm{U'(\Delta t) - U''(\Delta t)} \le 5/2 \Delta t^2 \lVert A \rVert_{max}^2.
\end{align}
Using \cref{corollary:implied_unitary_deviation_bound} by observing that $\kappa = 5/2 \lVert A \rVert_{max}^2$ and that $\Delta t = \frac{T}{r}$, we get
\begin{align}
    \lnorm{\H'(s) - \H''(s)} \le 5/2 \lVert A \rVert_{max}^2\frac{T}{r}.
\end{align}
Moreover, from Lemma~\ref{cor:discret_error}, 
\begin{align}
    \lVert \mathcal{H}(s) - \mathcal{H}'(s) \rVert_{2}
    &\leq
	\frac{2}{r}\left(1 + 3\frac{\mathcal{N}(1)}{\sqrt{N}} + 2\frac{\mathcal{N}(1)^2}{N} \right) = \delta_0.
\end{align}
As a result, Lemma~\ref{lemma:double_deviation_induced_unitary_bound} gives the bound
\begin{align}
    \lnorm{\H(s) - \H''(s)} \le \delta_0 + \delta_1.
\end{align}
Thus, Lemma~\ref{lemma:implied_unitary_deviation_bound} gives,
\begin{align}
    \lnorm{U(T) - U''(T)} \le \sqrt{2T(\delta_0 + \delta_1)}.
\end{align}
\end{proof}

\subsection{Asymptotic limits \label{sec:asymptotic}}

Above we have derived our bounds in terms of two quantities as the matrix norms 
$\lnorm{A(1)/N} = \mathcal{N}(1)^2/N$ and $\lVert A(1) \rVert_{max}$. 
We prove in \cref{sec:further_asymptotic} that both these quantities naturally admit
asymptotic limits for $N \rightarrow \infty$ for any continuous function. Furthermore,
we also prove that Lipschitz continuous functions---which cover nearly all instances in practice---approach
these limits in exponential order in the number of qubits $n$ as illustrated in \cref{fig:matrix_norms} and in \cref{fig:matrix_norms2}.
We apply these results and establish asymptotic limits of our error bounds.

\begin{figure*}[tb]
	\begin{centering}
		\includegraphics[width=\textwidth]{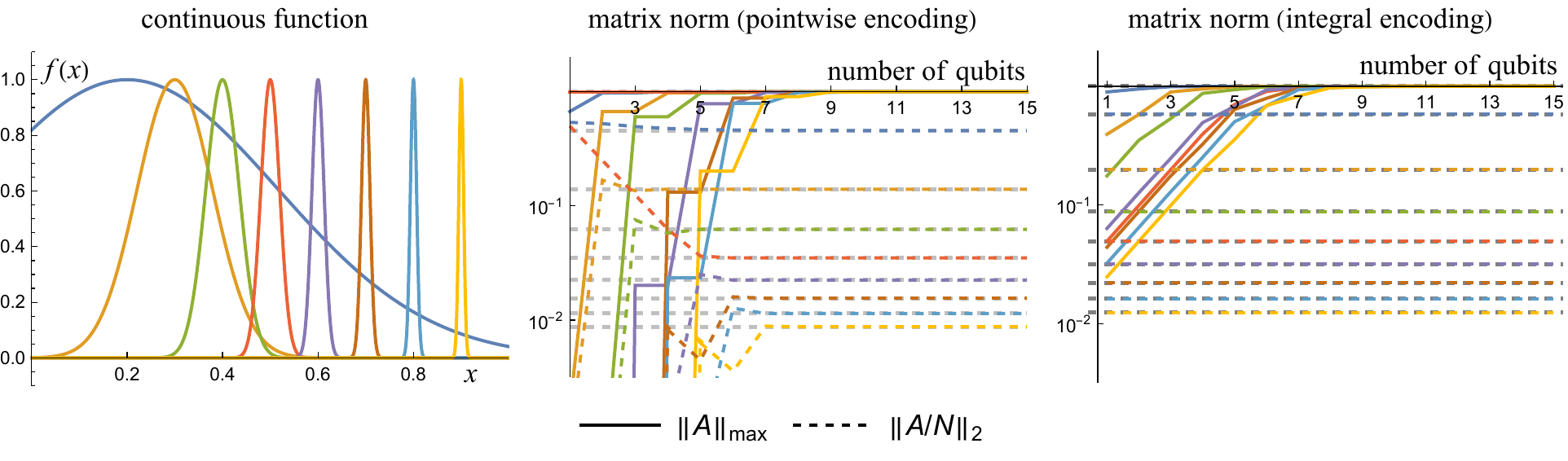}
		\caption{
			(left) Encoding a Gaussian function into our rank-1 matrix $A \propto | \psi \rangle \langle \psi |$
			from \cref{eq:param_hamil_def} at $s=1$.
			(middle-right)
			The matrix norms 
$\lnorm{A(1)/N} = \mathcal{N}(1)^2/N$ and $\lVert A(1) \rVert_{max}$ (solid and dashed colored lines) determine the simulation error of the 1-sparse encoding in case of (middle) the integral  and (right) pointwise encoding as we proved in \cref{sec:non-asymptotic}. 			Both matrix norms converge to a constant value (dashed grey lines) asymptotically in $n$ where the constants are given by function norms of $f$ as established in Lemmas~\ref{lemma:amatr_pointwise}-\ref{lemma:amatrix_integral} -- and the convergence rate is exponential given Lipschitz continuity of the functions. Colour coding represents the width of the Gaussian function and thus a decreasing function norm.
The asymptotic runtime of our state-preparation only depends on the ratio of these two matrix norms and thus depends only on the absolute area occupied by the function (relative to its maximum value) as captured by our filling ratio $\mathcal{F}$ which we illustrate in \cref{fig:matrix_norms2}.			\label{fig:matrix_norms}
		}
	\end{centering}
\end{figure*}

\begin{figure*}[tb]
	\begin{centering}
		\includegraphics[width=\textwidth]{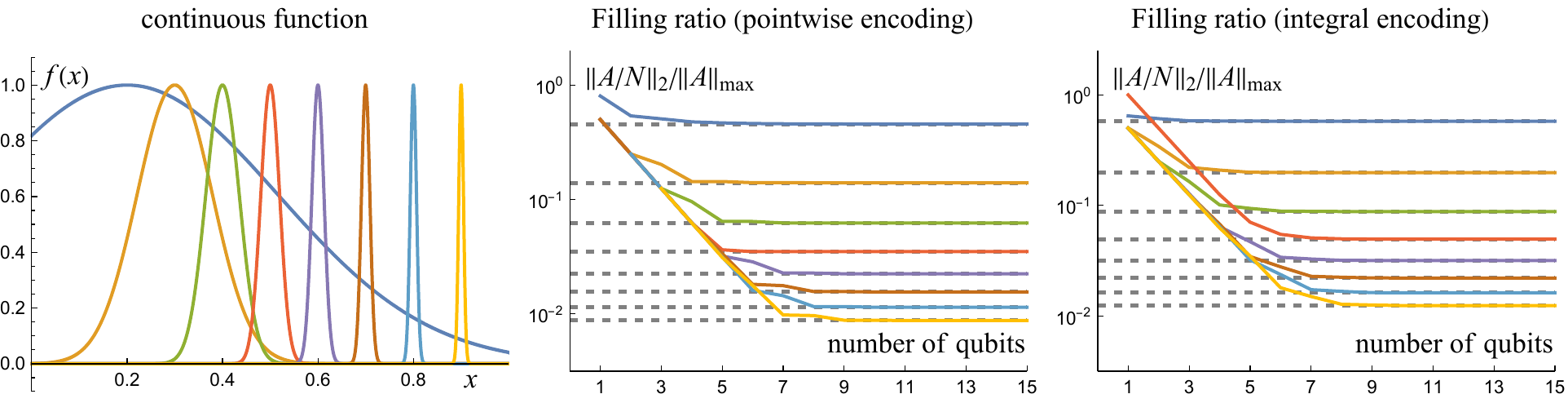}
		\caption{
			Same as in \cref{fig:matrix_norms2} but showing the ratio of the two matrix
			norms which asymptotically approaches the respective filling ratios.	\label{fig:matrix_norms2}
		}
	\end{centering}
\end{figure*}

\begin{remark}[Growth rate of normalisation]\label{remark:asymptotically_constant_ratio}
	For any fixed continuous function $f_1$ we obtain the scaling for our pointswise encoding $\mathcal{N}(1) \in O(\sqrt{N})$ in the asymptotic limit of big $N$.
	More generally we can state $\mathcal{N}(1)/\sqrt{N} \in \Theta(1)$, which
	immediately follows from the fact that the normalisation
	constant asymptotically approaches the Riemann integral $\mathcal{N}(1)^2/N \propto \lVert f_1 \rVert^2_2$ as per \cref{lemma:riemann_sum}, \cref{lemma:amatr_pointwise} and \cref{lemma:amatrix_integral}.
\end{remark}
These asymptotic limits are illustrated in \cref{fig:matrix_norms}.
Intuitively, each $\epsilon$ neighbourhood around any $f_1(x)$ is asymptotically constant and so the summation for $\mathcal{N}(1)$ is a sum of $O(N)$ constants and therefore taking the square root results in a quantity which clearly grows as $O(\sqrt{N})$.
As a consequence of this asymptotic limit, we are guaranteed that the total adiabatic evolution time from
\cref{lemma:total_evolution_time} is $T \in \Theta (1)$, which property we exploit in the following.

\begin{corollary}[Asymptotic simulation error]\label{lemma:simulation_algorithmic_error_bound_asymp}
	By Lemma~\ref{lemma:simulation_algorithmic_error_bound}, if we re-scale $f_1$ as per Remark~\ref{remark:easy_rescaling_remark}, the algorithmic simulation error from simulating $\H(s)$ for time $T/r$ using the low-rank Hamiltonian simulation technique as per Statement~\ref{statement:low_rank_simulation} is bounded by $
		\epsilon_0(s) \in O  (\lVert A(1) \rVert_{max}^2 / r^2  )
	$.
	In the limit of big $N$, this may equivalently be written in terms of the filling ratio $\fillr$ from \cref{def:peakedness}
	as
	\begin{align*}
		\epsilon_0(s) \in O\left(\frac{\mathcal{F}^p}{r^2}\right)
	\end{align*}
	where $p=-4$ in the point-wise sampling case, and $p=-2$ in the integration case
	via \cref{lemma:amatr_pointwise} and \cref{lemma:amatrix_integral}.
\end{corollary}

Using the low-rank simulation approach has two significant advantages.
First, as the matrix entries of $A$ are computed using arithmetic operations, it is beneficial that the entries have asymptotically constant magnitudes (i.e. $[A]_{kl} \in \mathcal{O}(N^0)$). Thus, we can use a fixed resolution in encoding the matrix entries $A_{kl}$ as integers, whereas the desired quantum state might have exponentially small amplitudes.
Second, the fact that our simulation actually simulates $A/N$ instead of just $A$ is actually advantageous, as the simulation of $A/N$ automatically absorbs the exponentially decreasing (in terms of the number of qubits) absolute values of the amplitudes of the quantum states (that encode continuous function) --  which we prove in \cref{prop:pointwise_amplitude} and in \cref{prop:integral_amplitude}. In other approaches, encoding the wave function may necessitate using exponentially decreasing rotation angles, but here this is entirely avoided.

\begin{lemma}[Total success probability]\label{cor:prob_failure_bound_asympt}
	The success probability in \cref{lemma:prob_failure_bound} admits the asymptotic limit
	\begin{equation*}
		\mathrm{prob}_{tot} \geq 1- O(\mathcal{F}^p/r),
	\end{equation*}
	where $p=-4$ in the point-wise sampling case, and $p=-2$ in the integration case
	via \cref{lemma:amatr_pointwise} and \cref{lemma:amatrix_integral}.
\end{lemma}
\begin{proof}
	We obtain the asymptotic bounds by substituting
	limits from \cref{lemma:amatr_pointwise} and \cref{lemma:amatrix_integral}
	into our bounds in \cref{lemma:prob_failure_bound}.
\end{proof}

\begin{corollary}[Total simulation error]\label{cor:total_error_deviaton}
	By Lemma~\ref{lemma:joint_error_bound} and Remark~\ref{remark:asymptotically_constant_ratio}, the total error of the direct state preparation procedure, as measured by the deviation of the exact unitary evolution $U(T)$ from the worst-case unitary evolution we implement $U''(T)$, is given by,
	\begin{align}
		\lnorm{U(T) - U''(T)} \in O\left( 
		\sqrt{\frac{\mathcal{F}^p}{r}}
		\right),
	\end{align}
	in the limit of big $N$ via \cref{lemma:amatr_pointwise} and \cref{lemma:amatrix_integral}..
\end{corollary}
\begin{proof}
	Taking the limit of the bound in \cref{lemma:joint_error_bound} for big $N$, assuming a normalized $f_1$ (although this assumption is not necessary), $T\in O(1)$, $\delta_0 \in \tfrac{1}{r}$, $\epsilon_0(s) \in O(\frac{\mathcal{F}^p}{r^2})$, we get the asymptotic upper bound
	\begin{align}
		\lnorm{U(T) - U''(T)} \in O\left( 
		\sqrt{\frac{1 + \mathcal{F}^p}{r}}
		\right).
	\end{align}
	Noting that $\mathcal{F}^p \ge 1$, $1 + \mathcal{F}^p \le 2 \mathcal{F}^p$, and so our asymptotic statement simplifies to,
	\begin{align}
		\lnorm{U(T) - U''(T)} \in O\left( 
		\sqrt{\frac{\mathcal{F}^p}{r}}
		\right).
	\end{align}
	It is worth explicitly noting that this asymptotic bound also holds when $f_1$ is not normalized. If $f_1$ is not normalized, the ratio $\mathcal{N}(1)^2/N$ is still asymptotically $O(1)$, however it may be a large constant, depending on how ``unnormalized'' $f_1$ is. As a result, this will be off by at most a constant factor, and so the asymptotic bound remains unchanged. 
\end{proof}

\begin{remark}[When adiabatic error is negligible]
	In practical applications (especially for small filling ratios $\fillr$)  the total error
	is not dominated by the error of the discretised adiabatic evolution but rather by the low-rank simulation from \cref{lemma:simulation_algorithmic_error_bound} for which we incur an error at each small timestep $\Delta t$ asymptotically via 
	\cref{lemma:simulation_algorithmic_error_bound_asymp} as $ O(  \mathcal{F}^p/ r^2 )$.
	Thus neglecting the adiabatic simulation error and applying the triangle inequality $r$ times consecutively
	(given $\epsilon_0(s)$ is a bound on the unitary operator distance) we obtain the bound
	\begin{align*}
		\lnorm{U(T) - U''(T)}  \, \lessapprox \, O\left(\frac{\mathcal{F}^{p}}{r}\right),
	\end{align*}
	where $p=-4$ in the point-wise sampling case, and $p=-2$ in the integration case.
	This scaling is tighter in $r$ but looser in the filling ratio $\fillr$ than in \cref{cor:total_error_deviaton}.
\end{remark}

\subsection{Generalisation to arbitrary functions (mappings)\label{sec:generalisation}}

So far we have considered continuous functions $f_s(x)$ and discretised them to obtain our desired normalised quantum state from \cref{eq:encoding}. We proved in \cref{sec:non-asymptotic} that for any finite qubit count our approach allows us to prepare our final desired state with complexity that depends only on the two matrix norms
$\lnorm{A(1)/N} = \mathcal{N}(1)^2/N$ and $\lVert A(1) \rVert_{max}$, in fact on the ratio of these two.
We did not actually  have to invoke the property that the function is continuous thus far. In fact we invoked continuity because it guarantees that asymptotically in $N$ the matrix norm $\lVert A(1) \rVert_{max} \rightarrow \max_{x}|f(x)|^2$  is approached via the boundedness theorem and that 
$\lVert A/N \rVert_{2}
\rightarrow (b-a) \lVert f \rVert_2^2 $ due to Riemann integrability
as we prove in \cref{lemma:amatr_pointwise}. For this reason we do not actually need to necessary assume continuity of the function as we explain now.

In fact our proofs in \cref{sec:non-asymptotic} are applicable to the case of \emph{any discrete mapping}
as $F^{(n)}: \{0,1\}^n \mapsto \mathbb{C}$ in the sense that we can efficiently prepare \textit{any} $n$-qubit quantum state of the form
\begin{equation}\label{eq:general_encode}
    \ket{\psi}_n = (\mathcal{N}_N)^{-1}  \sum_{j\in \{0,1\}^n}    F^{(n)}(j) \ket{j}_n,
\end{equation}
given the following conditions are satisfied.

\begin{theorem}
Given an arbitrary efficiently computable mapping
$F^{(n)}: \{0,1\}^n \mapsto \mathbb{C}$ that depends on the qubit count $n$
with absolute maximum value $\lVert F^{(n)} \rVert_{max} = \max_{k \in \{0,1\}^n} |F^{(n)}(k)|$ 
and vector norm $\mathcal{N}^2 = \sum_{k \in \{0,1\}^n} |F^{(n)}(k)|^2$
we can prepare the
quantum state in \cref{eq:general_encode} efficiently
using our adiabatic approach given that $\lVert F^{(n)} \rVert_{max} \in O(1)$ is asymptotically upper bounded that $\mathcal{N}/\sqrt{N} \in \Theta(1)$ is asymptotically both upper and lower bounded.
In direct analogy with \cref{theorem:total_query_complexity}
the query complexity then depends on the generalised filling ratio $r \in O(   \tilde{\fillr}^{-4} /\epsilon^2 )$.
\end{theorem}
\begin{proof}
We define the pointwise encoded matrix $[A(s)]_{kl}:= F_s(k) F^*_s(l)$ following \cref{eq:param_hamil_def} with the parametric mapping $F_s :=  (1-s) F_0
    + s F_1$ where $F_0$ is the same trivial function with a properly chosen phase as in \cref{sec:non-asymptotic} and $F_1$ is our desired mapping  $F_1 := F^{(n)}$ where we will  drop the dependence on $n$ for ease of notation. We proved in \cref{lemma:total_evolution_time} that the total adiabatic evolution time is bounded as
    $T\in O\left(\frac{\left(\mathcal{N}(1)^2/N + 1\right)^2}{\left(\mathcal{N}(1)/\sqrt{N}\right)^3} \right)$
    and in \cref{lemma:simulation_algorithmic_error_bound} that the simulation error is bounded as
    \begin{align*}
    \epsilon_0(s) 
    &\in 
    O\left(\frac{1}{r^2}\left(\frac{\lVert A(1) \rVert_{max}}{\lnorm{A(1)/N}} 
    \frac{\left(1 + \lnorm{A(1)/N} \right)^3}{\lnorm{A(1)/N}^{3/2}}
    \right)^2\right).
\end{align*}
for any finite qubit count. As such, given the matrix norms $\lnorm{A(1)/N} = \mathcal{N}^2(1)/N$ and $\lVert A(1) \rVert_{max} = \lVert F_1 \rVert_{max}$ are asymptotically in $n$ both lower and upper bounded as $\mathcal{N}(1)/\sqrt{N}\in \Theta(1)$ and upper bounded $\lVert F_1 \rVert_{max} \in O(1)$ we can guarantee a bounded total evolution time
$T \in O(1)$ and that the simulation error is bounded as $\epsilon_0 \in O(1/r^2)$. 
Indeed in the special case when  $F(j) = f(x_j)$, $F(j)$ is just a discretisation of a continuous function, and the  norms are convergent for $n\rightarrow \infty$ as
 $\lVert A(1) \rVert_{max} \rightarrow \max_{x}|f(x)|^2$  is approached via the boundedness theorem and that 
$\mathcal{N}(1)/\sqrt{N}
\rightarrow \sqrt{(b-a)} \lVert f \rVert_2 $ due to Riemann integrability
as we prove in \cref{lemma:amatr_pointwise}. In contrast,  our more general asymptotic conditions do not require convergence
and it suffices that there exist constants such that the $n$-dependent norms
$\mathcal{N}(1)/\sqrt{N}$ and $\lVert F_1 \rVert_{max}$ are asymptotically bounded.
Clearly, in analogy with the filling ratio in \cref{theorem:total_query_complexity} in the more
general case the query complexity depends on the  generalised filling ratio $r \in O(   \tilde{\fillr}^{-4} /\epsilon^2 )$ which we can define as an asymptotic bound on the ratio of norms as
\begin{equation}
\frac{   \mathcal{N}(1)  } {   \sqrt{N}\lVert F_1 \rVert_{max}     } \in O( \tilde{\fillr}  ).
\end{equation}
\end{proof}

Let us illustrate the above result on two examples.

\noindent \textbf{Example:} Let us consider a mapping  $F^{(n)}: \{0,1\}^n \mapsto [-1,1]$ that takes a binary integer $j \in \{0,1\}^n$ and uses it as a seed to efficiently compute a sample $F^{(n)}(j) = x_j$ from a random distribution $X$ such that the mean is $\mu(X) =0$.  
By definition this mapping is bounded $\lVert F^{(n)} \rVert_{max} \leq 1$ and we can see that the norm converges to the asymptotic constant  \begin{equation*}
    \lim_{n \rightarrow \infty}\mathcal{N}^2/N =
    \lim_{n \rightarrow \infty}\frac{1}{N}\sum_{k \in \{0,1\}^n} x_k^2
    = \mathrm{Var}[X].
\end{equation*}
As such, we can efficiently prepare an exponentially large list of $N$ random numbers of zero mean with complexity determined by the variance via the generalised filling ratio $\tilde{\fillr} = \mathrm{Var}[X]$.

\noindent
\textbf{Counterexample:}
Let us consider the unstructured search problem, where we have a function $F^{(n)}: \{0,1\}^n \mapsto \{0,1\}$  such that there is only one marked input $m$ for which $F^{(n)}(m) = 1$ and for all other inputs $F^{(n)}(j) =0$ with $j\neq m$.
By definition this mapping is bounded $\lVert F^{(n)} \rVert_{max} = 1$ and we can straightforwardly evaluate the normalisation as
\begin{equation*}
\mathcal{N}^2/N = \frac{1}{N}\sum_{k \in \{0,1\}^n} |F^{(n)}(k)|^2
=
\frac{1}{N}.
\end{equation*}
Indeed we find the generalised filling ratio  $\tilde{\fillr} = 2^{-n/2}$ and thus the query complexity of preparing this mapping grows exponentially as $O(2^{2n})$. Given efficiently preparing this state would solve the unstructured problem we see it is in fact less efficient than a classical direct search.

\section{State Preparation via Quantum Phase Estimation}\label{section:qpe_state_prep}

In this section, we introduce a non-deterministic variant of the state preparation procedure utilizing quantum phase estimation.
Begin by defining the two $n$-qubit states,
\begin{align}
    \ket{\psi} = \frac{1}{\mathcal{N}}\sum_{j=0}^{2^n - 1}f(x_j)\ket{j},
    \quad \quad \text{and} \quad \quad
    \ket{\phi} = \sum_{j=0}^{2^n - 1}f(x_j)\ket{j}
\end{align}
where $\ket{\psi}$ is normalized (with normalizing coefficient $\frac{1}{\mathcal{N}}$) and $\ket{\phi}$ is not. Our low-rank Hamiltonian simulation procedure is capable of performing the evolution $e^{-i \frac{t}{N} \op{\phi}{\phi}}$ for some time $t$ with error scaling as $O(t^2)$, as per Statement~\ref{statement:low_rank_simulation}. Thus, we define the unitary $U(t) = e^{-i \frac{t}{N} \op{\phi}{\phi}}$, and note that $U(t)^k = U(tk)$. 

Noting that $\op{\phi}{\phi} = \mathcal{N}^2\op{\psi}{\psi}$, we can rewrite
$ e^{-i \frac{t}{N} \op{\phi}{\phi}} 
    =
    e^{-i \frac{t \mathcal{N}^2}{N} \op{\psi}{\psi}}$.
Suppose $\{\ket{\psi_j}\}_j$ forms an orthonormal eigenbasis for $e^{-i \frac{t \mathcal{N}^2}{N} \op{\psi}{\psi}}$, such that $\ket{\psi_0} \equiv \ket{\psi}$ and,
\begin{align}
    e^{-i \frac{t \mathcal{N}^2}{N} \op{\psi}{\psi}}
    \ket{\psi_j}
    =
    e^{-i \gamma_j}\ket{\psi_j},
\end{align}
where $\gamma_0 = \frac{t\mathcal{N}^2}{N}$ and $\gamma_j = 0$ for $j>0$.
Then, we create the initial state $\ket{\gamma}$, and expand it in the aforementioned eigenbasis,
\begin{align}
    \ket{\gamma} = \sum_{j=0}^{2^n - 1}c_j\ket{\psi_j},
\end{align}
where $c_j$ is the complex amplitude of each eigenvector in this basis associated with our state. Moreover, let $t<0$, so that $-t$ is positive.

Now suppose we have an ancillary $m$-qubit register initially in the state $\ket{+}^{\otimes m}$. Then the initial joint state of our system is given by,
\begin{align}
    \ket{+}^{\otimes m}\ket{\gamma} = \sum_{j=0}^{2^n - 1}c_j \ket{+}^{\otimes m}\ket{\psi_j}.
\end{align}
We now apply a sequence of $m$ controlled $U(t)$ gates, such that when the control is applied on the $j^{th}$ qubit in the $m$-qubit register, we apply $U(t)^{2^{m - j - 1}} = U(t 2^{m - j - 1})$ on the $n$-qubit register, as shown in Figure~\ref{fig:phase_estimation_state_preparation}.
Assuming the absence of algorithmic error in the simulation of $U(t)$, this yields the state,
\begin{align}
    \frac{1}{\sqrt{2^m}}\sum_{j=0}^{2^n - 1}c_j (\ket{0} + e^{2\pi i\gamma_j 2^{m - 1}}\ket{0})\otimes ... \otimes (\ket{0} + e^{2\pi i\gamma_j 2^{0}}\ket{0})\ket{\psi_j}.
\end{align}
For ease of explanation, assuming that $\gamma_j$ can be exactly represented in $m$ bits, applying an inverse-QFT in the first register yields,
\begin{align}
    \sum_{j=0}^{2^m - 1}c_j\ket{\gamma_j}\ket{\psi_j}
    =
    c_0\ket{\gamma_0}\ket{\psi_0} + \ket{0}\sum_{j=1}^{2^m - 1}c_j\ket{\psi_j},
\end{align}
where the equality follows from the fact that $\gamma_j \neq 0$ if and only if $j=0$. Thus, measuring the first register collapses the system to the desired state $\ket{\psi_0}$ with probability $|c_0|^2$ (if $t$ is selected to be \textit{sufficiently large} -- defined shortly) and collapses the state to a superposition of all orthogonal eigenvectors with probability $1-|c_0|^2$. 

We must now consider how a \textit{sufficiently large} $t$ is selected. In this analysis, we assume the value of $\mathcal{N}$ is known (noting that if it is not known, the procedure in the following subsection can be used to find it). If $t$ is selected to be too small, each evolution performed is too close to identity, and the value in the first register after the inverse QFT will always be $\ket{0}$ (and thus our target register does not collapse to the desired state). However, we want to pick the minimum $t$ possible, as in practice, the error in each simulation scales with the square of the total simulation time.
We note that the value of the most-significant qubit in the $m$-qubit ancillary register must perform a rotation of at least $2\pi$ for the inverse QFT to yield a non-zero result. As such, since $\gamma_0 = \frac{t \mathcal{N}^2}{N}$ and the longest time evolution $U(t 2^{m-1})$ we implement should be at least a full period, the minimum possible value of $t$ must satisfy
\begin{align}
    \gamma_0 2^{m - 1} = 1 
    \implies
    \frac{t \mathcal{N}^2 2^{m - 1}}{N} = 1.
\end{align}
Thus, we have a bound for the smallest possible $t$ that gives a non-zero result in the first register with probability $|c_0|^2$:
\begin{align}
    t \ge \frac{N}{\mathcal{N}^2 2^{m-1}}.
\end{align}
Moreover, we also note that in QPE, the phase we learn must be between $0$ and $1$. As such, we have the constraint on the maximum value of $t$ of
\begin{align}\label{eqn:t_upperbound}
    e^{2\pi i \gamma_0} = e^{2\pi i \frac{\mathcal{N}^2 t}{N}}
    \implies
    t \le \frac{N}{\mathcal{N}^2}.
\end{align}
In order for this procedure to allow the efficient preparation of $\ket{\psi}$, all that remains is to show that it is possible to prepare some initial $\gamma$ such that $|c_0|^2$ is reasonably bounded (and thus that we have an efficient probability of success). To do this, we simply use Property~\ref{property:overlap} and prepare the state described in \cref{eqn:initial_overlap} in \cref{section:prep_via_destructive_interference}. Doing so guarantees that $c_0 \ge \mathcal{F}/(b - a)$ and so our probability of successfully preparing the state is lower-bounded by $\mathcal{F}^2/(b - a)^2$ (i.e. it only depends on properties of the function and on the width of the interval being considered). We note that the probability of success presented here is a very loose lower-bound, and there are a number of straightforward ways to boost this.

\begin{figure*}[tb]\label{fig:phase_estimation_state_preparation}
    \[
        \Qcircuit @C=1em @R=2em {
            \lstick{\ket{+}_1}      & \qw    & \ctrl{3}                & \qw & \hdots  && \qw & \qw                 & \multigate{2}{QFT^{\dagger}} & \qw & \meter \\
                                    &  &                         &     &   \vdots      &&     &                     & \nghost{QFT^{\dagger}} & &   \\
            \lstick{\ket{+}_1}      & \qw    & \qw                     & \qw & \hdots  && \qw & \ctrl{1}            & \ghost{QFT^{\dagger}} & \qw & \meter  \\
            \lstick{\ket{\gamma}_n} & \qw    & \gate{U(t)^{2^{m - 1}}} & \qw & \hdots  && \qw & \gate{U(t)^{2^{0}}} & \qw & \qw & \qw 
        }
    \]
    \caption{Quantum circuit used for preparing a continuous quantum state via phase estimation, or for estimating the normalization constant of a given function. The initial state is given by $\ket{+}_1^{\otimes m}\ket{\gamma}_n$.}
\end{figure*}
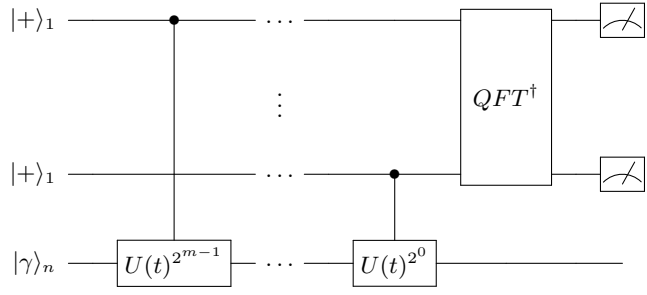

\subsection{Application: Estimating Normalization Factor}\label{subsec:estimating_normalization_factor}
Again, suppose $\{\ket{\psi_j}\}_j$ forms an orthonormal eigenbasis for $e^{-i \frac{t \mathcal{N}^2}{N} \op{\psi}{\psi}}$, such that $\ket{\psi_0} \equiv \ket{\psi}$ and, $e^{-i\frac{t \mathcal{N}^2}{N} \op{\psi}{\psi}}\ket{\psi_j} = e^{-i \gamma_j}\ket{\psi_j}$, where $\gamma_0 = \frac{t\mathcal{N}^2}{N}$ and $\gamma_j = 0$ for $j>0$.

Moreover, suppose we have prepared the state $\ket{\tilde{\psi}}$ such that
\begin{align}
    \ket{\tilde{\psi}} = \sqrt{\lambda}\ket{\psi_0} + \sqrt{1-\lambda}\ket{\psi_{\perp}},
\end{align}
where $\ket{\psi_{\perp}} = \sum_{j=1}^{N-1}c_j\ket{\psi_j}$ and $\ket{\psi_{\perp}}$ is normalized.
We can then apply the phase estimation circuit shown in Figure~\ref{fig:phase_estimation_state_preparation}, assuming perfect error-free simulation, to obtain the state
\begin{align}
    \sqrt{\lambda}\ket{\gamma_0}\ket{\psi_0} + \sqrt{1-\lambda}\sum_{j=1}^{N-1} c_j\ket{\gamma_j}\ket{\psi_j}.
\end{align}
Of course, for $j>0$, $\gamma_j = 0$, and so this can be simplified to
\begin{align}\label{eqn:prepared_normalization_comp_state}
    \sqrt{\lambda}\ket{\gamma_0}\ket{\psi_0} + \sqrt{1-\lambda}\ket{0}\ket{\psi_{\perp}}.
\end{align}
If we have selected a sufficiently large value for $t$ (\textit{sufficiently large} is defined in Section~\ref{section:qpe_state_prep}), then the first register only collapses to the $\ket{0}$ state when the second register collapses to the $\ket{\psi_{\perp}}$ state. In this scenario, we measure the first register until we see a non-zero value, which is guaranteed to be $\ket{\gamma_0}$. This value can then be directly used to compute the normalization factor with,
\begin{align}
    \mathcal{N} = \sqrt{\frac{N \gamma_0}{t}}.
\end{align}
Here, $N$ is clearly known based off of the system size, $t$ is a user specified parameter, and $\gamma_0\in[0,1]$ is read out from the quantum register.
In the absence of simulation error, naive direct sampling of the quantum state takes an expected number of trials of order $\Theta(\frac{1}{\lambda})$ to measure the state $\ket{\gamma_0}$ and thus to compute the normalization factor. Noting that we can easily create a state $\ket{\tilde{\psi}}$ such that $\lambda \ge \mathcal{F}^2/(b-a)^2$, the expected number of trials required for success is then bounded by $\Omega\left((b-a)^2/\mathcal{F}^2 \right)$ (i.e. it depends only on the width of the interval, and on properties of the function).

However, if $t$ is not sufficiently large, then the first register will be unentangled and in the $\ket{0}$ state, regardless of the state of the second register. As a result, no matter how many times you sample the resulting quantum state, you will always measure the first register in the $\ket{0}$ state, and thus gain no information about the normalization factor. Then, if we take $s$ samples from the state in Equation~\ref{eqn:prepared_normalization_comp_state} and never observe a non-zero state, we can conclude with high probability that $t$ is too small, so long as $s$ is selected to be a large enough value. Precisely, the probability of failing in any given sample (assuming that $t$ is sufficiently large) is $1-\lambda$. As a result, the probability of failing after $s$ samples is simply,
\begin{align}
    (1-\lambda)^s \le \left(\frac{(b - a + \mathcal{F})(b - a - \mathcal{F})}{(b - a)^2}\right)^s.
\end{align}
This motivates an algorithm. Suppose we select an acceptable probability of failure $\epsilon \in (0, 1)$, such that we want to ensure if we conclude $t$ is too small, that we are only wrong with probability at most $\epsilon$. Then we have the bound,
\begin{align}
    \left(\frac{(b - a + \mathcal{F})(b - a - \mathcal{F})}{(b - a)^2}\right)^s \le \epsilon,
\end{align}
which implies we need to perform a number of samples with lower-bound,
\begin{align}
    \left|\frac{\log (b - a + \mathcal{F}) + \log (b - a - \mathcal{F}) - 2\log(b - a)}{\log\epsilon}\right| \le s.
\end{align}
Thus, $s$ depends logarithmically on constants such as the width of the grid, and properties of the function, and has inverse-logarithmic scaling in the desired probability of failure. Then, if we conclude that $t$ is too small (by taking $s$ samples and observing all $\ket{0}$ outcomes) we can exponentially increase $t$ (e.g. by doubling) and repeat this process until we observe a non-zero measurement in the first register -- at which point we can directly compute the value of the normalization factor, as already described. Note that the specific rate at which we increase $t$ each time probably requires some more specific consideration, as we need to be careful to ensure that $t$ never exceeds the maximum value prescribed by Equation~\ref{eqn:t_upperbound}.

However, if we do not wish to rely on knowledge of properties of the function such as its filling ratio $\mathcal{F}$, we can instead combine this procedure with the adiabatic direct state preparation approach. In doing so, we iteratively construct the state $\ket{\tilde{\psi}}$ by computing the normalization factor for each of the $r$ discretizations of $\op{\psi(s)}{\psi(s)}$ (as defined in Equation~\ref{eq:param_hamil_def}). In each stage, we compute the normalization factor for $\frac{1}{N}\op{\psi(s)}{\psi(s)}$, and use it to perform the appropriate evolution, giving the state $\ket{\psi(s+\frac{1}{r})}$. We then repeat this process until $s=1$, at which point we have some initial state that is arbitrarily close to the exact desired state, giving this procedure an arbitrarily high probability of success. 
Note that the probability of success does not exponentially decay, as once we have computed the normalization factor at time $s$ we can restart the procedure and deterministically prepare the state up to time $s$.

However, this analysis assumes that we can implement each of the controlled evolutions of $e^{-i\frac{t}{N}\op{\phi}{\phi}}$ with no error -- in practice, we incur error proportional to $t^2$ for each application of $U(t)$. We leave the analysis of this error prone situation to future work.

\subsection{Application: Efficient Integration of Lipschitz Continuous Functions}\label{appendix:efficient_integration_routine}
The procedure for computing the normalization factor of a function immediately leads to an efficient algorithm for integrating that function over an arbitrary domain.
In principle, in the absence of simulation error, this algorithm allows for an exponential improvement in integration over quantum sampling based methods (e.g. using the SWAP test) or any classical methods.

Let us first introduce the approach on the example of a non-negative function.
Formally, we wish to integrate the Lipschitz continuous non-negative function $h$ on the interval $[a, b]$, $\int_{a}^{b}h(x)dx$. Using Lemma~\ref{lemma:riemann_sum}, we immediately get
\begin{align}
    \int_{a}^{b}h(x)dx = \Delta_{N}\sum_{j=0}^{N-1} h(x_j) + \epsilon,
\end{align}
where $\epsilon \in O(N^{-1})$ for Lipschitz continuous functions. Define the function $f(x) = \sqrt{h(x)}$. Then,
\begin{align}
    \Delta_{N}\sum_{j=0}^{N-1} h(x_j) + \epsilon
    &=
    \Delta_{N}\sum_{j=0}^{N-1} f(x_j)^2 + \epsilon \\
    &= \Delta_{N} \mathcal{N}^2 + O(N^{-1}), 
\end{align}
where the last equality follows from the fact that the normalization constant for the function $f$ is given by $\mathcal{N} = \sqrt{\sum_{j=0}^{N-1}|f(x_j)|^2}$. Clearly, computing the normalization factor for the function $f$ then immediately leads to an approximation for the integral of $g$ in the domain $[a, b]$ with error $O(N^{-1})$ (in the absence of simulation error). The asymptotic complexity of this procedure is then dominated by the complexity of computing $\mathcal{N}$ as described in Subsection~\ref{subsec:estimating_normalization_factor}.

It follows that any real, Lipschitz continuous function can be integrated this way. We split $f$ into two non-negative  functions as $f_+:=(f + |f|)/2$ and $f_-:=-(f - |f|)/2$, and apply our approach to these two functions separately. We then subtract the results from each other. This similarly applies to complex functions by splitting them into four quadrants in the complex plane and separately integrating four non-negative functions.

\section{State Preparation via Destructive Interference}\label{section:prep_via_destructive_interference}

\subsection{The Procedure}

Here present a simplified version of the above phase-estimation protocol that only requires a single ancilla qubit.
Assume we know the normalisation constant $c^2 =\lVert A/N \rVert_2 = \mathcal{N}^2/N$
we can set an evolution time $t=T/c^2$ and thus we can implement the fixed-time evolution under the projector as $e^{-i T/c^2 A/N} = e^{-i T  |\psi \rangle \langle \psi |}$.

Recall that we can write any quantum state in terms of the ideal $| \psi \rangle$ and an orthogonal component as
$\sqrt{\lambda} |\psi \rangle + \sqrt{1-\lambda} |\psi_\perp \rangle$
up to trivial global phases.
Setting $T=\pi$ allows us to invert the sign of the ideal component $| \psi \rangle$ in any quantum state as
\begin{equation*}
   e^{-i \pi  |\psi \rangle \langle \psi |}
   [
   \sqrt{\lambda} |\psi \rangle + \sqrt{1-\lambda} |\psi_\perp \rangle
   ]
   =- \sqrt{\lambda} |\psi \rangle + \sqrt{1-\lambda} |\psi_\perp \rangle,
\end{equation*}
without affecting the orthogonal component $|\psi_\perp \rangle$
in Hilbert space. Here $\lambda$ is the fidelity with respect to the ideal
component.

Applying this evolution as part of the circuit in Fig.~\ref{fig:destructive}
allows us to prepare the quantum state and we obtain the joint state after the controlled evolution
but before the second Hadamard gate as
\begin{align*}
    |\Psi\rangle = &[\sqrt{\lambda} |\psi \rangle + \sqrt{1-\lambda} |\psi_\perp \rangle] \otimes |0\rangle/\sqrt{2}
    +
    [- \sqrt{\lambda} |\psi \rangle + \sqrt{1-\lambda} |\psi_\perp \rangle] \otimes |1\rangle /\sqrt{2}
    \\
    &\downarrow\\
    & \text{Hadamard gate}\\
    &\downarrow\\
    |\Psi'\rangle = & \sqrt{1-\lambda} |\psi_\perp \rangle \otimes |0\rangle
    +
    \sqrt{\lambda} |\psi \rangle
    \otimes |1\rangle.
\end{align*}
As such, upon measuring the ancilla qubit in the $|1\rangle$ state we perfectly obtain $|\psi \rangle$
with probability exactly given by the fidelity $\lambda$.

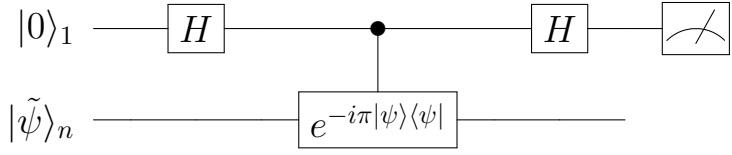
\begin{figure*}[tb]
	\begin{centering}
	\Large
	\[
        \Qcircuit @C=1em @R=1em {
            \lstick{\ket{0}_1}      & \qw  &\gate{H} & \qw  & \ctrl{1}                & \qw &\gate{H}& \qw & \meter \\
            \lstick{\ket{\tilde{\psi}}_n} & \qw  & \qw   & \qw    & \gate{e^{-i\pi \op{\psi}{\psi}}}         & \qw & \qw & \qw \\
        }
    \]
		\caption{
        Circuit for preparing the continuous function through a
        controlled time evolution that reflects the sign of the ideal component
        $|\psi\rangle$ in an erbitrary input state $\ket{\tilde{\psi}}$ (the subscript $n$ in  $\ket{\tilde{\psi}}_n$ in the figure illustrates that the single line represents an $n$-qubit register).
			\label{fig:destructive}
		}
	\end{centering}
\end{figure*}


The approach is most efficent when we can initialise $|\Psi\rangle$ such that
we have a good approximation of $|\psi\rangle$, e.g., by first performing an
adiabatic evolution or by efficiently preparing an approximation to $|\psi\rangle$.
The success probability of the approach is then completely determined
by the quality of the inital state via the fidelity $\mathrm{prob}_1 = \lambda$.

In principle we can also prepare our desired quantum state just by starting
from the all plus state without having an initial approximation to $|\psi \rangle$, i.e., a worst-case approximation.
Let us illustrate this on the example of real functions whereby we can perform the procedure in two steps. First we initialise in the all plus state and prepare
\begin{equation}\label{eqn:initial_overlap}
    f_+(x) := [f(x) \pm |f(x)|  ]/2,
\end{equation}
which has a fidelity at least $\lambda \geq \fillr^2/ (b-a)^2$
in case of the pointwise encoding
via Property~\ref{property:overlap} with the all plus state for a properly chosen sign. We then  prepare $f(x)$ in a second stage which has fidelity at least $1/2$ with respect to $\tilde{f}(x)$. The overall probability of success is then larger then $\fillr^2/ (b-a)^2/2$.
This ensures us that even in the worst case, by choosing the initial state to be the trivial all-plus state we can guarantee a polynomially bounded (in the filling ratio) success probability.

We can straightforwardly bound the error of the simulation approach.
\begin{statement}\label{statement:destructive_error}
Setting
the evolution time as $T=\pi / \lVert A/N \rVert_2$
to simulate the time evolution under the projector as
$e^{-i T A/N} = e^{-i \pi  |\psi \rangle \langle \psi |}$ in $r$ consecutive steps,
the algorithmic error is generally bounded as
$|\epsilon| \leq 
\mathcal{O}(\fillr^p/r)
$
where $p=-2, -4$ depending on the encoding and $\fillr$
is the filling ratio from Definition~\ref{def:peakedness}.
\end{statement}

\begin{proof}
Given the normalisation $c^2 =\lVert A/N \rVert_2 = \mathcal{N}^2/N$
we 
use the low-rank simulation approach to implement the piecewise evolutions $e^{-i \Delta t/c^2 A/N}$ for $\Delta  t = T/r$ and repeat the procedure $r$ times
to simulate the dynamics under
the projector as $e^{-i T/c^2 A/N} = e^{-i T  |\psi \rangle \langle \psi |}$. 
We have already established the simulation error for a small time step $\Delta t =T/r$ in \cref{lemma:simulation_algorithmic_error_bound_asymp} in terms of a unitary operator distance as
\begin{align*}
		\epsilon_0(s) \in O\left(\frac{\mathcal{F}^p}{r^2}\right).
	\end{align*}
When repeating the small time simulation $r$ times
the error accumulates via a triangle inequality at worst as $|\epsilon| \leq r |\epsilon_0| \in
\mathcal{O}(\fillr^p/r)
$.
\end{proof}

\subsection{Application: approximating the normalisation factor}\label{subappendix:estimating_normalization_factor_v2}

If we do not exactly know the normalisation $c^2 =\lVert A/N \rVert_2 = \mathcal{N}^2/N$ we can approximate it the following way. Applying the circuit from Fig.~\ref{fig:destructive} but with a parametric time evolution $e^{-i t A/N}$ where we vary $t$ we obtain the following quantum state immediately before measurement
\begin{equation} \label{eq_destructive_intf}
    |\Psi'\rangle =
    [1 - e^{-i t c^2}] \sqrt{\lambda} |\psi \rangle  \otimes |1\rangle/2
    + \dots.
\end{equation}
The probability of the outcome $|1\rangle$ is 
\begin{equation}
    \mathrm{prob}_1(t) = \lambda |1- e^{-i t c^2}|^2/2
    = \lambda (1-\cos{t c^2}).
\end{equation}
As such, by estimating the probability $\mathrm{prob}_1(t)$ for an increasing $t$ and, e.g., curve fitting or Fourier transforming the time-dependent signal $\mathrm{prob}_1(t)$, allows us to obtain an estimate of $c^2$ which is subject to shot noise and also to the algorithmic error in Statement~\ref{statement:destructive_error}. Note that none of the techniques in the present work need an exact knowledge of the normalisation $c^2$ and a rough estimate is always sufficient. 

In particular, note that even if we do not exactly know $c^2$ we still obtain exactly the desired state in Eq.~\eqref{eq_destructive_intf} after measuring the $\ket{1}$ ancilla outcome which we obtain with a  probability $\mathrm{prob}_1(t)$
that is non-zero
 as long as $ t \neq n 2\pi/c^2$
for any $n\in \mathbb{N}$.
Indeed the maximum of the probability is at $t=\pi/c^2$.

\subsection{Application: efficient state verification \label{sec:verification}}

Given  a state $|\phi \rangle =\sqrt{\lambda} |\psi \rangle + \sqrt{1-\lambda} |\psi_\perp \rangle$
that we have for example prepared via the adiabatic approach in \cref{sec:direct_state_prep}
such that $\lambda \approx 1$, we can use circuit in Fig.~\ref{fig:destructive}
to efficiently verify that we have indeed prepared the desired state and, e.g.,
the hyperparameters of adiabatic evolution are properly chosen: we perform the ancilla
measurement in Fig.~\ref{fig:destructive} and we indeed obtain the quantum state 
$|\psi \rangle$ with probability $\mathrm{prob}_1 = \lambda$ (up to the alogrithmic error of the
time evolution). When the ancilla measurement returns the $\ket{0}$
outcome then we need to restart the procedure. The probability that
we still fail after $m$ trials then decreases exponentially as $(1-\lambda)^m$.
Another significant advantage of this technique is that we can this way remove/project out the
algorithmic error of adiabatic evolution due to finite, discretised time evolution
at the cost of introducing a small probabilty of failure -- non-deterministic approach.

\section{Details of \cref{fig:fig2} \label{sec:numerics}}

In this section we discuss how we numerically simulated our state-preparation approach in  \cref{fig:fig2}.
Note that the low-rank simulation approach we utilise requires two quantum registers.
In our implementation we explicitly represent the two quantum
registers of $n$ qubits each which we use to perform the exact
1-sparse simulation of $S_A$ for a small time step $\Delta t$.
We thus numerically \emph{exactly} simulate the low-rank simulation approach and thus \emph{explicitly}
take into account all error sources of our state preparation approach. 
After each time step we reset the ancilla state by  by first performing Hadamard gates
and then measuring it in the standard basis.

Via the 1-sparse simulation approach we can exactly simulate the dynamics under a matrix $S_A$ whose matrix entries
are given by the function values in \cref{eq:encoding}. Given these function values are computed via an oracle as a binary integer in \cref{def:eff_comp} they are thus prone to a (practically negligible) digitisation error when compared to the continuous functions $\tilde{f}:[a,b]\rightarrow\mathbb{R}$ we consider in \cref{fig:fig2}, e.g., $\alpha e^{-\beta |x - x_0|}$.
We take this digitisation error into account by assuming $d=32$ bits resolution and simulate its effect by computing the digitised function values $f(x_k) = \ceil{ 2^d \tilde{f}(x_k)/\lVert \tilde{f}\rVert_{max} }/2^d  \lVert \tilde{f}\rVert_{max}$ where the symbol $\ceil{\cdot}$ here stands for rounding to the nearest integer.

In our numerical simulations we took explicitly into account the algorithmic error in the low-rank simualtion approach by explicitly simulating the dynamics under $S_A$ over the two registers as
\begin{equation*}
    |\Psi(t)\rangle =
    e^{-i t S_A } 
    |\Psi\rangle =
    e^{-i t S_A } 
     | +^n \rangle |\psi \rangle,
\end{equation*}
where the state of the ancilla register is $| +^n \rangle := |+\rangle^{\otimes n}$ the uniform superposition. We represent the joint state of the two quantum registers via the rank-2 tensor as $\Psi_{kl} := (\langle k | \langle l |) \ket{\Psi}$ with basis states indexed as $k,l \in \{0,1\}^n$.
The action of the matrix $\ket{\Psi^{(1)}}: =  S_A \ket{\Psi}$ can be conveniently computed in terms of the ``entry-wise'' product of the corresponding tensors as
\begin{equation}\label{eq:apply_sa_matrix}
    \Psi^{(1)}_{kl}: =  S_A \ket{\Psi} = ( [A]_{kl}\Psi_{kl} )^T
    =
    [A]_{lk}\Psi_{lk},
\end{equation}
where the transpose operation interchanges the indexes of the rank-2 tensor. We can iteratively compute higher powers of the matrix acting on the vector as $\ket{\Psi^{(2)}} =  S_A \ket{\Psi^{(1)}}$  etc. via the above formula. This allows us to numerically compute the effect of the matrix exponential on the tensor using the Taylor expansion
without explicitly computing the large matrix $S_A$  as
\begin{equation*}
|\Psi(t)\rangle=
   e^{-i t S_A } 
    |\Psi\rangle
=	
	\sum_{k,l \in \{0,1\}^n}
	[
    \sum_{j=0}^m
    \frac{(-i t)^j }{j!}
    \Psi_{kl}^{(j)}
    ]
    |k\rangle |l\rangle,
\end{equation*}
where we compute the tensor in the square brackets up to $m=7$ to minimise numerical error.

In the numerical simulations we computed the joint quantum state $|\Psi( \Delta t) \rangle$ for small time steps $\Delta t$, then performed a Hadamard transformation on the ancilla register $H^{\otimes n} \otimes \mathrm{Id} |\Psi( \Delta t) \rangle$ and then performed a measurement on the ancilla rgister in the standard basis. We can accept any outcome and reset the ancilla register accordingly to $|+^n\rangle$.
However, note that the probability that the ancilla register collapses into $|+^n\rangle$ approaches $1$ for small target $\epsilon$ in \cref{theorem:total_query_complexity} 
and  in \cref{fig:fig2} we indeed consider the practically relevant scenario when the state preparation procedure is very accurate
$\epsilon \ll 1$, i.e., the infidelity is smaller than $10^{-3}$.
As such, in \cref{fig:fig2} (b) at 11 qubits with the Slater-type orbital preparation the probability was above $(1-\mathrm{prob})^r > 0.95$ that all measurements throughout the entire state-preparation procedure result in the outcome $|+^n\rangle$.

\section{Generalisation to multivariate functions\label{sec:multivar}}

Many applications actually require multivariate input states, such as in quantum chemistry~\cite{chan2022grid}, quantum field theory~\cite{jordan2012quantum}, and derivative pricing with multiple underlying assets~\cite{stamatopoulos2020option}. However, there is only limited work exploring the preparation of such input states. In 2008, Kitaev and Webb presented an algorithm for preparing multivariate normal distributions~\cite{kitaev2008wavefunction}, which was expanded upon by a 2021 paper conducting resource estimates and optimizations~\cite{bauer2021practical}.  Finally, some work in preparing chemistry-specific multivariate input states has also been explored~\cite{ward2009preparation}.

Recall that our algorithm uses an oracle that efficiently computes the function values in the sense
of the mapping $O_f\ket{k}\ket{0} = \ket{k}\ket{f(x_k)}$ for a discretised $x$ value with $k \in \{0,1\}^n$.
We can straightforwardly extend the present approach to multivariate functions which we illustrate on the example
of a 2-variate function $f(x,y)$. As such, we aim to prepare a quantum state of $2n$ qubits such that
its amplitudes are proportional to samples of the function.
It is straightforward to show that this is equivalent to using our single-variate approach on $2n$ qubits such that the oracle implements the mapping $O_f\ket{k}\ket{0} = \ket{k}\ket{f(  x_{k_M} , y_{k_L}  )}$ for a discretised $x$ value with $k \in \{0,1\}^{2n}$ and $k_M$ are the $n$ most significant bits of $k$ while $k_L$ are the $n$ least significant bits of $k$. With this definition
of an oracle, our approach then prepares the quantum state
$$\ket{\psi} = (\mathcal{N}_N)^{-1} \sum_{k_M\in \{0,1\}^n} \sum_{k_L\in \{0,1\}^n}    f(x_{k_M} , y_{k_L}) \ket{k_M} \ket{k_L}.$$
The preparation of general $d$-dimensional (i.e. $d$-variate) functions then clearly follows by using our approach
on $dn$ qubits and implementing the oracle $O_f\ket{k}\ket{0} = \ket{k}\ket{f(  x_1(k), x_2(k), \dots x_d(k)   )}$
with $k\in \{0,1\}^{dn}$ and, e.g, $x_1(k)$ is computed from the most significant $n$ bits of $k$.

While above we have related the preparation of multivariate distributions to our single-variate approach, it is important to note that the main limitation of our technique is that our state-preparation error bounds depend on the filling ratio of the function $f(x)$. It is easy to see that the filling ratio might decrease exponentially with an increasing number of variables, e.g., in case of separable functions $f_1(x_1)f_2(x_2) \cdots f_d(x_d)$ the filling ratios are products of the individual filling ratios. Nevertheless, such separable functions can be prepared by $d$ applications of our approach onto $d$ registers each preparing a single-variate function, e.g., $f_1(x)$. The joint quantum state of the $d$ registers is then naturally of a tensor product form.
In typical applications it may already be a significant advantage to be able to prepare low-dimensional functions to high resolution in separate registers such that the joint quantum state is a separable high-dimensional function. An example can be found in preparing 3-dimensional molecular orbitals in real-space quantum chemistry simulations~\cite{chan2022grid}.
The initial product state is only an approximation and the quantum computer is then used to entangle the registers anyway, e.g., by performing a simulation between interactive particles that are initially in a separable state.

\section{Further technical details\label{sec:further_technical}}

\subsection{Grover-Rudolph State Preparation}\label{si:grover_rudolph}
The objective of the Grover-Rudolph state preparation procedure~\cite{grover2002creating} is to prepare the $n$ qubit state,
\begin{align}
    \ket{\psi}_n = \sum_{j = 0}^{2^n - 1}g(x_j)\ket{j}_n,
\end{align}
where $g(x_j)$ is the (normalized) density of the given function $f$ in the $j^{th}$ grid interval. 
In essence, given the correct $k$ qubit state (with $1 \le k < n$), they construct the $k+1$ qubit state, iterating this process until the desired $n$ qubit state is obtained.
They then define the function,
\begin{align}
    h^{(k)}(j) = \frac{\int_{x^{(k)}_j}^{x^{(k)}_j + \frac{1}{2}\Delta_{K}} f(x) dx}{\int_{x^{(k)}_{j}}^{x^{(k)}_{j+1}} f(x) dx},
\end{align}
where $K=2^k$ and $\Delta_{K}$ is the step-size in a $k$ qubit grid.
That is, $h^{(k)}(j)$ represents the conditional probability of $x$ being on the left half of grid interval $j$ on a $k$-qubit grid, given that $x$ lies in the $j^{th}$ interval. Clearly, if $f$ can be integrated efficiently, then $h^{(k)}(j)$ can be computed efficiently and so can $\theta^{(k)}_j\equiv \arccos\left(h^{(k)}(j)\right)$. An ancillary $d$ qubit register is then added, as well as a least-significant $1$ qubit register.
The following state is then obtained by evaluating the oracle performing the mapping $\ket{j}_k\ket{0}_d \mapsto \ket{j}_k\ket{\theta^{(k)}_j}_d$,
\begin{align}\label{eqn:grover_intermediate_step}
    \sum_{j = 0}^{2^k - 1}g(x^{(k)}_j)\ket{j}_k\ket{\theta^{(k)}_j}_d\ket{0}_1.
\end{align}
It is important to note that this step requires integration to be performed in superposition. As observed in the Grover paper, any log-concave distribution can be efficiently integrated using Monte Carlo integration, and can thus be integrated efficiently on a quantum computer.
Define a CRy gate (with a definition analogous to the CRx gate shown in Figure~\ref{fig:crx_gate_definition}) such that it acts on the last two registers and performs the mapping $\ket{a}_d\ket{0}_1 \mapsto \ket{a}_d\left(\cos a\ket{0}_1 + \sin a\ket{1}_1\right)$. Note that such a gate is derived in Section~\ref{section:hamiltonian_simulation_techniques} (in the general setting, in this case $t=1$). The CRy gate is then applied to the last two registers in Equation~\ref{eqn:grover_intermediate_step}, yielding the state
\begin{align}
    \sum_{j = 0}^{2^k - 1}g(x^{(k)}_j)\ket{j}_k\ket{\theta^{(k)}_j}_d\left(
    \cos \theta^{(k)}_j \ket{0}_1 + \sin \theta^{(k)}_j \ket{1}_1
    \right)
    =
    \sum_{j = 0}^{2^k - 1}g(x^{(k)}_j)\ket{j}_k\ket{\theta^{(k)}_j}_d\left(
    h^{(k)}(j) \ket{0}_1 + (1-h^{(k)}(j)) \ket{1}_1
    \right).
\end{align}
We then uncompute the $d$ qubit ancillary register (by applying the inverse of the integration circuit -- noting that if the integration is classically stochastic, it must be seeded) and discard it, yielding the state
\begin{align}
    \sum_{j = 0}^{2^k - 1}g(x^{(k)}_j)\ket{j}_k\left(
    h^{(k)}(j) \ket{0}_1 + (1-h^{(k)}(j)) \ket{1}_1
    \right)
    =
    \sum_{j = 0}^{2^{k+1} - 1}g(x^{(k+1)}_j)\ket{j}_{k+1}
    =
    \ket{\psi}_{k+1}.
\end{align}
The procedure is then repeated until the desired $n$ qubit state is obtained.

\subsection{Asymptotic properties of the encoded quantum states \label{sec:further_asymptotic}}

Recall that in this work we refer to a continuous function $f: [a,b] \mapsto \mathbb{C}$ on the \emph{closed interval}
$a \leq x \leq b$ as $f(x)$, and will now briefly summarize some notion related to continuous functions which we will frequently utilize in the following proofs. First, recall the general definition of a continuous function $f \in \mathcal{C}$  as $\lim_{c \rightarrow x} f(c) = f(x)$
and due to the boundedness theorem all such functions are bounded as $|f(x)|\leq \lVert f \rVert_{max}$.
Second, we define the subset of Lipschitz continuous functions as $f \in \mathcal{C}_{L} \subseteq \mathcal{C}$, which are slightly more restrictive as they satisfy $|f(x_1) - f(x_{2})| \leq K |x_1 - x_{2}|$ (for a real constant $K\ge 0$), ensuring that they asymptotically approach any value of the function through discretization in \emph{order one} $\mathcal{O} (N^{-1})$.
Third, given all continuous functions are Riemann integrable, it will be natural for us to approximate the finite sum $\sum_k f(x_k) \Delta_N$ in terms of the integral $\int_a^{b}f(x) \, \mathrm{d} x$, for which we incur an asymptotically vanishing approximation error 
$\lim_{N\rightarrow\infty} \epsilon =0$ while for Lipschitz continuous functions we can additionally guarantee the scaling
 $\epsilon \in \mathcal{O} (N^{-1})$  exponential in the number of qubits as we prove now.

\begin{lemma}[Riemann sums]\label{lemma:riemann_sum}
Given an arbitrary continuous function $f(x): [a,b] \mapsto \mathbb{C}$ on the interval
$a \leq x \leq b$, the finite Riemann sum approximates the integral as
\begin{equation*}
    \int_a^b f(x) \,\mathrm{d}x =  \sum_{k=0}^{N-1} f(x_k) \Delta_N + \epsilon
\end{equation*}
up to an asymptotically vanishing error $\lim_{N\rightarrow\infty} \epsilon =0$.
For functions of bounded variation the error term additionally satisfies $\epsilon\in \mathcal{O}(N^{-1})$, see refs.~\cite{owens2014exploring,tasaki2009convergence} which also applies to Lipschitz continuous functions given
they are of bounded variation.
\end{lemma}
\begin{proof}
Given all continuous functions are Riemann integrable
we can approximate their Riemann integral via the convergent
(for $N \rightarrow \infty$) 
sum as
\begin{equation*}
\int_a^b f(x) \,\mathrm{d}x 
=
\sum_{k=0}^{N-1} f(x_k) \Delta_N + \epsilon,
\end{equation*}
where the error term satisfies $\lim_{N\rightarrow\infty} \epsilon =0$ for all Riemann integrable functions.
Here we denote the discretisation interval length as $\Delta_N = (b-a)/N$ and the asymptotic scaling of the approximation error
can be established $\epsilon \leq \mathrm{const} \, \Delta_N$~
\cite{
owens2014exploring,tasaki2009convergence} for a broad class of functions, where the constant factor depends on properties of the functions. For example, if $f$ is  differentiable then it is proportional to the largest value of the derivative function withing the interval $a \leq x \leq b$. Of course, $f$ need not be differentiable for the above error bound to hold and, e.g., ref~\cite{owens2014exploring} established the bound $|\epsilon| \leq V (b-a)/N$ using the absolute largest variation $V$ of the function $f$. Recall that functions of bounded variation $V$ contain all Lipscitz continuous functions.
\end{proof}

\begin{property}\label{prop:encodig_equiv}
Given an arbitrary continuous, non-negative function $f$, the quantum state $|\psi\rangle$ in Eq.~\ref{eq:encoding} obtained via the integral encoding $g(x)$ becomes identical to the pointwise encoding of the corresponding function $\sqrt{ f(x) }$ for large $N$.
\end{property}
\begin{proof}
Recall that all continuous functions are Riemann integrable and thus via
Lemma~\ref{lemma:riemann_sum} we can approximate the Riemann integral as 
\begin{align}
    g(x_j) = [\int_{x_j}^{x_j + \step}f(x)dx]^{1/2} =
     [\Delta_N 
     f(x_j) + \epsilon]^{1/2},
\end{align}
and the error generally satisfies $\lim_{N \rightarrow \infty} \epsilon = 0$
and for Lipschitz continuous functions we can also guarantee the scaling $\epsilon \in \mathcal{O}(N^{-1})$. The encoded quantum state thus satisfies
\begin{align}
    \ket{\psi} = \sum_{j\in \{0,1\}^n}    
    g(x_j)
    \ket{j} = \sqrt{\Delta_N} \sum_{j\in \{0,1\}^n}\sqrt{f(x_j)}\ket{j}
    + \epsilon,
\end{align}
where $\sqrt{\Delta_N}$ approaches the normalisation constant $(\mathcal{N}_N)^{-1}$ from a pointwise encoding of $\sqrt{f}$ -- and thus indeed the two encodings become identical for large $N$.
Note that we satisfy our normalisation condition from Definition~\ref{def:integrable_fct} that
$\lVert \psi \rVert^2 = 1 \approx \sum_{j\in \{0,1\}^n} f(x_j) \step$.
\end{proof}

\begin{property}[pointwise encoded amplitudes]\label{prop:pointwise_amplitude}
Given an arbitrary continuous function, the pointwise encoded amplitudes satisfy the general bound $N |\psi_k|^2 
\leq
(b-a) \lVert f \rVert_{max}^2 / \lVert f \rVert_2^2 + \epsilon$.
The error term asymptotically vanishes $\lim_{N\rightarrow\infty} \epsilon =0$
and for Lipschitz continuous function it additionally satisfies $\epsilon\in \mathcal{O}(N^{-1})$
\end{property}
\begin{proof}

Note that at a level of discretisation $N$ the amplitudes are specified as $|\psi_k|^2 = |f(x_k)|^2/ \mathcal{N}^2 $, where the normalisation $\mathcal{N}^2 = \sum_{k=0}^{N-1} |f(x_k)|^2$ ensures that our quantum state has unit norm. Let us now upper bound the amplitudes as
\begin{equation*}
N |\psi_k|^2 = 
\frac{N |f(x_k)|^2 }{ \mathcal{N}_N^2 }
=
  \frac{(b-a) |f(x_k)|^2 }{ \mathcal{N}_N^2 \Delta_N }
\end{equation*}

It is a direct consequence of Lemma~\ref{lemma:riemann_sum} that the normalisation condition
$\mathcal{N}^2 \Delta_N   = \lVert f \rVert_2 ^2 + \epsilon$ can be approximated via the Riemann integral of $|f(x)|^2$.
Substituting this back we obtain the formula
\begin{equation*}
N |\psi_k|^2 =   \frac{(b-a) |f(x_k)|^2 }{ \lVert f \rVert_2 ^2 + \epsilon }
\leq 
\frac{(b-a) \lVert f \rVert_{max}^2 }{ \lVert f \rVert_2^2 + \epsilon }\\
=
\frac{(b-a) \lVert f \rVert_{max}^2 }{ \lVert f \rVert_2^2}
+ \epsilon.
\end{equation*}
Above we have first used the defintion of the infinity norm $\lVert f \rVert_{max} := \max_x |f(x)|$
and then we have expanded the fraction into a series in $\epsilon = \mathcal{O}(N^{-1})$. 

The above upper bound is actually saturated given $f$ is Lipschitz continuous as $\max_k |f(x_k)|^2 = \lVert f \rVert_\infty ^2 + \mathcal{O}(N^{-1})$. As such, we can establish that the largest amplitude converges in exponential order to the expression that is independent of the resolution $N$
\begin{equation*}
    N \max_k  |\psi_k|^2 = \frac{(b-a) \lVert f \rVert_{max}^2 }{ \lVert f \rVert_2^2}
    + \epsilon.
\end{equation*}
The error term asymptotically vanishes $\lim_{N\rightarrow\infty} \epsilon =0$
and for Lipschitz continuous function it additionally satisfies $\epsilon\in \mathcal{O}(N^{-1})$ via  Lemma~\ref{lemma:riemann_sum}.

\end{proof}

\begin{property}[integral encoded amplitudes]\label{prop:integral_amplitude}
Given an arbitrary continuous function, the integral encoded amplitudes satisfy the general bound $N |\psi_k|^2 
\leq
(b-a) \lVert f  \rVert_{max} +  \epsilon$.
The error term asymptotically vanishes $\lim_{N\rightarrow\infty} \epsilon =0$
and for Lipschitz continuous function it additionally satisfies $\epsilon\in \mathcal{O}(N^{-1})$.
\end{property}

\begin{proof}
Let us approximate the probabilities via the Riemann summation where we only use $1$ summand as
\begin{equation*}
    |\psi_k|^2 = g^2(x_k) = \int_{x_k}^{x_k + \Delta_N}f(x) \, \mathrm{d}x 
    =
    f(x_k) \Delta_N + \epsilon
\end{equation*}
via \cref{lemma:riemann_sum}.
As such we can establish the general bound
\begin{equation*}
    \max_k g(x_k) = \max_k |\psi_k|^2  = \lVert f  \rVert_{max} \Delta_N +  \epsilon,
\end{equation*}
Therefore we conclude that
\begin{equation*}
 N \max_k g(x_k) =   N \max_k |\psi_k|^2 = (b-a) \lVert f  \rVert_{max} +  \epsilon.
\end{equation*}
The error term asymptotically vanishes $\lim_{N\rightarrow\infty} \epsilon =0$
and for Lipschitz continuous function it additionally satisfies $\epsilon\in \mathcal{O}(N^{-1})$
via \cref{lemma:riemann_sum}.
\end{proof}

\begin{property}[pointwise overlap with uniform states]\label{property:overlap}
For an arbitrary continuous function $f(x): [a,b] \mapsto \mathbb{R}$ on the interval
$a \leq x \leq b$ we construct the non-negative function $\tilde{f}(x) := [f(x) \pm |f(x)|]/2$ which has overlap at least $1/2$ for a properly chosen sign with our function $f$ and
we encode its function values into a quantum state via our pointwise encoding from \cref{eq:encoding}. For any such encoding we can establish the asymptotically bounded fidelity with the easy-to-prepare state $\ket{+}^{\otimes n}$ as $|\langle \psi | (| + \rangle^{\otimes n})|^2 \geq \fillr^2/(b-a)^2 + \epsilon$
in terms of our filling ratio $\fillr$.
The error term asymptotically vanishes $\lim_{N\rightarrow\infty} \epsilon =0$
and for Lipschitz continuous function it additionally satisfies $\epsilon\in \mathcal{O}(N^{-1})$
\end{property}
\begin{proof}
Recall that the normalisation constant form \cref{eq:encoding} satisfies
\begin{equation*}
    N^{-1} (\mathcal{N}_N)^2 = \sum_{k} \tilde{f}(x_k)^2 N^{-1} = \lVert \tilde{f} \rVert_2^2/(b-a) + \epsilon,
\end{equation*}
via \cref{lemma:riemann_sum}.
Similarly we can obtain the one-norm of our non-negative function $\tilde{f}$
\begin{equation*}
   N^{-1}  \sum_{k} \tilde{f}(x_k) =
   N^{-1}  \sum_{k} |\tilde{f}(x_k) |
   =
    \lVert \tilde{f} \rVert_1/(b-a)+ \epsilon.
\end{equation*}
The error term asymptotically vanishes $\lim_{N\rightarrow\infty} \epsilon =0$
and for Lipschitz continuous function it additionally satisfies $\epsilon\in \mathcal{O}(N^{-1})$
via \cref{lemma:riemann_sum}.

Let us now expand the fidelity  $fid =|\langle \psi | (| + \rangle^{\otimes n})|^2$ as
\begin{align*}
    fid = \frac{ N^{-1} [\sum_{k} \tilde{f}(x_k) ]^2 } { (\mathcal{N}_N)^2 }
    &=
    \frac{
    N^{-2} [\sum_{k} \tilde{f}(x_k) ]^2
    }{
    N^{-1} (\mathcal{N}_N)^2
    }
    =
    \frac{
     [N^{-1} \sum_{k} \tilde{f}(x_k) ]^2
    }{
    N^{-1} (\mathcal{N}_N)^2
    }.
\end{align*}
We can now substitute back our approximations to the function norms as
\begin{equation*}
    fid
    =
    \frac{
     [\lVert \tilde{f} \rVert_1/(b-a) ]^2
    }{
    \lVert \tilde{f} \rVert_2^2/(b-a) 
    }
    + \epsilon
    =
    \frac{
     \lVert \tilde{f} \rVert_1^2
    }{
    (b-a) \lVert \tilde{f} \rVert_2^2 
    }
    + \epsilon
\end{equation*}
Let us now lower bound this quantity using that generally $\lVert \tilde{f} \rVert_2^2 \leq (b-a) \lVert \tilde{f} \rVert_{max}^2$ as
\begin{equation*}
    fid
   \geq 
   \frac{
     \lVert \tilde{f} \rVert_1^2
    }{
    (b-a)^2 \lVert \tilde{f} \rVert_{max}^2
    }
    + \epsilon
    =
   \fillr^2/(b-a)^2
       + \epsilon
\end{equation*}
and
we have also substituted back our filling ratio $\fillr$.

\end{proof}

     The above Lemma guarantees that given any function $f$ and a sufficiently large resolution $N$, we can start with the all plus state and efficiently prepare $\tilde{f}$ with the non-deterministic techniques in \cref{section:qpe_state_prep} and \cref{section:prep_via_destructive_interference}. The success probability depends on the fidelity which we established above depends polynomially on our filling ratio $\fillr$. 
    In a second step from $\tilde{f}$ we can prepare $f$ efficiently with a success rate at least $1/2$.
    The generalisation to complex functions is straightforward.

\subsection{Asymptotic properties of  the rank-one matrix $A$}

\begin{lemma}[pointwise encoding]\label{lemma:amatr_pointwise}
Given our pointwise encoding from Eq.~\eqref{eq:encoding} of an arbitrary continuous function $f$
and the corresponding matrix entries $A_{kl}:= f(x_k) f^*(x_l)$ form \cref{eq:param_hamil_def},
the matrix norm is constant bounded $\lVert A \rVert_{max} \leq \lVert f \rVert_{max}^2$.
The only non-zero eigenvalue of $A/N$ is given by the spectral norm as  $\lVert A/N \rVert_{2}
= \mathcal{N}_N^2 /N
        = (b-a) \lVert f \rVert_2^2 + \epsilon$.
Here $\lVert f \rVert_2$ and $\lVert f \rVert_{max}$ are the square integral ($L^2$) and absolute maximum value norms of the continuous function, respectively.
Furthermore, the ratio of matrix norms is generally upper bounded as $      \lVert A \rVert_{max} / \lVert A/N \rVert_{2} 
    \leq 
    \fillr^{-2}
    +\epsilon$
    by our filling ratio from Defition~\ref{def:peakedness}.
    The error term asymptotically vanishes $\lim_{N\rightarrow\infty} \epsilon =0$
    and for Lipschitz continuous functions it additionally satisfies $\epsilon\in \mathcal{O}(N^{-1})$.
    See Fig.~\ref{fig:matrix_norms} for a numerical verification of these results.
\end{lemma}
\begin{proof}

Let us now define the rank-1 matrix $A_{kl}:= f(x_k) f^*(x_l)$. We can bound the largest entry in the matrix as
\begin{equation*}
\lVert A \rVert_{max} =
 \max_{kl} |A_{kl}| = 
 \max_{kl} |f(x_k) f^*(x_l)|  
=  \max_{k} |f(x_k)|^2
\leq \lVert f \rVert_{max}^2.
\end{equation*}
The upper bound is actually saturated asymptotically for Lipschitz continuous functions with the error scaling $\lVert A \rVert_{max} = \lVert f \rVert_{max}^2 + \mathcal{O}(N^{-1})$.
Let us now bound the matrix norm using that $\lVert A \rVert_{2}^2
= 
\lVert A \rVert_{HS}^2$
for our rank-1 matrix $A$ as
\begin{equation*}
\lVert A \rVert_{2}^2
= 
\lVert A \rVert_{HS}^2
=
\sum_{k,l} |f(x_k) f(x_l)^*|^2
=
[ \sum_{k} |f(x_k)|^2]^2.
\end{equation*}
We can approximate this expression via the Riemann sum from Lemma~\ref{lemma:riemann_sum} as
\begin{equation}\label{eq:exact_eigenvalue}
    \lVert A/N \rVert_{2}
    = \sum_{k} |f(x_k)|^2/N
    = (b-a) \lVert f \rVert_2^2 + \epsilon,
\end{equation}
where the error term asymptotically vanishes $\lim_{N\rightarrow\infty} \epsilon =0$
and for Lipschitz continuous function it additionally satisfies $\epsilon\in \mathcal{O}(N^{-1})$.
Finally, we recollect all terms and establish the ratio of matrix norms as
\begin{equation*}
     \frac{  \lVert A \rVert_{max}}{ \lVert A/N \rVert_{2} }
     =
     \frac{ \max_{k} |f(x_k)|^2 }{ \sum_{k} |f(x_k)|^2/N }
    \leq 
    \frac{\lVert f \rVert_{max}^2}
    {(b-a) \lVert f \rVert_2^2 }
    + \epsilon
    \leq 
    \fillr^{-2}
    +\epsilon,
\end{equation*}
where we have used the general function-norm inequality
$\lVert f \rVert_2^2 (b-a) \geq \lVert f \rVert_1^2$.
We have also substituted back our filling ratio $\fillr$ from 
Definition~\ref{def:peakedness}.
\end{proof}

\begin{lemma}[integral encoding]\label{lemma:amatrix_integral}
Given an arbitrary continuous efficiently integrable function $f$ in Definition~\ref{def:integrable_fct} and our
corresponding rank-1 matrix  $A_{kl} := N g(x_k) g^*(x_l)$ from \cref{eq:param_hamil_def} where $g(x)$ are non-negative, piecewise integrals of $f$ via our integral in \cref{eq:encoding},
the matrix satisfies $\lVert A/N \rVert_2 = \lVert f \rVert_1$ and its matrix entries are asymptotically constant bounded as $\lVert A \rVert_{max}
=  (b-a) \lVert f  \rVert_{max} +  \epsilon$. Here we have used the $L^1$ and absolute maximum norms of the function. Furthermore, the ratio of the matrix norms is generally upper bounded as $
      \lVert A \rVert_{max}
      /\lVert A/N \rVert_{2} 
    \leq
    (b-a)\fillr^{-1}  + \epsilon
$ by our filling ratio from Defition~\ref{def:peakedness}.
   The error term asymptotically vanishes $\lim_{N\rightarrow\infty} \epsilon =0$
and for Lipschitz continuous function it additionally satisfies $\epsilon\in \mathcal{O}(N^{-1})$.
See Fig.~\ref{fig:matrix_norms} for a numerical verification of these results.
\end{lemma}
\begin{proof}

Let us first define the rank-$1$ matrix with non-negative entries as
\begin{equation*}
    A_{kl} := 
     N g(x_k) g(x_l)  .
\end{equation*}
We can bound the largest entry in $A$ as
\begin{equation}
\lVert A \rVert_{max} = \max_{kl}  | A_{kl}|
= N  \max_{k}  g^2(x_k)
= N  \max_{k}  f(x_k) \Delta_N
= (b-a) \lVert f  \rVert_{max} +  \epsilon,
\end{equation}
where we have used from Property~\ref{prop:encodig_equiv} that $g(x_k) = \sqrt{f(x_k) \Delta_N} +  \epsilon$
and recall that $\Delta_N = (b-a)/N$.
The error term asymptotically vanishes $\lim_{N\rightarrow\infty} \epsilon =0$
and for Lipschitz continuous function it additionally satisfies $\epsilon\in \mathcal{O}(N^{-1})$
via \cref{lemma:riemann_sum}.

Given the matrix $A$ is rank-1 by definition, it only has 1 non-zero eigenvalue which is equivalent to the norm $\lVert A \rVert_2$ which we can compute via the (equivalent) Hilbert-Schmidt norm as
\begin{equation*}
    \lVert A \rVert_{HS}^2
    =
    N^2 \sum_{k,l} [g(x_k) g(x_l)]^2
    =[N \sum_{k} g^2(x_k)]^2.
\end{equation*}
As such, we exactly obtain the largest eigenvalue of the matrix as
\begin{equation}\label{eqn:integral_encoding_a_over_n_norm}
    \lVert A/N \rVert_2
    =
     \sum_{k} g^2(x_k)
    =
     \sum_{k} \int_{x_k}^{x_k + \Delta_N}f(x) \, \mathrm{d} x 
    = \int_{a}^{b} f(x) \, \mathrm{d} x 
    =  \lVert f \rVert_1,
\end{equation}
where in the last equation we have used that $f$ is non-negative and therefore  its integral is equivalent to is $L^1$ norm.
Finally, we recollect all terms and establish the ratio of matrix norms as
\begin{equation*}
    \frac{  \lVert A \rVert_{max}}{ \lVert A/N \rVert_{2} }
    =
    \frac{ N  \max_{k}  g^2(x_k)   }{ \sum_{k} g^2(x_k) }
    \leq
   \frac{
    (b-a) \lVert f  \rVert_{max}
   }{\lVert f \rVert_1}
    + \epsilon
    =
    (b-a)\fillr^{-1}  + \epsilon
      ,
\end{equation*}
where we substituted back our definition of the filling ratio $\fillr$
from Definition~\ref{def:peakedness}.
\end{proof}

\end{document}